\documentclass{article}

\usepackage{arxiv}

\usepackage[english]{babel}
\usepackage[utf8]{inputenc}
\usepackage[pdftex]{graphicx}
\usepackage[table]{xcolor} 
\usepackage{amsthm,amsmath,amssymb,natbib}
\RequirePackage[colorlinks,citecolor=blue,urlcolor=blue]{hyperref}

\usepackage{pgf,tikz}
\usetikzlibrary{shapes,snakes}
\usetikzlibrary{arrows}

\usepackage{tabularx}
\usepackage{algorithmic}
\usepackage[subsection]{algorithm}
\usepackage{lscape}

\usepackage{caption}
\usepackage{subcaption}

\RequirePackage{hypernat}

\DeclareMathOperator*{\argmax}{argmax}

\newcommand{\proKEP}{KEP}
\newcommand{\eg}{\emph{e.g.}}
\newcommand{\ie}{\emph{i.e.}}
\newcommand{\IndA}{IA}
\newcommand{\NE}{NE}
\newcommand{\SWE}{SWE}
\newcommand{\FVS}{{\sc Feedback vertex set}}
\newcommand{\threeDIM}{{\sc 3-Dimensional matching}}
\newcommand{\BestReponsepes}{{\sc  Pessimistic Best Response}}
\newcommand{\BestReponseopt}{{\sc  Optimistic Best Response}}
\newcommand{\proKEGcard}{\#-KEG}
\newcommand{\proKEGwei}{W-KEG}

\newtheorem{axiom}{Axiom}[section]
\newtheorem{claim}[axiom]{Claim}

\newtheorem{example}[axiom]{Example}
\newtheorem{lemma}[axiom]{Lemma}

\newtheorem{theorem}[axiom]{Theorem}

\newcolumntype{H}{>{\setbox0=\hbox\bgroup}c<{\egroup}@{}}

\title{Game theoretical analysis of  Kidney Exchange Programs
\thanks{Parts of this material are based on data and information provided by Canadian Blood Services. However, the analyses, conclusions, opinions and statements expressed herein are those of the authors and not necessarily those of Canadian Blood Services}}

\author{
  Margarida Carvalho
  \\
  CIRRELT and D\'epartement d'Informatique et de Recherche Op\'erationnelle\\
   Universit\'e de Montr\'eal\\
  \texttt{carvalho@iro.umontreal.ca} \\
  \And
Andrea Lodi \\
 Canada Excellence Research Chair in Data Science for Real-Time Decision-Making\\
 Polytechnique de Montr\'eal\\
 \texttt{andrea.lodi@polymtl.ca} \\
}

\begin{document}

\maketitle

\begin{abstract}

The goal of a kidney exchange program (KEP) is to maximize number of transplants within a pool of incompatible patient-donor pairs by exchanging donors. A KEP can be modelled as a maximum matching problem in a graph. A KEP between incompatible patient-donor from pools of several hospitals, regions or countries has the potential to increase the number of transplants. These entities aim is to maximize the transplant benefit for their patients, which can lead to strategic behaviours. Recently, this was formulated as a non-cooperative two-player game and the game solutions (equilibria) were characterized  when the entities objective function is the number of their patients receiving a kidney. In this paper, we generalize these results for $N$-players and discuss the impact in the game solutions when transplant information quality is introduced. Furthermore, the game theory model is analyzed through computational experiments on instances generated through the Canada Kidney Paired Donation Program. These experiments highlighting the importance of using the concept of Nash equilibrium, as well as, the anticipation of the necessity to further research for supporting  police makers once measures on transplant quality are available. 
\end{abstract}


\paragraph{Kidney exchange programs ({\proKEP}s)} A patient suffering from chronic kidney disease can see their life quality significantly improved by receiving a transplant of a compatible kidney. Typically, patients are registered in a waiting list and whenever a deceased donor arrives an allocation algorithm is run to determine to which patient in the waiting list the graft will be offered. The allocation algorithms are designed in a way that reflects utilitarian and equitable objectives. In addition, a patient can also receive a kidney transplantation from a  healthy relative or a friend willing to make the donation.

Unfortunately, the allocation of deceased donors is far from meeting the demand for kidney transplant and living donors may not be able to donate to their loved one due to physiologically incompatibility. Around one in a thousand European citizens suffer from end-stage renal disease \citep{ref_Europe}. In United States and Canada, it is the 9th and 11th, respectively, leading cause of death each year \citep{ref_2_CDC,ref_3_KFC}. In particular, in 2017, 78\% of the Canadians on the waiting list for an organ were waiting for a kidney \citep{ref_3_KFC}. This disease has a high economic impact on national health services \citep{ref_2_NKP,ref_3_KFC}. Moreover, the number of Canadians living with end-stage kidney disease has grown \citep{ref_3_KFC}.

In an attempt to tackle the organ shortage, the idea of living-donor exchanges was proposed~\cite{Rapaport1986}. In this paradigm, two incompatible pairs can swap their donors if the donor in one pair is compatible with the patient of the other pair and vice-versa. This leads to 2-way exchanges, whose concept can be extended to $L$-way exchanges with $L\geq 2$; see Figure~\ref{fig:2way3way}. In 2004, \citet{Roth_Sonme_Unver_2004} proposed the first mechanism deciding kidney exchanges in a pool of incompatible pairs.\footnote{The authors also considered \emph{indirect exchanges}, where a donor from an incompatible pair donates to the waiting list in return for a priority on it.} It considered patients preferences over donors and no limit on $L$. Since the transplantions of an $L$-exchange must take place simultaneously to avoid donors withdraw and, for living donation, the patients preference can be reduced to compatible donors, bounds on $L$ are imperative and 0-1 preferences are enough. Hence, \citet{Roth_Sonme_Unver_2005_b} focused on mechanisms for $2$-way exchanges maximizing the number of transplants. Their mathematical setup for the graph representation of a KEP is used in this work.

Motivated by potential welfare gains~\cite{Roth_Sonmez_Unver_c}, kidney exchange programs started to be established.  The success of early {\proKEP}s in South Korea~\cite{Park2004} and the Netherlands~\cite{Klerk2005} was reported, and, nowadays, they are run in many countries; \eg, see \cite{FirstHandbook,biro_building_2018} for a recent report on current practices in Europe, and \cite{CanadianKEP} for the Canadian program. There is an extensive literature formulating the problem of selecting the optimal set of exchanges in a {\proKEP} with integer programming models and developing techniques to solve them efficiently \citep{Roth_Sonmez_Unver_2007,Abraham2007,Manlove2012,Constantino2013,Glorie2014,Dickerson:2016}. As a rule of thumb,  KEPs are centralized, that is, there is one decision maker deciding the exchanges to be executed. In general, the main goal is to maximize the number of patients receiving a donation, but other criteria have also been considered.

\begin{figure}
	\centering
	\begin{tabular}{cc}
		2-way exchange & $L$-way exchange \\ 
		\begin{tikzpicture}[-,>=stealth',shorten >=0.3pt,auto,node distance=3cm,
			thick,countryA node/.style={circle,draw},countryB node/.style={diamond,draw,fill=gray}, scale=.85, transform shape]
			\tikzstyle{matched} = [draw,line width=3pt,-]
			
			\node[countryA node] (1) {\begin{tabular}{cc}
					\tiny Patient  \\ 
					\tiny Donor
			\end{tabular}};
			\node[countryA node] (2) [right of=1] {\begin{tabular}{cc}
					\tiny Donor  \\ 
					\tiny Patient
			\end{tabular}};
			
			\path[->,every node/.style={font=\sffamily\small}]
			(1) edge[bend right] node  {} (2)
			(2) edge[bend right] node  {} (1);
		\end{tikzpicture}
		& 
		\begin{tikzpicture}[-,>=stealth',shorten >=0.3pt,auto,node distance=2cm,
			thick,countryA node/.style={circle,draw},countryB node/.style={draw=white}, scale=.85, transform shape]
			\tikzstyle{matched} = [draw,line width=3pt,-]
			
			\node[countryA node] (1) {$1$};
			\node[countryA node] (2) [right of=1] {$2 $};
			\node[countryB node] (3) [right of=2] {$ \ldots$};
			\node[countryB node] (4) [right of=3] {$ $};
			\node[countryA node] (5) [right of=4] {$L$};
			
			\path[->,every node/.style={font=\sffamily\small}]
			(1) edge[bend right] node  {} (2)
			(2) edge[bend right] node  {} (3)
			(4) edge[bend right] node  {} (5)
			(5) edge[bend right] node  {} (1);
		\end{tikzpicture}
	\end{tabular}
	\caption{Vertices represent  incompatible patient-donor pairs and arcs represent compatibilities.}
	\label{fig:2way3way}
\end{figure}
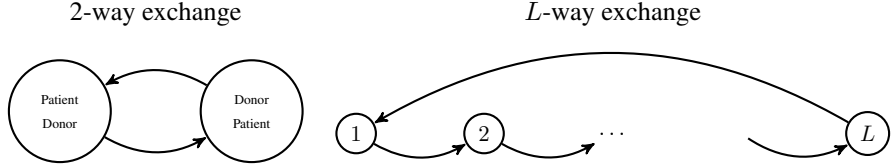

\paragraph{Multi-agent programs}
In the U.S., there are cross hospital  {\proKEP} platforms such as NKR, UNOS and APD. An international exchange between Czech Republic and Austria has taken place \citep{ref_4_Chez}. In August 2018, another international exchange was performed between a pair from Spain and another from Italy~\citep{ref_7_Italy_Spain}.
A cross continent exchange between Israel and Czech Republic leading to 6 transplants took place in December 2019~\citep{ref_8_Israel_Czech}. Currently, the hypothesis of an European Program~\citep{ref_ENCKEP} is being investigated with already some cross-border programs in place, \eg, the Scandiatransplant Kidney Paired Exchange Program 
and the South Alliance for Transplant.  In September 2019, it was announced a joint kidney exchange program between Australia and New Zealand~\cite{ref_Australia_NZ}.  Therefore, there is a recent motivation to study {\proKEP}s that can have incompatible pairs belonging to different agents (which can be hospitals, regions or countries) in an attempt to reach a higher social benefit in comparison to the individual case. Under these circumstances, the optimization cannot be anymore simply the maximization of the overall patients' benefit: it is natural to expect that each agent behaves strategically with the goal of maximizing the benefit of her patients receiving a transplant.  Indeed,  \citet{Ashlagi2017} show that  the centralized platforms UNOS and NKR do not register a high level of exchanges across hospitals, in part due to a market mechanism that fails to incentivize hospitals to submit all their pairs.\footnote{The other reason for the low number of cross-hospital exchanges is associated with the costs of participating in these platforms.} In this context, a  game theory setting can suitably frame the agents' incentives. In fact, the literature in this context can be divided on centralized and decentralized non-cooperative game models, and cooperative approaches.

\paragraph{Centralized (non-cooperative) multi-agent {\proKEP}s} The centralized programs focus on mechanisms guaranteeing individual rationality (IR) and strategy-proofness (SP) when each player's utility is the number of her patients receiving a kidney. While an IR mechanism ensures that each agent matches an equal or greater number of her patients than individually, an SP mechanism makes dominant for each agent the strategy of truthfully revealing her incompatible patient-donor pairs. In order to achieve IR or SP,  the social optimum (maximum number of exchanges when there is full collaboration) might have to be sacrificed~\cite{Roth_Sonmez_Unver_2005_a,Sonmez_Unver_2011}. Table~\ref{Table:strategyproof} provides the current landscape of known approximation ratios on the social optimum for these mechanisms. The values in bold correspond to the mechanisms that are tight, \ie, the mechanisms that match a known upper bound. 
	Alternatively to the analysis of the worst-case graph for such mechanisms, Bayesian models have been described. For example, \citet{Blum2017} achieve positive results by  considering deterministic graphs but  randomly associating vertices to players; under certain conditions, it is proven that a social optimum is likely to be individually rational. The authors also provide a simple mechanism that is individually rational and  likely to be a social optimum. Bayesian models reproducing the topology of exchange pools have also been proposed. Here, a random graph is generated based on blood type and tissue type compatibility distributions. In~\cite{Ashlag:2011,Ashlagi2014}, it is shown under mild conditions that, in almost all large graphs, there is an IR allocation using exchanges of size at most 3 close to the social optimum. The authors also devise a bonus mechanism such that there is an almost social optimum Nash equilibrium where agents truthfully report their pairs. ~\citet{Toulis2011,Toulis2015} show a similar results for $L=2$ with the main difference that a fixed number of hospitals has a very large pool, while in~\cite{Ashlag:2011,Ashlagi2014} there is a growing number of small hospitals.

While the previous papers study the static case, \citet{Dickerson2015} considers a multi-period dynamic model where a  credit mechanism is build, proven to be strategy-proof  and to lead to a global maximum matching; that approach does not rely on the structure of the compatibility graph, it assumes that each patient-donor pair remains only one period in the pool and \emph{a priori} knowledge of the  average arrival rate of pairs to each hospital.

\begin{table}[ptbh]
	\tiny
		\begin{tabular}{rrr|c|cc|cc|cc|c|c}
			& Setup &           &\multicolumn{9}{c}{Literature} \\ 
			\hline
			&        &          & \multicolumn{1}{c}{\cite{Roth_Sonmez_Unver_2005_a,Sonmez_Unver_2011}} 
			& \multicolumn{2}{c}{\cite{Ashlagi2015}}& \multicolumn{2}{c}{\cite{Caragiannis2015}}&
			\multicolumn{2}{c}{ \cite{Ashlagi2014,Ashlag:2011} }  & \multicolumn{1}{c}{\cite{Toulis2011,Toulis2015} } & \cite{Blum2017} \\\cline{4-12}
			& & & approx & approx & UB & approx & UB & approx & UB &  \\
			Worst-case graph & 2 players, $L=2$	& SP    & &$\boldsymbol{\frac{1}{2}}$ (D) && $\frac{2}{3}$ (R) & $\frac{4}{5}$ (R) & & & &\\
			&  $n$ players, $L=2$& IR  & $\boldsymbol{1}$ (D) & & &  & & & & & \\
			&                    & SP  & & $\frac{1}{2}$ (R) & $\frac{1}{2}$ (D),$\frac{7}{8}$ (R) & & & & & & \\
			&  $n$ players, $L\geq 3$ & IR  & & &&&& $\boldsymbol{\frac{1}{L-1}}$ (D) & $\frac{1}{L-1}$ (D) & &\\
			&                         & SP &  & & & & &&$\frac{7}{8}$ (R) & &\\  \hline
			Random graph     & & IR & & & & & & \multicolumn{2}{c|}{ $\checkmark$} & \multicolumn{1}{c|}{ $\checkmark$} & $\checkmark$\\
			(Bayesian model) & & SP & & && & & \multicolumn{2}{c|}{ $\checkmark$} & \multicolumn{1}{c|}{ $\checkmark$} &\\ \vspace{0.3cm}
	\end{tabular}
	\caption{Summary of approximation guarantees and upper bounds (UB) for centralized mechanisms on static {\proKEP}s. Notation D and R stands for deterministic and random mechanisms, respectively. For the Bayesian settings, approximation ratios to social optima are not provided since these references provide conditions for almost social optimal outcomes while heavily relying on the proposed Bayesian model - hence, for these works, the table only reports the type of mechanisms build.}
	\label{Table:strategyproof}
\end{table}

	\paragraph{Decentralized (non-cooperative) Multi-agent {\proKEP}s} For the first time, \citet{Carvalho2017} modeled a cross-agent {\proKEP} as a non-cooperative game where the decisions on internal exchanges are decentralized; also here, players' utilities are given by the number of their patients matched. Concretely, first, players reveal their {\proKEP} pools, second, each player decides her internal exchanges, and third, the system decides the external exchanges among the remaining pairs.  By giving to the players the power of deciding their internal exchanges  after observing the overall pool, they have incentive to truthfully present their individual {\proKEP}s.\footnote{This is a simple observation for the 2-player case (the focus of \cite{Carvalho2017}). In this work, we provide a proof of this claim for any finite number of players.}  Furthermore, it is proven and argued that the game outcome (Nash equilibrium) is a social optimum. These results focus on the 2-player caseand $L$=2.

\paragraph{Cooperative Multi-agent {\proKEP}s} Finally, there have been investigations modeling multi-agent {\proKEP}s through the lens of cooperative game theory. In this case, the contribution of each player for the overall social welfare of a multi-agent {\proKEP} is evaluated and a protocol to reward the players accordingly is determined. The cooperative and non-cooperative settings are related as both analyze the individual role of a player in the game. \citet{Klimentova2019} devise a  compensation mechanism driven by a fairness metric assessing each players' contribution. Their approach was recently compared with the use of Shapley values to guide the mechanism~\cite{Biro2020}. In another paper, \citet{Biro:2019_Cooperative} determine an international matching that is as close as possible to the core solution; a core solution is an utility distribution for the players that makes full cooperation a self-enforcing strategy.  \citet{biro2019ip} analyze the impact of cross-border KEPs when participating countries may have different constraints and objective functions.

\paragraph{Our contribution and paper structure} While in the U.S. the market is fragmented due to hospitals' strategic behavior, in European countries and in Canada, there are single national {\proKEP}s where collaboration among hospitals is mandatory. In this context, in Section~\ref{sec:background}, we focus on the case where each player (\eg, hospital) aims to maximize the number of her patients receiving a kidney. We recover from \cite{Carvalho2017} the 2-player game setting considering $L=2$ and generalize their positive theoretical results for any finite number of players,\footnote{If $L>2$, our results may not hold as discussed later in the paper.} discuss the differences with centralized mechanisms, namely, on the misrepresentation of pairs, and provide a new interpretation within nationwide programs. We remark that for 2-players, it is easy to see that the game is potential and both players benefit equally from any international matching, properties used in the results of~\cite{Carvalho2017}. Instead, for an arbitrary number of players those arguments do not extend, so the algorithms built within the constructive proofs in this work are thus independent.

While generally the priority is the maximization of the number of transplants, other utilities are also considered in {\proKEP}s; see Figure 1 in~\cite{biro_et_al2019} and \cite{KPD_report2018}. In Section~\ref{sec:qualitativeInfor}, we discuss the impact in the game social welfare of a measure of transplant quality when it is introduced in the game and thus, players'utilities move from the cardinality case to a weighted one.  Associating weights  to transplantations has been considered in the context of KEPs  in an attempt to reflect equity (prioritize certain patients). For example,  \citet{Freedman2018} develop a framework to learn societal weights and \citet{Dickerson2014} propose weights to promote group fairness for hard-to-match patients. \citet{Sommer2020} provide interval ranges for exchanges' weights that simultaneously guarantee the maximum number of transplants.  Distinct from the literature, our motivation to introduce weights comes from the efforts to predict graft quality. \citet{Kurt2011} propose a game  where pairs within an $L$-way exchange decide the time to advance with the associated transplants taking into account the evolution of their health-status. Thus, their game outcome can serve to assign weights to the exchanges in a {\proKEP}.  \citet{Massie2016} build a risk index for living donor kidney transplantation and more recently,  \citet{Luck2017,Luck2018} use machine learning techniques  on the problem of predicting the outcome (graft quality) of a kidney transplant when a pair patient-donor is matched. Concretely, they predict the time of failure for a transplant given a specific patient and a specific donor. Such prediction widely enlarges the quantity of information that is present in the system and, of course, helps determining the benefit of each matching, highly distinguishing one from another.  Information on transplantations' quality will lead to consider a utility further away from cardinality (maximizing the number of transplants) and definitely more into a weighted sum (the overall benefit of the quality of transplants).  However, when a measure of graft quality is associated with transplantations, we prove that it becomes computationally hard to even verify an equilibrium. Consequently, it may not be reasonable to expect all players to have the computational power to determine an equilibrium. To the best of our knowledge,  this is the first time that the addition of  information on the quality of a transplant to the multi-agent setting is discussed and we believe it is important to settle its associated effect on the KEP complexity.

After the rigorous theoretical analysis of Sections~\ref{sec:background} and~\ref{sec:qualitativeInfor}, the contribution of the experimental work of  Section~\ref{sec:computational} is twofold. On the negative side, we corroborate our analysis with extensive computational experiments that highlight the importance of considering game theoretical concepts for planning exchanges, namely, that not all social optima are Nash equilibria, as well as its current limitations. On the positive side, 
	\begin{enumerate}
		\item When players maximize the number of transplants, we provide strategies in order to guarantee, in practice, that the game outcome is simultaneously an equilibrium and a social welfare optimum. Moreover, we observe a high level (and thus benefit) of international exchanges at equilibira: in average, equilibria maximizing the number of transplants have at least 50\% of international exchanges.
		\item When players use additional information on the quality of the transplants, despite the negative theoretical results, equilibria close to be social optima generally exist and our algorithmic approach can efficiently compute them in practice. Similarly, we also observe a high level (and thus benefit) of international exchanges at equilibira: in average, the equilibria  with largest social welfare use 52\% of international exchanges.
	\end{enumerate}  Thus, we contribute to the design of \emph{(i) a methodology that finds the Nash equilibrium closer to the social optimum and \emph{(ii)} a protocol for proper evaluation of {\proKEP} outcomes.\footnote{For example, in the cardinality case, there might be multiple exchange plans that maximize the number of transplants. Therefore, one must be prudent when drawing conclusions based solely on one of the solutions returned by a solver.}}

Section~\ref{sec:conclusions} concludes this paper, summarizing our contributions and drawing the current open questions.

\section{Cardinality game:  {\proKEGcard}}
\label{sec:background}

\subsection{Preliminaries}

In the remaining of the paper, we will focus on KEPs restricted to 2-way exchanges. Mathematically,  limiting to $L=2$ has the advantage that the problem reduces to matchings on graphs, which is a well understood structure (that will be recalled later in the paper) and, for such case, optimizing the weighted benefit of selected exchanges can be solved in polynomial time~\cite{Papadimitriou:1982}, while for $L>2$ the problem becomes NP-hard~\cite{Abraham2007}. In practice (see, \eg, ~\cite{biro_et_al2019}), most KEPs prioritize 2-way exchanges (\eg, the Scandiatransplant KEP) as they lead to more robust solutions (in this case, a vertex failure   minimizes the damage to the overall solution) and indeed, there are national KEPs that only consider $L=2$ (\eg, France~\cite{FirstHandbook}). Furthermore, \citet{Kratz_2019} motivate and highlight the potential gains achievable with 2-way exchanges once the blood group compatibility limitation will be surpassed with the current medical technology allowing it. 

\citet{Roth_Sonmez_Unver_2005_a,Sonmez_Unver_2011} were the first considering the strategical aspects of a multi-player {\proKEP}. Next, we borrow the setup described by them but model the players interaction in a decentralized way as in~\cite{Carvalho2017}.

Let $N=\lbrace 1,2, \ldots, n \rbrace$ be the finite set of players. For each $p \in N$, $G^p=(V^p,E^p)$ is player $p$'s internal compatibility graph, \ie, the set of vertices $V^p$ represents the set of incompatible patient-donor pairs in player $p$'s pool and the set of edges $E^p$ represents the pairwise compatibilities within $V^p$. Let $E^I$ be the set of external (international) edges, \ie, $(a,b) \in E^I$ if  $a \in V^i$ and $b \in V^j$ with $i \neq j$. The set $E^I_p$ represents the subset of $E^I$ incident with vertices in $V^p$. The overall game graph is $G=(V,E)$ with $V= \cup_{p=1}^n V^p$ and $E=E^I \cup_{p=1}^n E^p$.

A subset $M^p \subseteq E^p$ is called a matching of graph $G^p$ if no two edges of it share the same node. A $M^p$-alternating path in $G^p$ is a simple path in $G^p$ that alternates between edges in $M^p$ and edges not in $M^p$. A $M^p$-augmenting path  in $G^p$  is a $M^p$-alternating path that starts in a $M^p$-unmatched vertex and ends in a $M^p$-matched vertex. In this non-cooperative game, each player $p$ strategy is a matching $M^p$ of $G^p$. The game goes as follows. First, the players simultaneously share their pools $G^p$ and $E^I$ is determined. Then, independently, each player $p$ selects a matching $M^p$ of $G^p$ and, simultaneously,  players  reveal their internal matching to the system that we call independent agent ({\IndA}). Finally, {\IndA} selects a matching of maximum cardinality $M^I(M^1, \ldots, M^n)$ in the remaining international graph $(V,E^I(M^1, \ldots, M^n))$ where $E^I(M^1, \ldots, M^n)= \lbrace (a,b) \in E^I: a,b \textrm{ not incident with edges in } \cup_{p=1}^n M^p \rbrace$. Whenever the context makes it clear, we drop the dependence of the {\IndA} matching on the players strategies and simply write $M^I$. Let $M^I_p$ be the subset of $M^I$ that is incident with vertices in $V^p$. In order to ensure that $M^I$ is unambiguously computed, for $\mathbf{M}=(M^1, \ldots, M^n)$, we define $\mathcal{A}(\mathbf{M})$ as the deterministic algorithm that computes $M^I$; given a matching $M$ such that $M \cap E^I \neq \emptyset$, then  we define $\mathbf{M}=(M\cap E^1, \ldots, M\cap E^n)$, \ie,  the matched edges in $E^I$ are ignored.
For sake of simplicity and whenever the context makes it clear, we use the notation $M^I$ to represent the {\IndA} decision that implicitly uses algorithm $\mathcal{A}$ and depends on the players internal matchings.

The utility function of each player is the number of her patients receiving a kidney, \ie,
\begin{equation}
	2 \vert M^p \vert + \vert M^I_p \vert.
	\label{Obj:cardinality}
\end{equation}
This choice of the utility is motivated by the fact that in a KEP, one of the most common primary goals is to maximize the number of transplants. We call the game with utilities~\eqref{Obj:cardinality}, cardinality game and denote it by  {\proKEGcard}. Nevertheless, in Section~\ref{sec:qualitativeInfor}, more general utility functions will be considered.

Note that by construction, it seems intuitive that each player $p$ will truthfully reveal the $G^p$ since it is player $p$ that controls the internal exchanges that will take place and, by providing true information, the player can evaluate the full potential of the game, namely, through external exchanges. The idea is that in this way player $p$ can anticipate which vertices to leave to be matched externally. For now, let us assume that this holds, and in the next section, we come back to this point and provide a formal prove to this intuition.

A player $p$ \emph{best response} to the other players matching $M^{-p}$ is a matching $M^{*^p}$ of $G^p$ that maximizes her utility. In the  {\proKEGcard}, it is
	\begin{equation}
		M^{*^p} = \argmax_{M^p} \lbrace 2\vert M^p\vert+\vert M^I_p(M^{-p},M^p) \vert: M^p\textrm{ is a matching of } G^p \rbrace.
	\end{equation}

The goal of a non-cooperative game is to determine the game outcome. The concept of solution used by \citet{Carvalho2017} is that of \textit{pure Nash equilibrium} ({\NE}). The vector of internal matchings $(M^1,M^2, \ldots, M^n)$ is a {\NE} if for each $p \in N$, there is no incentive to deviate from $M^p$. For  {\proKEGcard}, it means 
\begin{equation}
	2 \vert M^p \vert + \vert M^I_p \vert \geq 2 \vert R^p \vert + \vert M^I_p(M^{-p}, R^p) \vert \qquad \forall p \in N,  \forall \textrm{ matching } R^p \textrm{ of } G^p,
\end{equation}
where the operator $(\cdot)^{-p}$ denotes $(\cdot)$ except player $p$.

The \textit{social welfare} value for a  {\proKEGcard} where players play $(M^1,M^2, \ldots, M^n)$ is the total number of patients receiving a kidney. A \textit{social optimum} is a maximum matching of the overall graph game $G = (V,E)$, \ie, a matching in which the social welfare value is maximized. A \textit{social welfare equilibrium} ({\SWE}) is a {\NE} that is also a social optimum.

\subsection{$N$-players \proKEGcard}

\citet{Carvalho2017} proved that there is always a {\SWE} for the 2-player game that can be computed in polynomial time, and they briefly discussed the potential extension to the
$N$-player case. Their proof technique for 2 players heavily relies on the \emph{potential game} property of the 2-player case: a game is potential if there is a real-valued function from the domain of strategy profiles such  that  its  value  strictly increases   when  a  player  unilaterally  switches to a strategy that strictly increases her utility. In the next example, we show that their conjectured potential function for the $N$-player case does not hold.

\begin{example}
		Consider the 3-player kidney exchange game represented in the top of Figure~\ref{fig:pot}. Note that only player 1 (that denoted by white circles) can behave strategically as the remaining players have no internal edges. 
		
		In~\cite{Carvalho2017}, the authors conjecture that the generalization of their potential function to the $N$-player case is $\Phi(M^1,\ldots,M^n)=\sum_{p \in N} 2\vert M^p \vert+\vert M^I(M^1, \ldots,M^n) \vert.$ However, this example disproves the conjecture. In the second image of Figure~\ref{fig:pot}, player 1 selects $M^1=\lbrace (5,13),(8,9) \rbrace$, making the {\IndA} matching equal to $M^I=\lbrace (2,3),(4,12),(6,7),(10,11) \rbrace$; note that the {\IndA} has multiple matchings of maximum size and hence, in this example, we suppose that the {\IndA} has a deterministic way to select among them given $M^1$. Player 1 can increase her utility by unilaterally deviating to $M^1=\lbrace (4,5),(13,14),(8,9)\rbrace$; see the third image. Finally, player 1 can further increase her number of patients matched by playing $M^1=\lbrace (13,14),(9,10) \rbrace$; see last image. Observe the value of $\Phi$ over these unilateral deviations of player 1: $\Phi$ did not strictly increase whenever player 1 unilaterally deviated to a strictly better strategy. This contradicts the definition of potential function.
		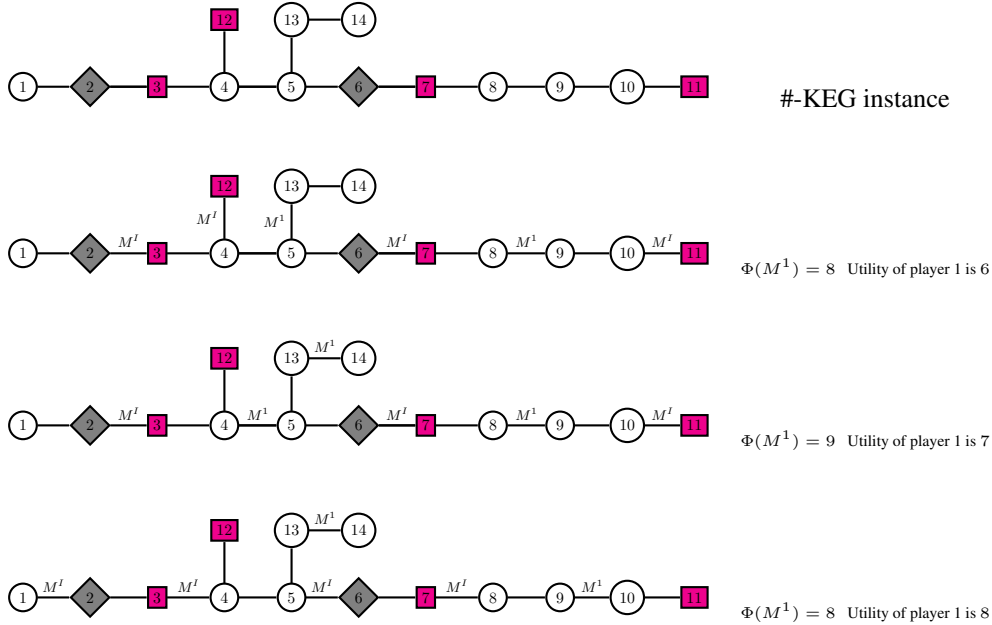
\begin{figure} \centering
			\begin{tabular}{cc}
				\begin{tikzpicture}[-,>=stealth',shorten >=0.3pt,auto, scale=0.7,node distance=1.5cm,
					thick,countryA node/.style={circle,draw},countryB node/.style={diamond,draw,fill=gray},countryC node/.style={rectangle,draw,fill=magenta}, scale=.85, transform shape]
					\tikzstyle{matched} = [draw,line width=3pt,-]
					
					\node[countryA node] (1) {$1$};
					\node[countryB node] (2) [right of=1] {$2$};
					\node[countryC node] (3) [right of=2] {$3$};
					\node[countryA node] (4) [right of=3] {$4$};
					\node[countryA node] (5) [right of=4] {$5$};
					\node[countryB node] (6) [right of=5] {$6$};
					\node[countryC node] (7) [right of=6] {$7$};
					\node[countryA node] (8) [right of=7] {$8$};
					\node[countryA node] (9) [right of=8] {$9$};
					\node[countryA node] (10) [right of=9] {$10$};
					\node[countryC node] (11) [right of=10] {$11$};
					\node[countryC node] (12) [above of=4] {$12$};
					\node[countryA node] (13) [above of=5] {$13$};
					\node[countryA node] (14) [right of=13] {$14$};
					
					\path[-,every node/.style={font=\sffamily\small}]
					(1) edge node  {} (2)
					(2) edge node  {} (3)
					(3) edge node  {} (4)
					(4) edge node  {} (5)
					(5) edge node  {} (6)
					(6) edge node  {} (7)
					(2) edge node  {} (3)
					(4) edge node  {} (5)
					(6) edge node  {} (7)
					(7) edge node  {} (8)
					(8) edge node  {} (9)
					(9) edge node  {} (10)
					(10) edge node  {} (11)
					(4) edge node  {} (12)
					(5) edge node  {} (13)
					(13) edge node  {} (14);
				\end{tikzpicture}
				& \#-KEG instance \\[2em]
				\begin{tikzpicture}[-,>=stealth',shorten >=0.3pt,auto, scale=0.7,node distance=1.5cm,
					thick,countryA node/.style={circle,draw},countryB node/.style={diamond,draw,fill=gray},countryC node/.style={rectangle,draw,fill=magenta}, scale=.85, transform shape]
					\tikzstyle{matched} = [draw,line width=3pt,-]
					
					\node[countryA node] (1) {$1$};
					\node[countryB node] (2) [right of=1] {$2$};
					\node[countryC node] (3) [right of=2] {$3$};
					\node[countryA node] (4) [right of=3] {$4$};
					\node[countryA node] (5) [right of=4] {$5$};
					\node[countryB node] (6) [right of=5] {$6$};
					\node[countryC node] (7) [right of=6] {$7$};
					\node[countryA node] (8) [right of=7] {$8$};
					\node[countryA node] (9) [right of=8] {$9$};
					\node[countryA node] (10) [right of=9] {$10$};
					\node[countryC node] (11) [right of=10] {$11$};
					\node[countryC node] (12) [above of=4] {$12$};
					\node[countryA node] (13) [above of=5] {$13$};
					\node[countryA node] (14) [right of=13] {$14$};
					
					\path[-,every node/.style={font=\sffamily\small}]
					(1) edge node  {} (2)
					(2) edge node  {$M^I$} (3)
					(3) edge node  {} (4)
					(4) edge node  {} (5)
					(5) edge node  {} (6)
					(6) edge node  {$M^I$} (7)
					(4) edge node  {} (5)
					(6) edge node  {} (7)
					(7) edge node  {} (8)
					(8) edge node  {$M^1$} (9)
					(9) edge node  {} (10)
					(10) edge node  {$M^I$} (11)
					(4) edge node  {$M^I$} (12)
					(5) edge node  {$M^1$} (13)
					(13) edge node  {} (14);
				\end{tikzpicture} & \tiny $\Phi(M^1)=8$ \ \  Utility of player 1 is $6$\\[2em]
				\begin{tikzpicture}[-,>=stealth',shorten >=0.3pt,auto, scale=0.7,node distance=1.5cm,
					thick,countryA node/.style={circle,draw},countryB node/.style={diamond,draw,fill=gray},countryC node/.style={rectangle,draw,fill=magenta}, scale=.85, transform shape]
					\tikzstyle{matched} = [draw,line width=3pt,-]
					
					\node[countryA node] (1) {$1$};
					\node[countryB node] (2) [right of=1] {$2$};
					\node[countryC node] (3) [right of=2] {$3$};
					\node[countryA node] (4) [right of=3] {$4$};
					\node[countryA node] (5) [right of=4] {$5$};
					\node[countryB node] (6) [right of=5] {$6$};
					\node[countryC node] (7) [right of=6] {$7$};
					\node[countryA node] (8) [right of=7] {$8$};
					\node[countryA node] (9) [right of=8] {$9$};
					\node[countryA node] (10) [right of=9] {$10$};
					\node[countryC node] (11) [right of=10] {$11$};
					\node[countryC node] (12) [above of=4] {$12$};
					\node[countryA node] (13) [above of=5] {$13$};
					\node[countryA node] (14) [right of=13] {$14$};
					
					\path[-,every node/.style={font=\sffamily\small}]
					(1) edge node  {} (2)
					(2) edge node  {$M^I$} (3)
					(3) edge node  {} (4)
					(4) edge node  {$M^1$} (5)
					(5) edge node  {} (6)
					(6) edge node  {$M^I$} (7)
					(4) edge node  {} (5)
					(6) edge node  {} (7)
					(7) edge node  {} (8)
					(8) edge node  {$M^1$} (9)
					(9) edge node  {} (10)
					(10) edge node  {$M^I$} (11)
					(4) edge node  {} (12)
					(5) edge node  {} (13)
					(13) edge node  {$M^1$} (14);
				\end{tikzpicture} & \tiny $\Phi(M^1)=9$ \ \  Utility of player 1 is $7$\\[2em]
				\begin{tikzpicture}[-,>=stealth',shorten >=0.3pt,auto, scale=0.7,node distance=1.5cm,
					thick,countryA node/.style={circle,draw},countryB node/.style={diamond,draw,fill=gray},countryC node/.style={rectangle,draw,fill=magenta}, scale=.85, transform shape]
					\tikzstyle{matched} = [draw,line width=3pt,-]
					
					\node[countryA node] (1) {$1$};
					\node[countryB node] (2) [right of=1] {$2$};
					\node[countryC node] (3) [right of=2] {$3$};
					\node[countryA node] (4) [right of=3] {$4$};
					\node[countryA node] (5) [right of=4] {$5$};
					\node[countryB node] (6) [right of=5] {$6$};
					\node[countryC node] (7) [right of=6] {$7$};
					\node[countryA node] (8) [right of=7] {$8$};
					\node[countryA node] (9) [right of=8] {$9$};
					\node[countryA node] (10) [right of=9] {$10$};
					\node[countryC node] (11) [right of=10] {$11$};
					\node[countryC node] (12) [above of=4] {$12$};
					\node[countryA node] (13) [above of=5] {$13$};
					\node[countryA node] (14) [right of=13] {$14$};
					
					\path[-,every node/.style={font=\sffamily\small}]
					(1) edge node  {$M^I$} (2)
					(2) edge node  {} (3)
					(3) edge node  {$M^I$} (4)
					(4) edge node  {} (5)
					(5) edge node  {$M^I$} (6)
					(6) edge node  {} (7)
					(7) edge node  {$M^I$} (8)
					(8) edge node  {} (9)
					(9) edge node  {$M^1$} (10)
					(10) edge node  {} (11)
					(4) edge node  {} (12)
					(5) edge node  {} (13)
					(13) edge node  {$M^1$} (14);
				\end{tikzpicture} & \tiny $\Phi(M^1)=8$ \ \  Utility of player 1 is $8$\\
			\end{tabular}
			\caption{(Top) Compatability graph where player 1 owns the white circles, player 2 owns gray diamonds and player 3 owns the magenta squares. (Remain) Unilateral deviations of player 1.}
			\label{fig:pot}
		\end{figure}
	\end{example}

We take on from this discussion and we formally prove the generalization of the result in \cite{Carvalho2017} without using a potential function argument.

\begin{theorem}
	For any game with a set of players $N$ and any deterministic algorithm $\mathcal{A}$ (cardinality maximizing) for the {\IndA}, there is a {\SWE} and it can be computed in polynomial time.
	\label{THM:SWE_EXISTENCE}
\end{theorem}
\proof
See the proof in Appendix~\ref{app:SWE}.

We note that the underlying polynomial time algorithm of Theorem~\ref{THM:SWE_EXISTENCE} does not use a best-response dynamics as in \cite{Carvalho2017}. In fact, in each iteration, it can break simultaneously a series of instabilities, \ie, players incentive to deviate from a current matching.

{\SWE} are Pareto efficient and, for the 2-player case, \citet{Carvalho2017} build an algorithm that, given any {\NE} that is not a social optimum as input, outputs  a {\SWE} that dominates it. Thus, it is reasonable to assume that players will focus on  {\SWE}. This phenomenon is called the focal-point-effect; see \cite{Schelling1960}.

\subsection{Market Design: impact on social welfare}

The way in which players' interaction is designed in a game significantly impacts the social welfare. Indeed, by looking at Table~\ref{Table:strategyproof}, it can be concluded that strategy-proof mechanisms cannot guarantee that the social optimum is attained. Although individually rational mechanisms can guarantee a social optimum, note that this does not ensure that players will participate with their full pool of incompatible patient-donor pairs. Let us clarify these statements by recovering an example from~\cite{Roth_Sonmez_Unver_2005_a,Sonmez_Unver_2011}:
\begin{example}[\citet{Roth_Sonmez_Unver_2005_a,Sonmez_Unver_2011}]
	Consider a 2-player kidney exchange game. The graph in Figure~\ref{CirclesDiamonds} represents the overall compatibility graph.
	
	Both players alone can match 2 vertices. This means that  an individually rational mechanism must ensure that each player has at least two of her vertices matched. Assume that the central authority outputs a maximum matching. The graph has 4 maximum matchings, $M= \lbrace (2,3), (4,5), (6,7)  \rbrace$, $R= \lbrace (1,2), (3,4), (5,6)  \rbrace$, $S=\lbrace (1,2),(4,5),(6,7) \rbrace$ and $T=\lbrace (1,2),(3,4),(6,7) \rbrace$, and all have the property of being individually rational. However, none of these four social optimal solutions is strategy-proof: under $M$ and $T$, player 1 has incentive to hide and match $(5,6)$, forcing the mechanism to select $\lbrace (1,2), (3,4) \rbrace$; under $R$ and $S$, player 2 has incentive to hide and match $(2,3)$, forcing the mechanism to select $\lbrace (4,5), (6,7)  \rbrace$.

	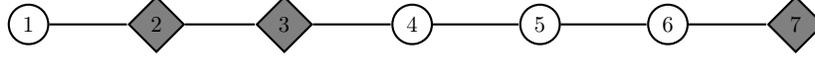
\begin{figure}
		\centering
		\begin{tikzpicture}[-,>=stealth',shorten >=0.3pt,auto,node distance=2cm,
			thick,countryA node/.style={circle,draw},countryB node/.style={diamond,draw,fill=gray}, scale=.85, transform shape]
			\tikzstyle{matched} = [draw,line width=3pt,-]
			
			\node[countryA node] (1) {$1$};
			\node[countryB node] (2) [right of=1] {$2$};
			\node[countryB node] (3) [right of=2] {$3$};
			\node[countryA node] (4) [right of=3] {$4$};
			\node[countryA node] (5) [right of=4] {$5$};
			\node[countryA node] (6) [right of=5] {$6$};
			\node[countryB node] (7) [right of=6] {$7$};
			
			\path[-,every node/.style={font=\sffamily\small}]
			(1) edge node  {} (2)
			(2) edge node  {} (3)
			(3) edge node  {} (4)
			(4) edge node  {} (5)
			(5) edge node  {} (6)
			(6) edge node  {} (7);
			(2) edge node  {} (3)
			(4) edge node  {} (5)
			(6) edge node  {} (7);
		\end{tikzpicture}

		\caption{Player 1 owns the white circles and player 2 owns the gray diamonds.}
		\label{CirclesDiamonds}
	\end{figure}
	
\end{example}

\begin{figure}
	\centering
	\begin{minipage}{.5\textwidth}
		\centering
		\begin{tikzpicture}[-,>=stealth',shorten >=0.3pt,auto,node distance=2cm,
			thick,countryA node/.style={circle,draw},countryB node/.style={diamond,draw,fill=gray},countryC node/.style={rectangle,draw,fill=magenta} , scale=.85, transform shape]
			\tikzstyle{matched} = [draw,line width=3pt,-]
			
			\node[countryA node] (1) {$1$};
			\node[countryC node] (3) [right of=1] {$3$};
			\node[countryC node] (2) [above of=3] {$2$};
			\node[countryB node] (4) [below of=3] {$4$};
			
			\path[-,every node/.style={font=\sffamily\small}]
			(1) edge node  {} (2)
			(1) edge node  {} (3)
			(1) edge node  {$M^I$} (4);
		\end{tikzpicture}
	\end{minipage}%
	\begin{minipage}{0.5\textwidth}
		\centering
		\begin{tikzpicture}[-,>=stealth',shorten >=0.3pt,auto,node distance=2cm,
			thick,countryA node/.style={circle,draw},countryB node/.style={diamond,draw,fill=gray},countryC node/.style={rectangle,draw,fill=magenta} , scale=.85, transform shape]
			\tikzstyle{matched} = [draw,line width=3pt,-]
			
			\node[countryA node] (1) {$1$};
			\node[countryC node] (3) [right of=1] {$3$};
			\node[countryC node, dashed] (2) [above of=3] {$2$};
			\node[countryB node] (4) [below of=3] {$4$};
			
			\path[-,every node/.style={font=\sffamily\small}]
			(1) edge[dashed]  node  {} (2)
			(1) edge node  {$M^I$} (3)
			(1) edge node  {} (4);
		\end{tikzpicture}
	\end{minipage}%
	\caption{Player 1 owns the white circles, player 2 owns the gray diamonds and player 3 owns the magenta squares. In the left figure, when all players reveal their pairs, {\IndA} selects the matching $\lbrace (1,4) \rbrace$. In the right figure, player 3 hides vertex 2 and the {\IndA} selects the matching $\lbrace (1,3) \rbrace$. Then, this $\mathcal{A}$ of  {\IndA} is not giving incentive for full participation. On the other hand, if for the left case, $\mathcal{A}$ selects  $\lbrace (1,3) \rbrace$, independently on the decision of  $\mathcal{A}$ for the right case, all players have incentive (or are indifferent) to reveal their full pools. }
	\label{fig:fullInformation}
\end{figure}
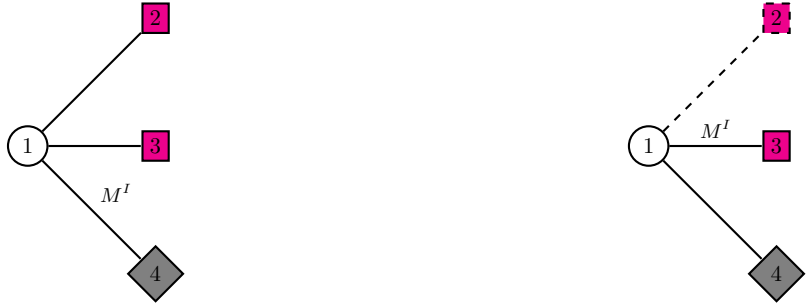

On the other hand, the decentralized non-cooperative game-theory setup, \#-KEG, overcomes the drawbacks of individually and strategy-proof mechanisms. The intuition on why \#-KEG makes it dominant for each player $p$ to reveal her full pool $G^p$ is on the distinction with centralized mechanisms: under \#-KEG, players have full control of their internal matchings and therefore, they must benefit from sharing $G^p$ and anticipating the potential external matchings. For example, note that player 1 from Figure~\ref{CirclesDiamonds} does not need to hide vertices 5 and 6 as \#-KEG allows her to internally match these vertices. The
	idea is that under our game rules, hiding can be seen as internally matching vertices. Nevertheless, this intuition has to be cautiously taken.  Figure~\ref{fig:fullInformation} shows an example where the design of $\mathcal{A}$ plays an important role on incentivizing or desincentivizing players to fully reveal their pools of patient-donor pairs. We have the following sufficient condition,\footnote{We note that this result is in line with the recent paper~\cite{Smeulder2020} where it is studied the complexity of determining the set of vertices that a player must misrepresent within a centralized mechanism. The authors show that for $L=2$, this optimization problem can be reduced to computing the optimal set of internal matchings for that player in \#-KEG.}

	\begin{lemma}
		When there are multiple optimal solutions for the {\IndA}, if $\mathcal{A}$ lexicographically optimizes the player with more vertices in the game, then there is no strict incentive to hide vertices under \#-KEG. 
		\label{lem:full}
	\end{lemma}
	\begin{proof}
		Consider player $p$ and her  opponents' strategies $M^{-p}$. We will show that, player $p$ has no incentive to hide a subset of her vertices if $\mathcal{A}$ prioritizes the players with more vertices in the game when deciding among multiple optimal solutions.
		
		By contradiction, suppose that player $p$ has incentive to hide the set of vertices $V' \subseteq V^p$ given some $M^{-p}$. Assuming that no information can be hidden, let $M^p$ be player $p$'s best response and $ M^I_p $ the associated external matching involving the vertices of this player. Likewise, let $\hat{M}^p$ be player $p$'s best response when vertices $V'$ are not available in the game, $\hat{M}^I_p$ the associated external matching involving the vertices of this player and $\bar{M}^p$ the maximum matching in the graph induced by $V'$. Mathematically, we should have
		$$2\vert M^p \vert +   \vert M^I_p \vert < 2\vert \bar{M}^p \vert + 2\vert \hat{M}^p \vert +   \vert \hat{M}^I_p \vert,$$
		otherwise, player $p$ would not have incentive to hide vertices $V'$.
		
		If $\bar{M}^p$ matches all vertices in $V'$, then such matching could have been implemented in the game without any damage  for player $p$'s objective function; recall that, by the hypothesis on $\mathcal{A}$, increasing  the number of revealed vertices can can only benefit a player. Otherwise, there is some vertex $v \in V'$ that is not matched. Note that revealing such vertex must decrease player $p$'s objective function. If adding $v$ to the induced graph $G^p \left[ (V^p - V') \cup \lbrace v \rbrace \right]$, results in a $\hat{M}^p$-augmenting path, then player $p$ has no incentive to hide $v$. Thus, since $v$ is not matched internally (either with hidden or non-hidden vertices), player $p$ is indifferent to reporting it to the game. In fact, supposing that player $p$ hides $V'- \lbrace v \rbrace$ and plays $\hat{M}^p$, the external matching selected by $\mathcal{A}$ can only mantain or improve player $p$'s utility. If there is a $\hat{M}^I_p$-augmenting path $\mathfrak{p}$, starting in the unmatched vertex $v$,  $\hat{M}^I_p \oplus \mathfrak{p}$, only adds more matched vertices and hence, player $p$'s objective increases by 1 unit. Thus, there is no benefit on hiding $v$ and $v$ can be removed from $V'$. By iteratively eliminating such vertices from $V'$, either it becomes empty or we get into the first case (all vertices in $V'$ are matched).
		
		Since we selected an arbitrary player and an arbitrary strategy for her opponents, the result follows.
	\end{proof}


Furthermore, the sufficient condition of Lemma~\ref{lem:full} can be efficiently implemented:

\begin{theorem}
	The computation of a maximum matching prioritizing the players with more vertices revealed can be done in polynomial time.
	\label{THM:FULLINFORMATION}
\end{theorem}
\proof
See Appendix~\ref{app:FullInformation}.

In addition to the design of a game with good outcomes from the social welfare point of view and from the players' perspective in a cross-boarder KEP, the result of Theorem~\ref{THM:SWE_EXISTENCE} has major consequences in the understanding of {\proKEP}. It does not matter how we partition a pool of a {\proKEP} into different groups (players): for each partition, there is always a Nash equilibrium that is also a social optimum. For instance, in a {\proKEP}, players can be thought as agents owning incompatible pairs of distinct regions, or distinct incomes, or distinct minorities, etc., and according to Theorem~\ref{THM:SWE_EXISTENCE}, for each partition of the players' pairs,  there will always exist an equilibrium that is also a social optimum. In a certain sense,  one could interpret our result as  a  theoretical explanation of why for regional or national {\proKEP}s (\eg, in UK and Netherlands where exchanges of size 2 are prioritized and they represent a significant fraction of the selected exchanges~\cite{FirstHandbook}) the solution of a single decision-maker optimization problem is enough. Nevertheless, the computational results of Section~\ref{sec:computational} will highlight that these maximum matchings must be carefully computed because their number can be large and only a very small fraction of them can be Nash equilibria. Thus, under strategical behavior, for the determination of a {\SWE}, it is essential the use of our algorithmic tools that allow to \emph{(i)} compute a {\SWE} in polynomial time (Algorithm~\ref{Alg:SWE_matching}), \emph{(ii)} verify Nash equilibria (Algorithm~\ref{Alg:Best_weighted_Optimistic_response}, mentioned later in Section~\ref{sec:computational})  and \emph{(iii)} sample among maximum matchings (Appendix~\ref{app:uniform_gen}, mentioned later in Section~\ref{sec:computational}).


\section{Information on transplant's quality: what changes? The \proKEGwei\ case}
\label{sec:qualitativeInfor}

Predicting the graft survival of a patient matched with a donor does not depend only on whether they are compatible, but also on their health characteristics (age, life habits, etc.). This means that, for a set of compatible donors, a patient associates a preference according to the predicted graft quality. Thus, when additional information on the quality of   transplantations is added to the kidney exchange game, it is expected that the players will weight that information on their utility functions. For instance, a player optimization problem can become a maximization of a weighted matching, instead of the cardinality one.

Let us emphasize the importance of having measures on the quality of a transplant. \citet{Ahmed2008} build 5-year graft survival predictions for living-donors with a Cox regression-based monogram and an artificial neural network model. Their motivation is to find the best prediction methodology since that would allow the choice of the best possible kidney donor and the optimum immunosuppressive therapy for a given patient.  The statistical study in \cite{nemati_does_2014} shows that living-donor recipients have better graft survival rates in comparison with deceased-donation. Nevertheless, there is literature supporting the fact that predicting these rates for living-donation can significantly improve the hospitals allocation strategies, as even among living donors there can be significant differences. \citet{krikov2007} develop a tree-based model to predict the probability of kidney graft survival at 1, 3, 5, 7, and 10 years, either with deceased and living donors. It is claimed that such predictions allow to identify the factors affecting the survival  that  may help in improving treatment strategies. \citet{amiajnl-2010} estimate glomerular filtration rate (eGFR) of the recipient 1 year after the transplant. It is argued that since eGFR is an interpretable real-valued quantity, it might be more helpful than a binary success/failure prediction.  \citet{Massie2016} develop an index to be used as a metric to compare multiple living donors. More recently, \citet{Luck2017} propose a deep neural network to predict graft survibability given a specific donor and patient; however, their experiments focus on deceased donors.

One way to embed graft quality information on KEPs is through the association of weights to the edges of the compatability graph. In fact, some KEPs already consider edge weights but mainly with the goal of representing exchange prioritizations. Hence, one would be interested in finding a maximum-weighted matching, which can be computed in polynomial time \citep{Papadimitriou:1982}. Let $w^p_e$ for $e \in E$ encode the information on graft quality associated with that matching for player $p$. At first glance, because in the kidney exchange game the objective functions~\eqref{Obj:cardinality}  of each player $p$ is replaced (generalized) by
\begin{equation}
	\sum_{e \in M^p} w^p_e + \sum_{e \in M^I_p} w^p_e,
	\label{Obj:weight}
\end{equation}
it seems that the game would continue to be ``\textit{easy to solve}''; we call the game with general utility functions~\eqref{Obj:weight} and  the {\IndA} objective with weights $w^I_{vu} = w_{vu}^i+w_{vu}^j$, for each edge $(v,u) \in E^I$ and $v \in V^i$ and $u \in V^j$, the weighted game and denote it by {\proKEGwei}. However, when a measure of graft quality is considered, we prove  that it becomes \textit{hard} to compute a player best response   and thus, to verify if a matching is an equilibrium. The \textit{social welfare} value for  a {\proKEGwei} where players play $(M^1,M^2, \ldots, M^n)$ is generalized to the total weight associated with patients receiving a kidney, where international and internal exchanges are evaluated according to {\IndA} and the associated player, respectively.

Player $p$ best response to a matching $M^{-p}$ corresponds to finding the internal matching $M^p$ of $G^p$ such that her total weight is maximized. For the sake of simplicity, in our investigation of player $p$'s best response, we assume that the {\IndA} continues to maximize the cardinality matching in $E^I$ requiring that a $\mathcal{A}$ is specified; note that this is always possible through a proper selection of the weights for the remaining players.
%
We will simplify even further and take inspiration from bilevel programming (see \eg ~\cite{Colson2007}),\footnote{Bilevel programming models the optimization problem of a player, called the leader, subject to \emph{a posteriori} reaction of another player, called the follower. For the leader's optimization to be well defined, in the literature, it is assumed that the follower behaves either pessimistically or optimistically. Note the direct connection between bilevel programming and Stackelberg games.} and assume that the $\mathcal{A}$ either behaves \emph{pessimistically}, \ie, the {\IndA} selects the matching of maximum cardinality that minimizes player $p$'s utility, or \emph{optimistically}, \ie, the {\IndA} selects her matching of maximum cardinality that maximizes player $p$'s utility. The pessimistic (respectively, optimistic) decision version, denoted by {\BestReponsepes} (respectively, {\BestReponseopt}), asks whether given a positive integer $P$, there is a matching $M^p \subseteq E^p$ such that player $p$ has a profit of at least $P$.
\begin{theorem}
	{\BestReponsepes} and {\BestReponseopt}  are NP-complete.
	\label{THM:COMPLEXITY}
\end{theorem}
\proof See Appendix~\ref{app:NP_complete}.

The result of Theorem~\ref{THM:COMPLEXITY} implies that the optimization problem of each player can become computationally intractable (assuming $P \neq NP$), raising the question on whether Nash equilibria can be attained. Furthermore, it might explain why some multi-agent kidney exchanges lead to so little transplantations as reported in~\cite{Ashlagi2017}: the players objective function is not only the number of patients receiving a kidney  but, potentially, a weighted benefit of their patients receiving a kidney, which can be \emph{hard} to optimize. We remark that in~\cite{Ashlagi2017}, the analyzed game is centralized and thus, the players optimize their outcome by deciding which vertices to reveal. This is not our game setup, although we can see pairs of vertices matched internally as vertices that could be hidden under the centralized mechanism.\footnote{Recently, \citet{Smeulder2020} showed the computational complexity of the decision problem associated with a player maximizing the number of her patients receiving a kidney by selecting a subset of vertices to hide under a centralized mechanism. In particular, they proved that this optimization can be reformulated as the best response problem in {\proKEGcard} and hence, it can be solved in polynomial time.}

In this context, it is crucial to study how the introduction of edge weights impacts the incentives for misrepresentation of individual pools.

\begin{lemma}
	For {\proKEGwei}, even in the 2-player case, a player can have incentive to hide vertices.
\end{lemma}
\begin{proof}
	Consider the example of Figure~\ref{WKEG:hidevertices} where players have no internal matching available and the {\IndA}  unique optimal solution\footnote{Recall that according to our description of {\proKEGwei}, the {\IndA} computes a maximum-weighted matching. On the other hand, for the computational complexity classification of a player best response, we can assume that {\IndA} used unitary weights since it guarantees that even for such subset of instances the problem is NP-hard.} is the matching $\lbrace (1,2),(3,4) \rbrace$ with value 12, and utility 2 for player 1 and 10 for player 2. Note that player 1 has incentive to hide vertex 1, so that the  {\IndA} optimal solution becomes matching $\lbrace (2,3) \rbrace$ with utility 10 for player 1.
	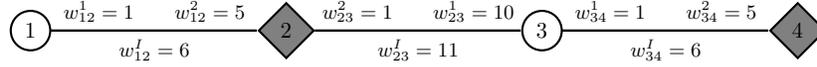
\begin{figure}
		\centering
		\begin{tikzpicture}[-,>=stealth',shorten >=0.3pt,auto,node distance=4cm,
			thick,countryA node/.style={circle,draw},countryB node/.style={diamond,draw,fill=gray}, scale=.85, transform shape]
			\tikzstyle{matched} = [draw,line width=3pt,-]
			
			\node[countryA node] (1) {$1$};
			\node[countryB node] (2) [right of=1] {$2$};
			\node[countryA node] (3) [right of=2] {$3$};
			\node[countryB node] (4) [right of=3] {$4$};
			
			\path[-,every node/.style={font=\sffamily\small}]
			(1) edge node  {$w_{12}^1=1 \qquad w_{12}^2=5$} (2)
			(1) edge node[below]  {$w_{12}^I=6$} (2)
			(2) edge node  {$w_{23}^2=1 \qquad w_{23}^1=10$} (3)
			(2) edge node[below]  {$w_{23}^I=11$} (3)
			(3) edge node  {$w_{34}^1=1 \qquad w_{34}^2=5$} (4)
			(3) edge node[below]  {$w_{34}^I=6$} (4);
		\end{tikzpicture}

		\caption{Player 1 owns the white circles and player 2 owns the gray diamonds.}
		\label{WKEG:hidevertices}
	\end{figure}
\end{proof}

\begin{lemma} 
	For  {\proKEGwei}, assuming that players cannot hide vertices, a player can have incentive to misreport her evaluation of edges' weights.
\end{lemma}
\begin{proof} 
	If some player $p$ reports a very large $w^p_{vu}$ for $(v,u) \in E^I$ and $v \in V^p$, then {\IndA} will give priority to matching the patients of player $p$. Recall that here we are not analyzing an individual player best response, hence, the {\IndA} assignings weights to edges that correspond to the sum of players' evaluation. For instance, in the example of Figure~\ref{WKEG:hidevertices}, player 1 has incentive to report $w^1_{23}=12$.
\end{proof}

	These results show that for {\proKEGwei}, we cannot conclude that complete information will naturally be achieved. If there was never incentive to hide a vertex, we could assume that the system in place would have its own means to assign values to the international exchanges and thus, misreporting patient information would not  be possible. Alternatively, in order to avoid that players report very large weights, the {\IndA} can assign them a total budget for it or the players should agree on an algorithm that computes the weights. In addition, since already computing the best response is NP-hard, we may assume that players do not have the computational power to find the optimal set of vertices to hide.

While the discussion above can mildly support the assumption of complete information, it still remains to understand the equilibria of {\proKEGwei}. In~\cite{Carvalho2017}, an example of a weighted kidney exchange game with no pure Nash equilibria is provided. This observation together with the theoretical complexity of a player best response motivate the computational analysis of the next section, where we aim at understanding the practical impact of the negative theoretical results above. On the negative side, we will conclude that computing any social optimum (\ie, to solve the problem as a single decision-maker) should unlikely result in an equilibrium. In a non-equilibrium outcome, players do not have incentive to follow it and to contribute with full information. On the positive side, we will \emph{(i)} show that potential Nash equilibria can be effectively verified in practice, and \emph{(ii)} propose novel algorithmic policies required to address multi-agent interactions. Those results allow, in practice, to get a kidney exchange program that engages agents and positively impacts the social welfare.

\section{Computational investigation}
\label{sec:computational}

In this section, we will compute pure Nash equilibria of the kidney exchange game when each player goal is to maximize the number of her patients receiving a kidney, objective function \eqref{Obj:cardinality}, and when the goal is to maximize the total weight associated with her patients receiving a kidney, objective function \eqref{Obj:weight}. To that end, we contribute with a class of realistic instances based on data from the Canadian Kidney Paired Donation program, as well as a generator for this class. The aim of this section is twofold:  evaluate the rate of social welfare equilibria for the cardinality case and  the rate of pure Nash equilibria for the weighted case. These experiments allow to understand if we really need equilibria (or we would most likely stumble in an equilibrium by simply maximizing transplants) and if they increase players' utilities and welfare.

In Section~\ref{subsec:implementation_details}, we clarify how we build our generator of instances and the algorithmic details to verify if a matching is an equilibrium. In Section~\ref{subsec:computational_simulations}, we evaluate the rate of equilibria on the generated instances for the cardinality game, {\proKEGcard}, and for the weighted game, {\proKEGwei}.

The instances used in our experiments,  the code of our implementations and all the obtained results are publicly  available.\footnote{\url{https://github.com/mxmmargarida/KEG}}

Finally, for the computational investigation, we assume complete information. Recall that this is theoretically motivated for {\proKEGcard}, while it is an assumption for {\proKEGwei}, which the computation helps to clarify.

\subsection{Implementation details}\label{subsec:implementation_details}

\paragraph{Test instances} Typically, the kidney exchange literature carries the computational analysis on graphs generated according  to~\citet{Saidman2006}; compatibility graphs are generated based on probabilities of blood type and of patient-donor tissue compatibility that were obtained through data from the United Network for Organ Sharing (UNOS) in the United States. \citet{ashlagi_new_2012} show that, in practice, kidney exchange pools can be much sparser than the pools generated by the Saidman simulator. For this reason, ~\citet{Glorie2014} introduced some modifications on this simulator.  \citet{Constantino2013} produced KEP instances with low, medium and high density graphs; the low density graphs have on average a density similar to the ones by~\citet{Saidman2006}. Since, there are no computational results on game theory approaches  using the aforementioned instances, we generated instances based on data from the Canadian Kidney Paired Donation program (KPD), contributing with a new set of realistic instances, as well as a generator. Our set of players are the Canadian provinces.

The Canadian Blood Services provided us data on the KPD, which enabled us to build the instance generator described in Table~\ref{table:generator}.
The only input parameters of the generator are $\vert V \vert$, the total number of incompatible patient-donor pairs, and a year between 2009 and 2013. The prevalence of highly-sensitized patients, \ie, patients with low probability of being compatible with a random donor, is generated based on~\cite{KPD_report2009_20013} for the years between 2009 and 2013. Determining the degree of sensitization is done through the calculated panel reactive antibody (cPRA).  For instance, according to that report, in 2009 there were approximately 14\% patients on the KPD with cPRA$\geq 97\%$, while in 2013 it increased to 53\%; thus, this input parameter controls the density of the instances.


{\tiny
	\begin{table}[H]\tiny
		\begin{tabularx}{\textwidth}{|X|X|}\hline
			Step 1
			
			Assignment of pairs to provinces &  Given $\vert V \vert$, our generator determines the number of pairs from $V$ belonging to each province according to their distribution in the KPD data; see the pie chart of Figure~\ref{fig:size_dist}.
			\\ \hline
			Step 2
			
			Generation of players vertices/Generation of provinces' pairs &  Given a province and the number of pairs belonging to it (determined in the previous step), our generator assigns to each pair, the patient blood group and cPRA, and the donor blood group. This assignment is based on the patient-donor blood group distribution of that province, while the cPRA distribution is based on its national distribution in the KPD (due to the lack of data for each specific province). Patients are considered highly-sensitized when their cPRA is high (high probability of being crossmatch incompatible). In~\cite{KPD_report2009_20013}, it is reported the distribution of the  cPRAs equal to 0\%, between 1-50\%, 51-94\%, 95-96\% and 97-100\% for each year between 2009 and 2013.  This distribution changed along those 5 years, and it seems to start stabilizing. However, we decided to leave as a parameter of our generator the selection of the year determining the cPRA distribution; the reason for this change is that on the first year of the report, 2009, there were 14\% of the patients with cPRA above 97\%, while in 2013, they become 53\% of the patients in the pool; by having the year as a parameter for the cPRA distribution, we will be able to understand the impact of highly-sensitized patients in the game outcome.  The higher the cPRA  value of a patient, the less likely is she compatible with a random donor. \\ \hline
			Step 3
			
			Generation of edges/Generation of compatabilities among pairs & In order to build the edges representing reciprocal compatibilities between pairs, blood compatibility is verified and, if satisfied, an edge is drawn with probability equal to the multiplication of (1- cPRA) of the patients on that pairs (i.e., the probability of  both patients being compatible with the donors).\\ \hline
		\end{tabularx}
		\caption{Instance generator}
		\label{table:generator}
	\end{table}
}

\begin{figure}
	\centering
	\includegraphics[scale=0.4]{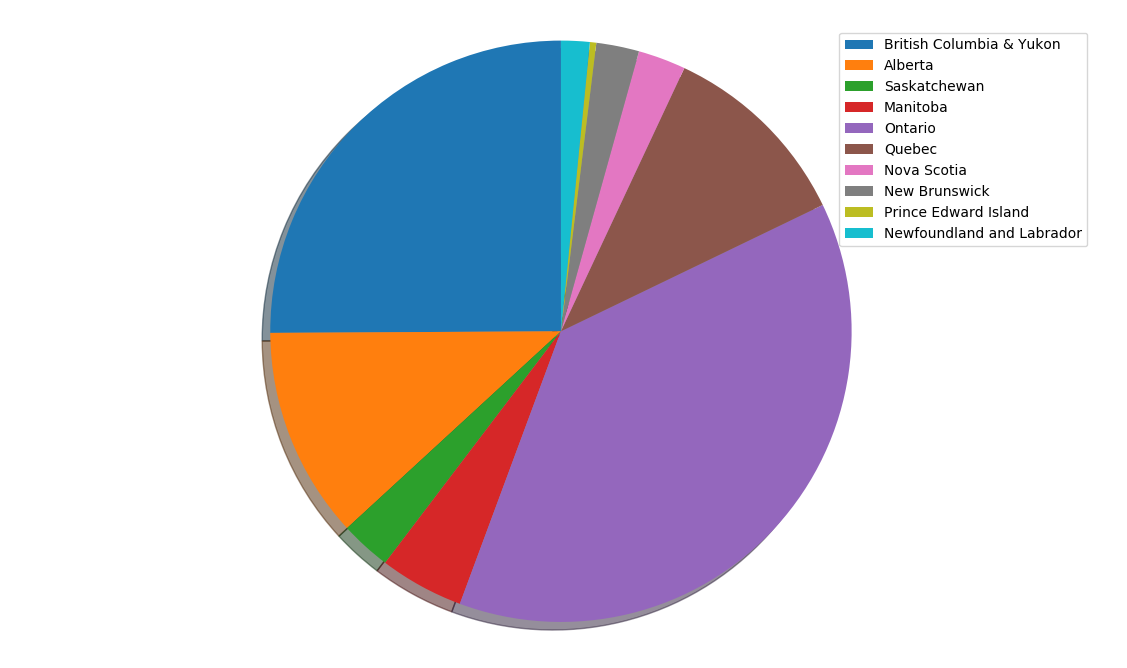}
	\caption{Distribution of provinces size.}
	\label{fig:size_dist}
\end{figure}


For each $\vert V \vert \in \lbrace 30, 40, 50 \rbrace$  and year $\in \lbrace 2009,2013\rbrace$, we  generated the following instances:
\begin{itemize}
	\item $ins \in \lbrace 1, \ldots, 10\rbrace$: 10 instances with all players (ten provinces);
	\item  $ins \in \lbrace 11, \ldots, 20\rbrace$: 10 instances with 3 players (Ontario, British Columbia \& Yukon and Alberta);
	\item  $ins \in \lbrace 21, \ldots, 30\rbrace$: 10 instances with 3 players (Ontario, Quebec and Prince Edward Island).
\end{itemize}

\paragraph{Algorithmic approach} For each instance, we compute \emph{(1)}  the $K$ best  solutions (game outcomes)\footnote{For few instances, there were less than $K$ distinct matchings.} and \emph{(2)} verify for each of them if it is an equilibrium, and, if yes, how close it is from a social optimum.

For {\proKEGcard} it is very likely the existence of multiple social optima, \ie, multiple matchings of maximum cardinality and we know that there must be at least one that is an equilibrium. Thus, for these games, we will compute at most $K=1,000$ random social optimal solutions (matchings of maximum cardinality) in task \emph{(1)}.  In order to  uniformly generate  $1,000$ matchings of maximum cardinality, we use the dynamic programming method described in Appendix~\ref{app:uniform_gen}.
On the other hand, for {\proKEGwei}, the mutiplicity of maximum-weighted matchings is less likely, no {\SWE} is guaranteed to exist and thus, we use the framework by \citet{Lawler1972} to perform task  \emph{(1)} with $K=1,000$. In other words, for {\proKEGwei} we cannot restrict our attention to maximum-weighted matchings.

Let us now concentrate on task  \emph{(2)}. Verifying if a matching $M$ is a Nash equilibrium implies verifying that any player $p$'s best response to $M^{-p}$ does not lead to an increase in the player's utility in comparison with $M \cap E^p$. Therefore, to execute task \emph{(2)} we will compute each player's best response; in this context, we concentrate on the player's best responses in  {\proKEGwei}, since it generalizes the best response computations under {\proKEGcard} ({\proKEGcard} is the special case of {\proKEGwei}  with  weights equal to 1).  Unless $P = NP$, a very unlikely event, Theorem~\ref{THM:COMPLEXITY} implies that there is no polynomial time algorithm to compute a player best response and thus, to verify if a matching is an equilibrium. Nevertheless, we were able to build an exact algorithm for this purpose that in practice, \ie, for the instances created by our generator, runs in seconds. Furthermore, in order to compute best responses, $\mathcal{A}$ must be known. In this way, for each maximum (weighted) matching, if there is a $\mathcal{A}$ making that matching an equilibrium, then it will be that one used in the best response computation; note that this will bias the analysis in the direction of increasing the likelihood of having equilibria.\footnote{See~\ref{app:Best_responses} for details on the implementation.} 

Our algorithmic implementations were in Python 3.7.1, except the algorithm of Appendix~\ref{app:uniform_gen} that was coded in Julia 1.1.0; the latter was the most time consuming part of our experiments and thus, it benefited from the use of a compiled language; nevertheless, we enforced a time limit of 2 hours for it. In case the time limit is reached, the corresponding entry in the tables indicates ``tl''. All the optimization problems have been solved with IBM CPLEX 12.9.0. The
experiments were conducted on an
Intel Xeon E5-2637 processor clocked at 3.50GHz
and equipped with 8 GB RAM, using a single core.

We close the section by noting that we believe that the ability of practically computing best responses for {\proKEGwei} has a potential wider applicability for game theory research on {\proKEP}, thus making the algorithm we designed and its associated (publicly available) implementation relevant contributions \emph{per se}.\footnote{In this concern, although improvements on the running times of the computational implementations are possible, we were happy with the performance on our dataset and we left them to future research.}

\subsection{Computational simulations}\label{subsec:computational_simulations}


\paragraph{{\proKEGcard}}

In our result Tables~\ref{Table:Provinces_game_small_ins}, \ref{Table:Provinces_game_small_ins_2}, \ref{Table:Provinces_game_small_ins_3}, \ref{Table:Provinces_game_small_ins_4}, \ref{Table:Provinces_game_small_ins_5}, \ref{Table:Provinces_game_small_ins_6}, $\vert M \vert_{max}$  denotes the number of edges of the maximum cardinality matching, $\# \vert M \vert_{max}$ denotes the number of distinct matchings of size $\vert M \vert_{max}$, \% \#-NE is the percentage of equilibria in the sample of $1,000$ matchings of maximum cardinality, $ IMP_{max}$  denotes the maximum increase on utility obtained by a province (or set of provinces) in an equilibrium in comparison with the maximum utility  when restricted its individual pool,  $Players_{max} $  denotes the set of provinces reaching that maximum increase,  $ IMP_{min}$ denotes the minimum increase on utility obtained by a province (or set of provinces) in an equilibrium in comparison with the maximum utility when restricted to its individual pool,  $Players_{min} $  denotes the set of provinces reaching that minimum increase,  {\IndA}$_{high}$  and  {\IndA}$_{low}$ are the maximum and minimum number of international transplants in a computed equilibrium, respectively, $time_{gen}$ is the  time in seconds to uniformally generate $1,000$ matchings of maximum cardinality and $time_{NE}$ is the average time to verify if a sampled matching is an equilibrium.

Tables ~\ref{Table:Provinces_game_small_ins}, \ref{Table:Provinces_game_small_ins_2} correspond to {\proKEGcard} with all the 10 provinces as players. As expected the instances associated with 2009 have matchings of higher cardinality and significantly higher number of maximum matchings  in comparison with 2013. Nevertheless, the density of the graphs does not seem to reflect any trend on the prevalence of {\SWE}: for the instances with 30 vertices, 2013 has a  very high percentage of equilibria, while for the instances of size 40, there is no different trend between 2009 and 2013 on the percentage of equilibria; for the instances of size 50 such comparison cannot be made, as the algorithms on 2009 instances exceeded the time limit to sample equilibria. Under an equilibrium, by definition, no player is worse than if she did not participate in the game, thus the column $IMP_{min}$ is always non-negative. In fact, in our computational results it takes always the value 0; this means that there were {\SWE} under which there was at least a player that did not improve her utility. For example, Prince Edward Island appears always in this column which can be explained by the fact that this is the smallest player and thus, less likely to contribute to the exchanges. In opposition, Ontario, the largest player, is the one that appears the less in the column  $Players_{min}$. Empirically this indicates that on {\SWE} the players that contribute the most to the game are compensated. Let us now concentrate on the column $IMP_{max}$.  In general, the maximum increases on utility are strictly positive. In column $Players_{max}$, except for instances with very few matchings of maximum cardinality (and thus, {\SWE}), the smallest provinces, Newfoundland and Labrador, Prince Edward Island, New Brunswick and Nova Scotia, never get the highest increase under a {\SWE}. Again, this seems to enforce that {\SWE} are the ``fair'' game outcome to concentrate as they benefit players according to their degree of participation. Finally, note that the average times in $time_{NE}$ show that this step can be done efficiently with our algorithm while the bottleneck is the sampling of matchings of maximum cardinality. The results of these tables support the importance of using the concept of equilibria to reflect strategical behaviors on the Kidney Exchange outcome. As such, they provide a significant, practical step in the understanding of the relevant questions of {\proKEP}.  Although, for the majority of instances the percentage of {\SWE} computed was higher than \%50 (recall that we are using $\mathcal{A}$ such that the likelihood of a matchings being an equilibrium is maximized), it is not guaranteed that simply using a maximum matching one will have an equilibrium.


\begin{landscape}
	\setlength{\tabcolsep}{6pt}
	\begin{table}[t]
		\hspace*{-8cm}
		\tiny
		\centering
		\begin{tabularx}{\textwidth}{rrrr|rrrrrrrrr}
			year & $\vert V \vert$ & ins & $\vert E \vert$ &
			$\vert M \vert_{max}$ & $\# \vert M \vert_{max}$   & \% \#-NE & $ IMP_{max}$  &  $Players_{max} $ &   $ IMP_{min}$ &  $Players_{min} $&
			$time_{gen}$ & $time_{NE}$   \\
			\cline{1-13}
			2009  &
			30 & 1 & 32 & 8 & 37 & 81.90 & 3 & MB QC & 0 & AB SK NS NB PE NL & 2.68 & 0.07\\
			& & 2 & 62 & 9 & 1980 & 58.20 & 2 & BCYT  & 0 & BCYT  AB SK MB QC NS NB PE NL & 32.77 & 0.37\\
			& & 3 & 30 & 8 & 7 & 100.00 & 4 & BCYT  & 0 & NS NB PE & 1.60 & 0.06\\
			& & 4 & 40 & 7 & 216 & 3.40 & 2 & AB & 0 & SK MB ON QC NS NB PE NL & 3.05 & 0.03\\
			& & 5 & 40 & 7 & 1140 & 48.20 & 4 & AB & 0 & SK MB NS NB PE NL & 3.74 & 0.07\\
			& & 6 & 20 & 5 & 28 & 35.10 & 2 & MB ON & 0 & BCYT  AB SK QC NS NB PE NL & 0.37 & 0.03\\
			& & 7 & 31 & 6 & 395 & 42.00 & 2 & BCYT  AB SK ON & 0 & MB QC NS NB PE NL & 1.57 & 0.06\\
			& & 8 & 29 & 8 & 2 & 50.10 & 2 & BCYT  MB ON QC & 0 & NS NB PE NL & 0.95 & 0.04\\
			& & 9 & 51 & 10 & 118 & 78.90 & 3 & BCYT  & 0 & SK MB ON NS NB PE NL & 33.22 & 0.06\\
			& & 10 & 16 & 3 & 75 & 95.30 & 2 & BCYT  ON & 0 & AB SK MB ON NS NB PE NL & 0.19 & 0.05\\
			2013 &
			30 & 1 & 1 & 1 & 1 & 100.00 & 1 & SK ON & 0 & BCYT  AB MB QC NS NB PE NL & 0.21 & 0.03\\
			& & 2 & 8 & 4 & 10 & 100.00 & 3 & BCYT  ON & 0 & MB NS NB PE NL & 0.17 & 0.04\\
			& & 3 & 9 & 3 & 4 & 100.00 & 1 & BCYT  AB MB ON QC NS NB & 0 & BCYT  SK ON QC NS PE NL & 0.31 & 0.04\\
			& & 4 & 19 & 5 & 4 & 100.00 & 2 & ON & 0 & SK MB QC NS NB PE & 0.20 & 0.04\\
			& & 5 & 25 & 5 & 135 & 18.50 & 3 & QC & 0 & BCYT  SK MB ON NS NB PE NL & 0.35 & 0.03\\
			& & 6 & 3 & 1 & 3 & 100.00 & 1 & BCYT  MB ON & 0 & AB SK MB ON QC NS NB PE NL & 0.07 & 0.04\\
			& & 7 & 8 & 3 & 6 & 83.90 & 2 & ON & 0 & AB SK MB NS NB PE NL & 0.07 & 0.03\\
			& & 8 & 4 & 2 & 3 & 33.20 & 1 & BCYT  ON & 0 & AB SK MB QC NS NB PE NL & 0.08 & 0.02\\
			& & 9 & 1 & 1 & 1 & 100.00 & 0 & BCYT  AB SK MB ON QC NS NB PE NL & 0 & BCYT  AB SK MB ON QC NS NB PE NL & 0.28 & 0.03\\
			& & 10 & 3 & 2 & 2 & 100.00 & 1 & BCYT  ON & 0 & AB SK MB QC NS NB PE NL & 0.07 & 0.04\\
			2009 &
			40 & 1 & 63 & 10 & 18720 & 47.60 & 6 & BCYT  & 0 & NS NB PE NL & 258.07 & 0.15\\
			& & 2 & 12 & 6 & 2 & 100.00 & 3 & ON & 0 & SK MB QC NS PE NL & 0.23 & 0.04\\
			& & 3 & 31 & 6 & 991 & 22.50 & 3 & BCYT  QC & 0 & AB SK MB ON QC NS PE NL & 1.94 & 0.05\\
			& & 4 & 61 & 10 & 1512 & 77.40 & 4 & BCYT  & 0 & AB SK ON PE NL & 279.68 & 2.06\\
			& & 5 & 56 & 13 & 1736 & 68.90 & 6 & ON & 0 & SK QC NS PE NL & 613.21 & 0.11\\
			& & 6 & 24 & 4 & 111 & 10.20 & 2 & BCYT  ON & 0 & AB SK MB ON QC NS NB PE NL & 0.24 & 0.02\\
			& & 7 & 76 & 10 & 177270 & 29.20 & 4 & BCYT  AB & 0 & SK MB ON QC NS NB PE NL & 342.42 & 0.20\\
			& & 8 & 11 & 3 & 14 & 100.00 & 2 & AB ON & 0 & BCYT  SK MB ON NS NB PE NL & 0.17 & 0.04\\
			& & 9 & 74 & 10 & 81472 & 26.00 & 3 & AB ON & 0 & BCYT  SK MB QC NS NB PE NL & 1032.54 & 0.33\\
			& & 10 & 50 & 10 & 1470 & 1.00 & 2 & BCYT  ON & 0 & AB SK MB QC NS NB PE NL & 70.88 & 0.03\\
			2013  &
			40 & 1 & 11 & 5 & 5 & 39.90 & 1 & MB ON & 0 & BCYT  AB SK QC NS NB PE NL & 0.09 & 0.03\\
			& & 2 & 6 & 3 & 3 & 69.50 & 2 & BCYT  & 0 & AB MB QC NS NB PE NL & 0.09 & 0.04\\
			& & 3 & 35 & 6 & 208 & 67.70 & 3 & ON & 0 & BCYT  AB SK PE NL & 1.32 & 0.15\\
			& & 4 & 0 & 0 & 1  & 100.00  & 0 &  BCYT  AB SK MB ON QC NS NB PE NL & 0  &  BCYT  AB SK MB ON QC NS NB PE NL  & 0.00  & 0.00\\
			& & 5 & 20 & 4 & 10 & 30.50 & 2 & BCYT  & 0 & SK ON QC NS NB PE NL & 0.13 & 0.03\\
			& & 6 & 34 & 6 & 356 & 78.80 & 3 & BCYT  ON & 0 & AB SK QC NS NB PE NL & 1.15 & 0.07\\
			& & 7 & 3 & 1 & 3 & 35.30 & 0 & BCYT  AB SK MB ON QC NS NB PE NL & 0 & BCYT  AB SK MB ON QC NS NB PE NL & 0.16 & 0.03\\
			& & 8 & 35 & 7 & 1108 & 18.00 & 2 & BCYT  ON & 0 & AB MB QC NB PE NL & 1.63 & 0.05\\
			& & 9 & 2 & 1 & 2 & 47.10 & 0 & BCYT  AB SK MB ON QC NS NB PE NL & 0 & BCYT  AB SK MB ON QC NS NB PE NL & 0.43 & 0.02\\
			& & 10 & 31 & 8 & 13 & 100.00 & 4 & ON & 0 & SK NS NB PE NL & 1.16 & 0.07\\
		\end{tabularx}
		\caption{Results for the  \#-KEG with the Canadian KPD instances.}
		\label{Table:Provinces_game_small_ins}
	\end{table}
	
\end{landscape}

\begin{landscape}
	\setlength{\tabcolsep}{6pt}
	\begin{table}[t]
		\hspace*{-5cm}
		\tiny
		\centering
		\begin{tabularx}{\textwidth}{rrrr|rrrrrrrrr}
			year & $\vert V \vert$ & ins & $\vert E \vert$ &
			$\vert M \vert_{max}$ & $\# \vert M \vert_{max}$   & \% \#-NE & $ IMP_{max}$  &  $Players_{max} $ &   $ IMP_{min}$ &  $Players_{min} $&
			$time_{gen}$ & $time_{NE}$   \\
			\cline{1-13}
			2009  &
			50 & 1 & 40 & 6 & 112 & 72.90 & 4 & BCYT  & 0 & SK NS NB PE NL & 1.83 & 0.10\\
			& & 2 & 128 & 15 &   &  &  & &  &  & tl & \\
			&  & 3 & 90 & 14 &   &  &  & &  &  & tl & \\
			& & 4 & 136 & 15 &   &  &  & &  &  & tl & \\
			& & 5 & 116 & 16 &   &  &  & &  &  & tl & \\
			& & 6 & 85 & 17 &                 &           &  &          &     &                                                & tl & \\
			&  & 7 & 99 & 10 & 303240 & 4.60 & 2 & ON & 0 & SK MB QC NS NB PE NL &  4870.62 & 21.72\\
			& &  8 & 99 & 16 &   &  &  & &  &  & tl & \\
			& &  9 & 74 & 12 & 872 & 66.70 & 4 & AB QC & 0 & MB NB PE NL & 1546.51 & 0.29\\
			& & 10 & 89 & 14 &   &  &  & &  &  & tl & \\
			2013  &
			50 & 1 & 8 & 3 & 2 & 100.00 & 1 & AB SK & 0 & BCYT  MB ON QC NS NB PE NL & 0.12 & 0.04\\
			&  & 2 & 23 & 5 & 30 & 100.00 & 2 & BCYT  QC & 0 & AB SK MB NS NB PE NL & 0.30 & 0.06\\
			&  & 3 & 24 & 4 & 287 & 26.40 & 2 & BCYT  & 0 & AB SK MB ON QC NS NB PE NL & 0.32 & 0.05\\
			&  & 4 & 29 & 7 & 38 & 73.40 & 3 & ON & 0 & AB SK QC NB PE NL & 1.31 & 0.05\\
			&  & 5 & 23 & 6 & 52 & 57.30 & 2 & BCYT  QC & 0 & AB SK MB ON NS PE NL & 0.39 & 0.05\\
			&  & 6 & 14 & 6 & 24 & 26.60 & 3 & AB & 0 & BCYT  SK MB NS NB PE NL & 0.16 & 0.04\\
			&  & 7 & 42 & 9 & 16 & 24.70 & 2 & BCYT  ON QC & 0 & MB NS NB PE NL & 9.36 & 0.04\\
			&  & 8 & 19 & 5 & 6 & 50.30 & 1 & BCYT  MB ON QC NS NB & 0 & AB SK PE NL & 0.16 & 0.03\\
			&  & 9 & 1 & 1 & 1 & 100.00 & 1 & BCYT  ON & 0 & AB SK MB QC NS NB PE NL & 0.22 & 0.04\\
			&  & 10 & 51 & 9 & 6736 & 79.00 & 6 & ON & 0 & SK QC NS NB PE NL & 64.80 & 0.20\\
		\end{tabularx}
		\caption{Continuation of Table~\ref{Table:Provinces_game_small_ins}. Results for the  \#-KEG with the Canadian KPD instances.}
		\label{Table:Provinces_game_small_ins_2}
	\end{table}
	
\end{landscape}

Tables~\ref{Table:Provinces_game_small_ins_3}, \ref{Table:Provinces_game_small_ins_4}, correspond to {\proKEGcard} with 3 players: Ontario, British Columbia \& Yukon and Alberta. The results in the tables indicate that when there are few matchings of maximum cardinality then the percentage of equilibria is very high, while in the reverse case the percentage of equilibria is small. For the cases of size 40 and 50, for year 2013, the prevalence of equilibria decreases in comparison with size 30 for the same year. Therefore, we can expect that in real instances, which typically have more than 100 pairs,  there is a low number of equilibria, probably explained by the fact that players also have large sets of feasible strategies. Thus, algorithms targeted to directly and efficiently compute social welfare equilibria become crucial to be able to compute all of them and, conversely, select the {\SWE} that optimizes a certain criterion. In this 3-player setting, columns $Players_{max}$ and $Players_{min}$, do not seem to present any trend, \ie, all three players seem to appear similar number of times in both of them, which might be explained by the fact that this is a more balanced setting, where players have similar contributions to the kidney exchange program.

Tables~\ref{Table:Provinces_game_small_ins_5}, \ref{Table:Provinces_game_small_ins_6} correspond to {\proKEGcard} with 3 players: Ontario, Quebec and Prince Edward Island. In these results, it is evident that Ontario and Quebec generally can take the most off a {\SWE} (see column $Players_{max}$), \ie, under an equilibrium, Ontario and Quebec are the players increasing the most their utilities in comparison to the ones achieved when they are by themselves.  On the other hand, for  Prince Edward Island there is always an equilibrium under which this player gets no advantage from participating in the market, except for instance $\vert V\vert=40$, year 2009, $ins=24$.

Among the results presented on this section, the number of players and their respective size does not seem to influence the percentage of equilibria.

\begin{landscape}
	\setlength{\tabcolsep}{6pt}
	\begin{table}[t]
		\hspace*{-5cm}
		\tiny
		\centering
		\begin{tabularx}{\textwidth}{rrrr|rrrrrrrrrrr}
			year & $\vert V \vert$ & ins & $\vert E \vert$ &
			$\vert M \vert_{max}$ & $\# \vert M \vert_{max}$   & \% \#-NE & $ IMP_{max}$  &  $Players_{max} $ &   $ IMP_{min}$ &  $Players_{min} $&
			{\IndA}$_{high}$ & {\IndA}$_{low}$ &
			$time_{gen}$ & $time_{NE}$   \\
			\cline{1-15}
			2009  &
			30 & 11 & 31 & 6 & 78 & 95.50 & 2 & AB ON & 1 & BCYT  AB ON & 10 & 8 & 1.32 & 0.09\\
			&  & 12 & 65 & 11 & 3 & 100.00 & 3 & BCYT  & 1 & ON & 16 & 16 & 114.31 & 0.07\\
			& & 13 & 21 & 6 & 65 & 4.30 & 1 & BCYT  ON & 0 & AB & 6 & 2 & 0.37 & 0.02\\
			& & 14 & 31 & 6 & 333 & 28.00 & 3 & BCYT  & 1 & ON & 8 & 6 & 1.30 & 0.04\\
			& &  15 & 28 & 5 & 127 & 86.80 & 2 & AB ON & 0 & BCYT  & 6 & 4 & 0.51 & 0.09\\
			& &  16 & 38 & 6 & 2001 & 32.40 & 3 & BCYT  AB & 0 & ON & 10 & 4 & 4.31 & 0.07\\
			& & 17 & 13 & 3 & 13 & 21.90 & 0 & BCYT  AB ON & 0 & BCYT  AB ON & 0 & 0 & 0.12 & 0.03\\
			& &  18 & 24 & 7 & 36 & 29.00 & 3 & BCYT  & 1 & AB & 12 & 8 & 0.47 & 0.04\\
			&  & 19 & 13 & 3 & 30 & 13.00 & 1 & BCYT  ON & 0 & AB & 2 & 2 & 0.19 & 0.03\\
			&  & 20 & 30 & 6 & 80 & 40.70 & 2 & BCYT  AB ON & 1 & BCYT  AB ON & 8 & 4 & 0.60 & 0.04\\
			2013 &
			30 & 11 & 6 & 3 & 2 & 100.00 & 2 & AB & 1 & BCYT  ON & 4 & 4 & 0.10 & 0.04\\
			&  & 12 & 2 & 1 & 2 & 100.00 & 0 & BCYT  AB ON & 0 & BCYT  AB ON & 0 & 0 & 0.33 & 0.03\\
			&  & 13 & 10 & 4 & 4 & 100.00 & 2 & AB ON & 0 & BCYT  & 8 & 4 & 0.23 & 0.04\\
			&  & 14 & 4 & 2 & 2 & 100.00 & 2 & ON & 1 & BCYT  AB & 4 & 4 & 0.11 & 0.03\\
			&  & 15 & 12 & 3 & 17 & 82.20 & 3 & AB & 0 & BCYT  & 6 & 4 & 0.18 & 0.04\\
			&  & 16 & 0 & 0 & 1  & 100.00  & 0  &  BCYT  AB ON & 0  &  BCYT  AB ON & 0 & 0 & 0.00 & 0.00 \\
			&  & 17 & 13 & 4 & 22 & 19.90 & 1 & BCYT  ON & 0 & AB & 2 & 2 & 0.13 & 0.02\\
			&  & 18 & 8 & 4 & 1 & 100.00 & 2 & ON & 0 & BCYT  AB & 4 & 4 & 0.13 & 0.04\\
			&  & 19 & 10 & 3 & 17 & 93.30 & 3 & BCYT  & 0 & AB ON & 6 & 4 & 0.09 & 0.04\\
			&  & 20 & 1 & 1 & 1 & 100.00 & 1 & BCYT  AB & 0 & ON & 2 & 2 & 0.60 & 0.03\\
			2009  &
			40 & 11 & 24 & 8 & 4 & 100.00 & 3 & BCYT  ON & 2 & AB & 10 & 10 & 0.59 & 0.05\\
			& &  12 & 13 & 6 & 12 & 17.00 & 2 & ON & 1 & BCYT  AB & 4 & 4 & 0.11 & 0.03\\
			&  & 13 & 59 & 10 & 34964 & 5.30 & 3 & BCYT  ON & 1 & AB & 14 & 6 & 406.44 & 0.10\\
			&  & 14 & 39 & 7 & 342 & 3.80 & 2 & BCYT  ON & 0 & AB & 4 & 4 & 3.83 & 0.04\\
			&  & 15 & 26 & 4 & 153 & 7.00 & 1 & BCYT  ON & 0 & AB & 2 & 2 & 0.34 & 0.02\\
			&  & 16 & 75 & 11 & 1334 & 2.40 & 3 & ON & 1 & AB & 8 & 8 & 1093.58 & 0.08\\
			&  & 17 & 33 & 8 & 218 & 51.20 & 3 & BCYT  AB & 1 & BCYT  ON & 12 & 8 & 2.53 & 0.07\\
			&  & 18 & 86 & 12 & 744 & 9.60 & 3 & AB & 1 & BCYT  & 16 & 8 & 1427.87 & 0.08\\
			&  & 19 & 52 & 10 & 128 & 30.10 & 6 & ON & 1 & BCYT  & 18 & 12 & 23.33 & 0.08\\
			&  & 20 & 92 & 12 & 4860 & 4.50 & 2 & BCYT  AB ON & 2 & BCYT  AB ON & 12 & 8 & 2293.17 & 0.12\\
			2013  &
			40 & 11 & 9 & 4 & 10 & 71.10 & 3 & BCYT  ON & 1 & AB & 8 & 6 & 0.07 & 0.04\\
			&  & 12 & 9 & 4 & 6 & 33.50 & 3 & BCYT  & 1 & AB & 6 & 6 & 0.07 & 0.02\\
			& & 13 & 15 & 3 & 54 & 3.80 & 0 & BCYT  AB ON & 0 & BCYT  AB ON & 0 & 0 & 0.08 & 0.01\\
			& & 14 & 22 & 5 & 16 & 51.00 & 2 & BCYT  ON & 0 & AB & 4 & 4 & 0.16 & 0.03\\
			& & 15 & 9 & 3 & 5 & 100.00 & 3 & AB ON & 1 & BCYT  & 6 & 6 & 0.06 & 0.04\\
			& &  16 & 32 & 4 & 735 & 4.70 & 0 & BCYT  AB ON & 0 & BCYT  AB ON & 4 & 0 & 0.48 & 0.02\\
			& & 17 & 6 & 2 & 9 & 36.10 & 1 & BCYT  ON & 0 & AB & 2 & 2 & 0.06 & 0.03\\
			& & 18 & 46 & 9 & 68 & 44.60 & 6 & BCYT  & 1 & AB & 16 & 14 & 10.19 & 0.05\\
			& & 19 & 5 & 2 & 6 & 16.70 & 0 & BCYT  AB ON & 0 & BCYT  AB ON & 0 & 0 & 0.05 & 0.03\\
			& &  20 & 13 & 5 & 36 & 64.90 & 2 & BCYT  AB ON & 0 & BCYT  & 6 & 4 & 0.09 & 0.04\\
		\end{tabularx}
		\caption{Results for the  \#-KEG with the Canadian KPD instances with the 3 largest players: Ontario, British Columbia \& Yukon and Alberta.}
		\label{Table:Provinces_game_small_ins_3}
	\end{table}
\end{landscape}
\begin{landscape}
	\setlength{\tabcolsep}{6pt}
	\begin{table}[t]
		\hspace*{-5cm}
		\tiny
		\centering
		\begin{tabularx}{\textwidth}{rrrr|rrrrrrrrrrr}
			year & $\vert V \vert$ & ins & $\vert E \vert$ &
			$\vert M \vert_{max}$ & $\# \vert M \vert_{max}$   & \% \#-NE & $ IMP_{max}$  &  $Players_{max} $ &   $ IMP_{min}$ &  $Players_{min} $&
			{\IndA}$_{high}$ & {\IndA}$_{low}$ &
			$time_{gen}$ & $time_{NE}$   \\
			\cline{1-15}
			2009  &
			50 & 11 & 47 & 9 & 117 & 24.00 & 4 & BCYT  & 1 & AB & 14 & 10 & 17.50 & 0.07\\
			& & 12 & 168 & 15 &   &  &  & &  & & & & tl & \\
			& & 13 & 103 & 12 &   &  &  & &  & & & & tl & \\
			& & 14 & 88 & 14 &   &  &  & &  & & & & tl & \\
			& & 15 & 48 & 8 & 758 & 11.50 & 4 & BCYT  & 1 & AB & 8 & 8 & 22.60 & 0.06\\
			& & 16 & 52 & 12 & 1188 & 10.50 & 5 & AB ON & 3 & BCYT  & 14 & 12 & 218.37 & 0.05\\
			& & 17 & 56 & 11 & 384 & 46.30 & 4 & ON & 1 & BCYT  AB & 18 & 10 & 200.67 & 0.18\\
			& & 18 & 76 & 11 & 145292 & 11.10 & 2 & ON & 0 & BCYT  AB & 12 & 4 & 2290.89 & 0.34\\
			& & 19 & 77 & 11 & 18260 & 28.10 & 4 & ON & 2 & BCYT  AB & 16 & 8 & 3292.81 & 0.20\\
			&& 20 & 63 & 11 & 225 & 37.00 & 4 & ON & 1 & BCYT  AB & 16 & 8 & 328.63 & 0.09\\
			2013    &
			50 & 11 & 22 & 6 & 46 & 8.20 & 0 & BCYT  AB ON & 0 & BCYT  AB ON & 0 & 0 & 0.55 & 0.03\\
			& & 12 & 22 & 6 & 78 & 57.90 & 4 & AB ON & 1 & BCYT  & 10 & 8 & 0.33 & 0.04\\
			& & 13 & 23 & 6 & 35 & 49.10 & 4 & ON & 1 & BCYT  AB & 10 & 6 & 0.29 & 0.05\\
			& & 14 & 30 & 6 & 12 & 100.00 & 2 & AB ON & 0 & BCYT  & 10 & 8 & 0.56 & 0.09\\
			& & 15 & 20 & 6 & 26 & 41.40 & 2 & BCYT  AB ON & 1 & BCYT  AB ON & 6 & 4 & 0.26 & 0.04\\
			& & 16 & 12 & 4 & 5 & 79.40 & 2 & BCYT  & 0 & AB ON & 6 & 4 & 0.09 & 0.04\\
			& & 17 & 52 & 7 & 445 & 47.70 & 3 & ON & 0 & AB & 12 & 6 & 11.13 & 0.28\\
			& & 18 & 14 & 4 & 15 & 34.60 & 2 & BCYT  ON & 0 & AB & 4 & 4 & 0.14 & 0.02\\
			& & 19 & 25 & 7 & 4 & 100.00 & 2 & AB & 1 & BCYT  ON & 10 & 8 & 0.41 & 0.05\\
			& & 20 & 10 & 3 & 4 & 49.90 & 0 & BCYT  AB ON & 0 & BCYT  AB ON & 0 & 0 & 0.07 & 0.03\\
		\end{tabularx}
		\caption{Continuation of Table~\ref{Table:Provinces_game_small_ins_3}. Results for the  \#-KEG with the Canadian KPD instances with  the 3 largest players: Ontario, British Columbia \& Yukon and Alberta.}
		\label{Table:Provinces_game_small_ins_4}
	\end{table}
\end{landscape}


\begin{landscape}
	\setlength{\tabcolsep}{6pt}
	\begin{table}[t]
		\hspace*{-5cm}
		\tiny
		\centering
		\begin{tabularx}{\textwidth}{rrrr|rrrrrrrrrrr}
			year & $\vert V \vert$ & ins & $\vert E \vert$ &
			$\vert M \vert_{max}$ & $\# \vert M \vert_{max}$   & \% \#-NE & $ IMP_{max}$  &  $Players_{max} $ &   $ IMP_{min}$ &  $Players_{min} $&
			{\IndA}$_{high}$ & {\IndA}$_{low}$ &
			$time_{gen}$ & $time_{NE}$   \\
			\cline{1-15}
			2009  &
			30 & 21 & 32 & 8 & 37 & 85.60 & 3 & QC & 0 & PE & 14 & 10 & 2.41 & 0.07\\
			& & 22 & 62 & 9 & 1980 & 59.00 & 1 & ON QC & 0 & QC PE & 16 & 4 & 32.54 & 0.37\\
			& & 23 & 30 & 8 & 7 & 100.00 & 3 & QC & 0 & PE & 14 & 14 & 1.50 & 0.06\\
			& & 24 & 40 & 7 & 216 & 2.90 & 1 & QC & 0 & ON QC PE & 6 & 6 & 3.11 & 0.03\\
			& & 25 & 40 & 7 & 1140 & 47.70 & 2 & ON QC & 0 & PE & 14 & 8 & 3.73 & 0.07\\
			& & 26 & 20 & 5 & 28 & 33.80 & 2 & ON & 0 & QC PE & 6 & 4 & 0.29 & 0.04\\
			& & 27 & 31 & 6 & 395 & 39.70 & 2 & ON & 0 & QC PE & 12 & 6 & 1.38 & 0.06\\
			& & 28 & 29 & 8 & 2 & 51.70 & 2 & ON QC & 0 & PE & 10 & 10 & 0.82 & 0.04\\
			& & 29 & 51 & 10 & 118 & 77.80 & 2 & QC & 0 & ON PE & 18 & 14 & 33.40 & 0.06\\
			& & 30 & 16 & 3 & 75 & 94.60 & 2 & ON & 0 & ON PE & 6 & 4 & 0.08 & 0.05\\
			2013  &
			30 & 21 & 1 & 1 & 1 & 100.00 & 1 & ON & 0 & QC PE & 2 & 2 & 0.17 & 0.03\\
			& &22 & 8 & 4 & 10 & 100.00 & 3 & ON & 0 & PE & 8 & 8 & 0.10 & 0.04\\
			& &23 & 9 & 3 & 4 & 100.00 & 1 & ON QC & 0 & ON QC PE & 4 & 4 & 0.08 & 0.04\\
			& &24 & 19 & 5 & 4 & 100.00 & 2 & ON & 0 & QC PE & 8 & 8 & 0.38 & 0.04\\
			& &25 & 25 & 5 & 135 & 19.00 & 3 & QC & 0 & ON PE & 6 & 6 & 0.37 & 0.03\\
			& &26 & 3 & 1 & 3 & 100.00 & 1 & ON & 0 & ON QC PE & 2 & 2 & 0.06 & 0.04\\
			& &27 & 8 & 3 & 6 & 83.40 & 2 & ON & 0 & PE & 4 & 4 & 0.15 & 0.03\\
			& &28 & 4 & 2 & 3 & 31.40 & 1 & ON & 0 & QC PE & 2 & 2 & 0.15 & 0.02\\
			& &29 & 1 & 1 & 1 & 100.00 & 0 & ON QC PE & 0 & ON QC PE & 0 & 0 & 0.17 & 0.03\\
			& &30 & 3 & 2 & 2 & 100.00 & 1 & ON & 0 & QC PE & 2 & 2 & 0.16 & 0.03\\
			2009  &
			40 & 21 & 32 & 6 & 592 & 7.50 & 1 & ON & 0 & QC PE & 6 & 2 & 1.67 & 0.04\\
			& & 22 & 79 & 12 & 3696 & 19.60 & 2 & QC & 0 & PE & 18 & 14 & 4917.50 & 0.08\\
			& & 23 & 8 & 2 & 15 & 28.80 & 0 & ON QC PE & 0 & ON QC PE & 0 & 0 & 0.10 & 0.03\\
			& & 24 & 40 & 9 & 200 & 84.40 & 4 & ON & 1 & PE & 14 & 10 & 7.75 & 0.13\\
			& & 25 & 67 & 9 & 9678 & 62.60 & 4 & QC & 0 & PE & 18 & 14 & 503.59 & 0.27\\
			& & 26 & 33 & 6 & 611 & 19.00 & 2 & ON QC & 0 & QC PE & 8 & 4 & 2.01 & 0.04\\
			& & 27 & 55 & 9 & 5792 & 40.70 & 5 & ON & 0 & PE & 16 & 12 & 102.98 & 0.08\\
			& & 28 & 48 & 10 & 232 & 74.60 & 4 & ON & 0 & PE & 14 & 10 & 108.22 & 0.11\\
			& & 29 & 61 & 10 & 84 & 100.00 & 5 & QC & 0 & ON PE & 18 & 18 & 109.50 & 0.18\\
			& & 30 & 37 & 9 & 421 & 88.40 & 3 & QC & 0 & PE & 18 & 14 & 6.68 & 0.08\\
			2013 &
			40 & 21 & 27 & 6 & 72 & 100.00 & 5 & ON & 0 & PE & 12 & 12 & 0.55 & 0.06\\
			& & 22 & 18 & 5 & 36 & 100.00 & 3 & ON & 0 & QC PE & 6 & 6 & 0.22 & 0.08\\
			& & 23 & 15 & 5 & 10 & 100.00 & 3 & ON & 0 & QC PE & 10 & 8 & 0.11 & 0.05\\
			& & 24 & 7 & 2 & 5 & 100.00 & 1 & ON PE & 0 & ON QC PE & 2 & 2 & 0.11 & 0.04\\
			& & 25 & 16 & 4 & 6 & 100.00 & 3 & ON & 0 & PE & 8 & 8 & 0.15 & 0.05\\
			& & 26 & 6 & 3 & 4 & 100.00 & 2 & ON & 0 & QC PE & 6 & 6 & 0.17 & 0.04\\
			& & 27 & 50 & 7 & 4046 & 7.00 & 3 & QC & 0 & ON PE & 12 & 4 & 14.13 & 0.04\\
			& & 28 & 34 & 8 & 54 & 27.30 & 2 & ON QC & 0 & QC PE & 12 & 10 & 2.96 & 0.06\\
			& & 29 & 14 & 3 & 28 & 6.80 & 1 & ON QC & 0 & QC PE & 2 & 2 & 0.17 & 0.03\\
			& & 30 & 7 & 4 & 3 & 100.00 & 3 & ON & 0 & PE & 8 & 8 & 0.12 & 0.04\\
		\end{tabularx}
		\caption{Results for the  \#-KEG with the Canadian KPD instances with the largest player, a medium player and the smallest player: Ontario, Quebec and Prince Edward Island.}
		\label{Table:Provinces_game_small_ins_5}
	\end{table}
	
\end{landscape}

\begin{landscape}
	\setlength{\tabcolsep}{6pt}
	\begin{table}[t]
		\hspace*{-5cm}
		\tiny
		\centering
		\begin{tabularx}{\textwidth}{rrrr|rrrrrrrrrrr}
			year & $\vert V \vert$ & ins & $\vert E \vert$ &
			$\vert M \vert_{max}$ & $\# \vert M \vert_{max}$   & \% \#-NE & $ IMP_{max}$  &  $Players_{max} $ &   $ IMP_{min}$ &  $Players_{min} $&
			{\IndA}$_{high}$ & {\IndA}$_{low}$ &
			$time_{gen}$ & $time_{NE}$   \\
			\cline{1-15}
			2009  &
			50 & 21 & 101 & 13 &  &  &  & &  & & & & tl & \\
			& & 22 & 85 & 11 &   &  &  & &  & & & & tl & \\
			&  & 23 & 82 & 10 & 128484 & 2.30 & 2 & ON & 0 & QC PE & 14 & 8 & 6105.72 & 1.97\\
			& & 24 & 100 & 13 &   &  &  & &  & & & & tl & \\
			& & 25 & 55 & 11 & 980 & 27.20 & 3 & ON & 0 & QC PE & 16 & 10 & 106.15 & 0.11\\
			& & 26 & 29 & 8 & 192 & 4.40 & 2 & ON & 0 & PE & 10 & 6 & 4.21 & 0.02\\
			& & 27 & 94 & 11 &   &  &  & &  & & & & tl & \\
			& & 28 & 47 & 9 & 786 & 34.40 & 5 & ON & 0 & QC PE & 14 & 10 & 38.16 & 0.09\\
			& & 29 & 53 & 9 & 6432 & 39.70 & 2 & ON & 0 & QC PE & 16 & 10 & 52.67 & 0.05\\
			& & 30 & 73 & 11 & 1368 & 6.30 & 2 & ON & 0 & PE & 12 & 10 & 899.55 & 0.07\\
			2013  &
			50 & 21 & 9 & 4 & 4 & 24.10 & 1 & ON QC & 0 & PE & 4 & 4 & 0.13 & 0.03\\
			& & 22 & 15 & 5 & 23 & 52.60 & 2 & ON & 0 & QC PE & 4 & 2 & 0.41 & 0.03\\
			& & 23 & 23 & 6 & 33 & 79.60 & 2 & ON QC & 0 & PE & 8 & 6 & 0.55 & 0.06\\
			& & 24 & 41 & 10 & 36 & 51.20 & 3 & QC & 0 & PE & 18 & 16 & 11.99 & 0.06\\
			& & 25 & 14 & 5 & 27 & 69.10 & 2 & ON & 0 & ON QC PE & 6 & 4 & 0.16 & 0.04\\
			& & 26 & 37 & 6 & 163 & 8.90 & 2 & ON & 0 & QC PE & 4 & 4 & 1.74 & 0.05\\
			& & 27 & 25 & 5 & 113 & 19.50 & 3 & ON & 0 & PE & 8 & 6 & 0.40 & 0.03\\
			& & 28 & 38 & 8 & 39 & 86.10 & 4 & QC & 0 & PE & 14 & 12 & 2.60 & 0.09\\
			& & 29 & 21 & 5 & 48 & 100.00 & 4 & QC & 0 & PE & 10 & 10 & 0.24 & 0.06\\
			& & 30 & 35 & 6 & 1861 & 39.90 & 3 & ON QC & 0 & ON QC PE & 12 & 6 & 3.32 & 0.06\\
		\end{tabularx}
		\caption{Continuation of Talbe~\ref{Table:Provinces_game_small_ins_5}. Results for the  \#-KEG with the Canadian KPD instances with the largest player, a medium player and the smallest player: Ontario, Quebec and Prince Edward Island.}
		\label{Table:Provinces_game_small_ins_6}
	\end{table}
	
\end{landscape}

\paragraph{{\proKEGwei}}

{\proKEGwei} is not guaranteed to have a social welfare equilibrium, \ie, a Nash equilibrium that is a matching of maximum weight. Thus, we concentrate on determining $\epsilon$-{\SWE}, \ie, the pure Nash equilibrium with highest social welfare value with $\epsilon$ being the ratio between this equilibrium welfare and the social optimum.
\begin{figure} \centering
	\includegraphics[scale=0.45]{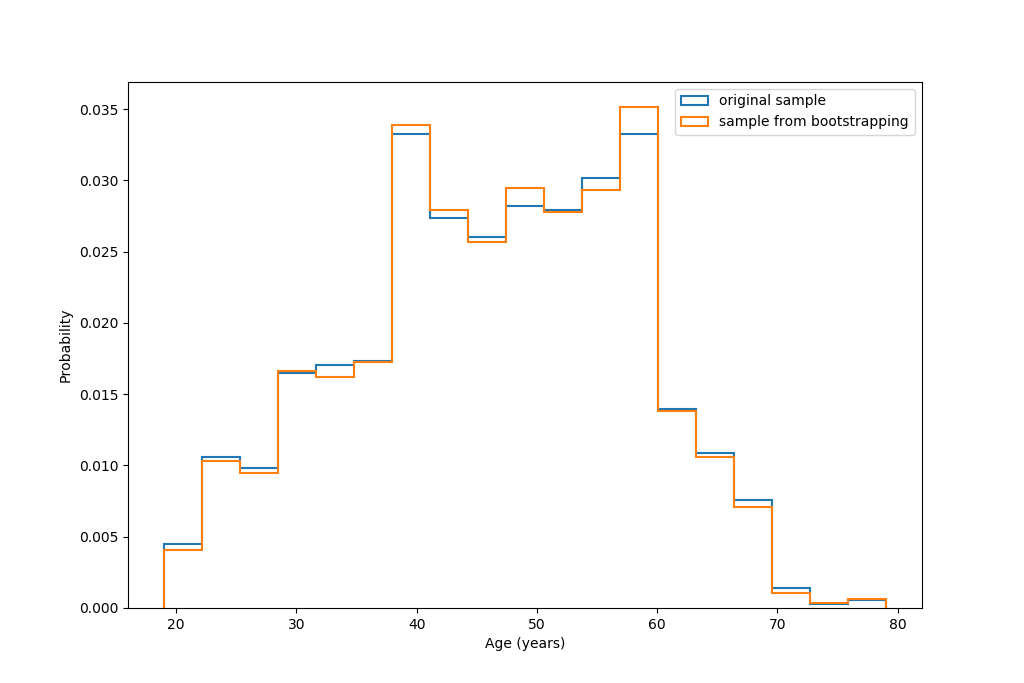}
	\caption{Histogram of ages on the Canadian Kidney Paired Donation Program}
	\label{fig:AgeDist}
\end{figure}
In the following results, the goal of the {\IndA} is to compute a maximum-weighted international matching, as referred to in the beginning of this section. Due to the lack of data on donors' characteristics  fundamental to compute quantities such as Risk Index for Living Donor Kidney Transplantation or to use machine learning models to predict graft survival, we generated the weights for this game according to the donors' age distribution on the Canadian Kidney Paired Donation Program; see Figure~\ref{fig:AgeDist}.  From the original data, we generated random ages with replacement (bootstrapping).
Indeed,  the literature on kidney transplantion has empirically demonstrated that donors' age is a significant predictor of graft survivability \citep{Terasaki1997,Berger2887}. In order to generate weights, we sample uniformly with replacement from the data points $x$ with donors' age; more precisely, let $x_u$ be an element of our sampling, then the weights are generated as follows: $$w_{uv}=\frac{-x_u+\max(x)}{\max(x)}, \quad \forall u \in N(v),$$
where $N(v)$ denotes the neighborhood of vertex $v$.
In this way, a younger donor has associated higher weights. Note that the formula above generates normalized weights, \ie,  between 0 and 1.

Tables~\ref{Table:Provinces_game_weigh_1_10}, \ref{Table:Provinces_game_weight_1_10_V50}, \ref{Table:Provinces_game_weight_3}, \ref{Table:Provinces_game_weight_4}, \ref{Table:Provinces_game_weight_5}, \ref{Table:Provinces_game_weight_6}, summarize our computational results concerning {\proKEGwei}. In these tables, we do not report the value of the maximum-weighted matching since it is not easy to understand its meaning as in the cardinality case. Therefore, the tables provide relative results. For instance,
$K$ are the best $K$ maximum-weighted matchings (which will be at most $1,000$), \% W-NE is the percentage of equilibria among the best $K$ maximum-weighted matchings,  $\epsilon$ is the ratio between the first computed weighted matching that is an equilibrium (iter-$th$ matching)  and the maximum-weighted matching; when an equilibrium is not found, the table reports the maximum value for $\epsilon$; note that $\epsilon$ close to 1 means that the first equilibrium computed is close to be a social optimum; iter is the iteration in Lawler algorithm \cite{Lawler1972}, in which the first equilibrium is determined;  $Players_{max}$ and $Players_{min}$ are defined as before; $IMP_{max}$ and $IMP_{min}$ are the maximum and minimum utility increases under equilibria divided by the total weight of the first equilibrium determined (so that this value is always smaller than 1); finally, $IA_{high}$ and $IA_{low}$ are the maximum and minimum fraction of international edges among the equilibria computed, respectively.


In Tables~\ref{Table:Provinces_game_weigh_1_10}, \ref{Table:Provinces_game_weight_1_10_V50}, all 10 provinces participate in the game.  Instances associated with year 2013 seem to be more likely to have a {\SWE}, which can be explained by the fact that they correspond to sparse graphs where players have less alternative strategies to deviate to. In these results, there is a significant number of instances for  which in the $K$ best weighted matchings there was no equilibria; in fact, \eg, for instance $\vert V \vert =30$, year 2009, $ins=2$, there are no pure equilibria among the first 1000 maximum-weighted matchings. Columns $Player_{max}$ and $Player_{min}$ seem to present the trend saw on the {\proKEGcard} results: the four smallest provinces typically appear in $Player_{min}$ while the largest appear in $Player_{max}$ . On the other hand, while in  {\proKEGcard} there was usually more than one province achieving the maximum improvement, here, in general,  there is  a single province attaining the maximum advantage.

Tables~\ref{Table:Provinces_game_weight_3} and \ref{Table:Provinces_game_weight_4}  report our computational results concerning {\proKEGwei} with the three largest provinces. In a majority of instances, the percentage of equilibria is small, in particular for all instances of size $\vert V \vert =50$. This might be explained by the fact that all players own a significant percentage of the pairs in the game and thus, they have more feasible strategies, which reflects to incentives for deviating from a $k$-th best weighted matching. It is interesting to note, that among all players, the largest one, Ontario, appears few times in the column $Players_{min}$ meaning that these $K$ best weighted matchings seem to fairly compensate the contribution of this player. Finally, observe that in most of the instances,  {\IndA}  establishes a significant benefit for  international exchanges.

Tables~\ref{Table:Provinces_game_weight_5}  and \ref{Table:Provinces_game_weight_6}  report our computational results concerning {\proKEGwei} with a large, a medium and a small player.  Most of these instances have a $\epsilon$-{\SWE} with $\epsilon>0.81$.  Combining this with the fact that the first  $\epsilon$-{\SWE} is generally determined before the $\frac{K}{2}$-th best weighted matching, it can be argued that these  $\epsilon$-{\SWE}  can be reasonable solutions for the game as the social optimum is not significantly sacrificed. Prince Edward Island is almost always present in the column $Players_{min}$, while in the column $Players_{max}$ Ontario appears more often than Quebec and  Prince Edward Island almost does not appear. Once again, this seems to indicate that equilibria provide fair benefits to the participating players.

\begin{landscape}
	\setlength{\tabcolsep}{6pt}
	\begin{table}[t]
		\hspace*{-7cm}
		\tiny
		\centering
		\begin{tabularx}{\textwidth}{rrrr|rrrrrrrrrrrrr}
			year & $\vert V \vert$ & ins & $\vert E \vert$ &
			$W(M)_{max}$ & $K$   & \% W-NE & $\epsilon$ & iter &$ IMP_{max}$  &  $Players_{max} $ &   $ IMP_{min}$ &  $Players_{min} $&
			{\IndA}$_{high}$ & {\IndA}$_{low}$ &
			$time_{gen}$ & $time_{NE}$   \\
			\cline{1-17}
			2009    &   30   &   1   &   32   &   8.73   &   1000   &   0.20   &  0.99  &   4   & 0.21 & QC & 0.00 & AB SK NS NB PE NL & 0.66 & 0.54 &   29.06   &   0.03\\
			&     &   2   &   62   &   8.82   &   1000   &   0.00   &  $<$ 0.96  &     &     &     &      &     &       &       &   100.61   &   0.03\\
			&     &   3   &   30   &   9.39   &   1000   &   24.80   &  0.97  &   3   & 0.23 & QC & 0.00 & AB SK NS NB PE NL & 0.86 & 0.50 &   37.13   &   0.05\\
			&     &   4   &   40   &   7.82   &   1000   &   0.30   &  1.00  &   0   & 0.19 & AB & 0.00 & SK MB QC NS NB PE NL & 0.58 & 0.41 &   47.08   &   0.03\\
			&     &   5   &   40   &   8.80   &   1000   &   0.70   &  0.97  &   118   & 0.22 & BCYT  & 0.00 & SK MB NB PE NL & 0.68 & 0.63 &   58.32   &   0.04\\
			&     &   6   &   20   &   6.01   &   567   &   0.18   &  0.99  &   1   & 0.23 & MB & 0.00 & BCYT  AB SK QC NS NB PE NL & 0.39 & 0.39 &   7.10   &   0.03\\
			&     &   7   &   31   &   6.95   &   1000   &   0.20   &  0.95  &   20   & 0.17 & AB & -0.00 & QC & 0.60 & 0.46 &   34.73   &   0.02\\
			&     &   8   &   29   &   8.28   &   1000   &   2.40   &  0.98  &   1   & 0.15 & MB & 0.00 & AB SK QC NS NB PE NL & 0.65 & 0.31 &   30.62   &   0.03\\
			&     &   9   &   51   &   11.07   &   1000   &   0.20   &  0.94  &   164   & 0.09 & BCYT  & 0.00 & MB NS NB PE NL & 0.63 & 0.63 &   94.40   &   0.04\\
			&     &   10   &   16   &   3.41   &   157   &   36.31   &  1.00  &   0   & 0.37 & BCYT  & 0.00 & BCYT  AB SK MB QC NS NB PE NL & 1.00 & 0.25 &   1.69   &   0.04\\
			2013    &   30   &   1   &   1   &   0.78   &   1   &   100.00   &  1.00  &   0   & 0.58 & SK & 0.00 & BCYT  AB MB QC NS NB PE NL & 1.00 & 1.00 &   0.01   &   0.03\\
			&     &   2   &   8   &   4.51   &   65   &   100.00   &  1.00  &   0   & 0.52 & ON & 0.00 & all players & 1.00 & 0.14 &   0.31   &   0.04\\
			&     &   3   &   9   &   3.39   &   29   &   6.90   &  1.00  &   0   & 0.29 & ON & 0.00 & BCYT  AB SK QC NS NB PE NL & 0.73 & 0.44 &   0.25   &   0.03\\
			&     &   4   &   19   &   5.65   &   318   &   0.63   &  0.85  &   6   & 0.15 & BCYT  & 0.00 & SK MB ON QC NS NB PE NL & 0.39 & 0.22 &   2.82   &   0.01\\
			&     &   5   &   25   &   5.32   &   1000   &   0.30   &  0.99  &   2   & 0.28 & QC & 0.00 & BCYT  AB SK MB NS NB PE NL & 0.57 & 0.42 &   22.14   &   0.03\\
			&     &   6   &   3   &   1.22   &   3   &   100.00   &  1.00  &   0   & 0.61 & BCYT  & 0.00 & AB SK MB ON QC NS NB PE NL & 1.00 & 0.90 &   0.02   &   0.03\\
			&     &   7   &   8   &   3.88   &   29   &   6.90   &  1.00  &   0   & 0.27 & ON & 0.00 & SK MB QC NS NB PE NL & 1.00 & 0.63 &   0.17   &   0.01\\
			&     &   8   &   4   &   2.13   &   7   &   14.29   &  0.96  &   1   & 0.28 & ON & 0.00 & AB SK MB QC NS NB PE NL & 0.42 & 0.42 &   0.03   &   0.01\\
			&     &   9   &   1   &   1.09   &   1   &   100.00   &  1.00  &   0   & 0.00 & QC & 0.00 & BCYT  AB SK MB ON NS NB PE NL & 0.00 & 0.00 &   0.01   &   0.03\\
			&     &   10   &   3   &   1.81   &   5   &   60.00   &  1.00  &   0   & 0.37 & ON & 0.00 & all players & 0.66 & 0.00 &   0.02   &   0.03\\
			2009    &   40   &   1   &   63   &   10.78   &   1000   &   0.00   &  $<$ 0.97  &     &     &     &      &     &       &       &   83.05   &   0.04\\
			&     &   2   &   12   &   6.95   &   207   &   0.48   &  1.00  &   0   & 0.21 & BCYT  & 0.00 & SK MB QC NS PE NL & 0.83 & 0.83 &   1.85   &   0.03\\
			&     &   3   &   31   &   7.30   &   1000   &   1.80   &  1.00  &   0   & 0.18 & BCYT  & 0.00 & AB SK MB QC NS NB PE NL & 0.45 & 0.32 &   44.15   &   0.04\\
			&     &   4   &   61   &   11.22   &   1000   &   0.00   &  $<$ 0.93  &     &     &     &      &     &       &       &   68.99   &   0.05\\
			&     &   5   &   56   &   14.53   &   1000   &   0.00   &  $<$ 0.97  &     &     &     &      &     &       &       &   71.56   &   0.04\\
			&     &   6   &   24   &   4.02   &   491   &   0.20   &  0.91  &   13   & 0.38 & BCYT  & 0.00 & AB SK MB QC NS PE NL & 0.49 & 0.49 &   11.77   &   0.01\\
			&     &   7   &   76   &   10.91   &   1000   &   0.00   &  $<$ 1.00  &     &     &     &      &     &       &       &   113.62   &   0.04\\
			&     &   8   &   11   &   3.65   &   55   &   10.91   &  0.74  &   16   & 0.34 & QC & 0.00 & BCYT  AB SK MB ON NS NB PE NL & 0.35 & 0.25 &   0.69   &   0.03\\
			&     &   9   &   74   &   11.59   &   1000   &   0.00   &  $<$ 0.97  &     &     &     &      &     &       &       &   108.39   &   0.06\\
			&     &   10   &   50   &   11.03   &   1000   &   0.10   &  0.94  &   395   & 0.10 & ON & 0.00 & AB SK MB QC NS NB PE NL & 0.20 & 0.20 &   43.35   &   0.03\\
			2013   &   40   &   1   &   11   &   5.36   &   139   &   0.72   &  0.97  &   2   & 0.10 & MB & -0.00 & BCYT  & 0.16 & 0.16 &   1.09   &   0.01\\
			&     &   2   &   6   &   3.23   &   18   &   11.11   &  0.98  &   1   & 0.26 & ON & 0.00 & BCYT  SK MB QC NS NB PE NL & 0.75 & 0.42 &   0.10   &   0.03\\
			&     &   3   &   35   &   7.46   &   1000   &   0.70   &  0.93  &   80   & 0.27 & ON & 0.00 & BCYT  AB SK PE NL & 0.78 & 0.56 &   32.09   &   0.04\\
			&     &   4   &   0   &   0.00   &   1   &   100.00   &  1.00  &   0   & 0.00 & all players & 0.00 & all players & 0.00 & 0.00 &   0.00   &   0.03\\
			&     &   5   &   20   &   4.42   &   241   &   2.07   &  1.00  &   0   & 0.29 & BCYT  & -0.00 & ON & 0.57 & 0.26 &   3.74   &   0.02\\
			&     &   6   &   34   &   7.45   &   1000   &   0.50   &  0.94  &   37   & 0.27 & BCYT  & 0.00 & AB SK MB QC NB PE NL & 0.67 & 0.50 &   40.54   &   0.02\\
			&     &   7   &   3   &   1.50   &   3   &   33.33   &  0.81  &   2   & 0.00 & all players & 0.00 & all players & 0.00 & 0.00 &   0.02   &   0.03\\
			&     &   8   &   35   &   9.36   &   1000   &   3.50   &  0.98  &   36   & 0.16 & BCYT  & 0.00 & AB MB QC NB PE NL & 0.57 & 0.37 &   29.39   &   0.02\\
			&     &   9   &   2   &   1.42   &   2   &   50.00   &  1.00  &   0   & 0.00 & all players & 0.00 & all players & 0.00 & 0.00 &   0.01   &   0.02\\
			&     &   10   &   31   &   8.78   &   1000   &   1.90   &  1.00  &   0   & 0.33 & ON & 0.00 & BCYT  SK NS NB PE NL & 0.88 & 0.54 &   27.72   &   0.04\\
		\end{tabularx}
		\caption{Results for the  {\proKEGwei} with the Canadian KPD instances.}
		\label{Table:Provinces_game_weigh_1_10}
	\end{table}

\end{landscape}

\begin{landscape}
	\setlength{\tabcolsep}{6pt}
	\begin{table}[t]
		\hspace*{-7cm}
		\tiny
		\centering
		\begin{tabularx}{\textwidth}{rrrr|rrrrrrrrrrrrr}
			year & $\vert V \vert$ & ins & $\vert E \vert$ &
			$W(M)_{max}$ & $K$   & \% W-NE & $\epsilon$ & iter &$ IMP_{max}$  &  $Players_{max} $ &   $ IMP_{min}$ &  $Players_{min} $&
			{\IndA}$_{high}$ & {\IndA}$_{low}$ &
			$time_{gen}$ & $time_{NE}$   \\
			\cline{1-17}
			2009   &    50   &   1   &   40   &   7.02   &   1000   &   13.40   &  0.97  &   5   & 0.34 & BCYT  & 0.00 & AB SK MB QC NS NB PE NL & 0.63 & 0.41 &   50.17   &   0.06\\
			&     &   2   &   128   &   18.07   &   1000   &   0.00   &  $<$ 0.98  &     &     &     &      &     &       &       &   160.04   &   0.13\\
			&     &   3   &   90   &   17.33   &   1000   &   0.00   &  $<$ 0.98  &     &     &     &      &     &       &       &   141.41   &   0.14\\
			&     &   4   &   136   &   18.66   &   1000   &   0.00   &  $<$ 0.99  &     &     &     &      &     &       &       &   267.80   &   0.25\\
			&     &   5   &   116   &   19.56   &   1000   &   1.30   &  0.99  &   88   & 0.15 & ON & 0.00 & SK NS NB PE & 0.67 & 0.52 &   316.52   &   0.12\\
			&     &   6   &   85   &   18.40   &   1000   &   1.50   &  0.99  &   43   & 0.20 & ON & 0.00 & NS NB PE NL & 0.68 & 0.61 &   165.87   &   0.10\\
			&     &   7   &   99   &   12.85   &   1000   &   0.00   &  $<$ 1.00  &     &     &     &      &     &       &       &   221.68   &   1.01\\
			&     &   8   &   99   &   18.61   &   1000   &   0.20   &  0.98  &   12   & 0.26 & ON & 0.00 & SK NB PE & 0.62 & 0.49 &   165.25   &   0.06\\
			&     &   9   &   74   &   12.64   &   1000   &   4.20   &  0.99  &   43   & 0.17 & QC & 0.00 & MB NB PE NL & 0.76 & 0.58 &   127.31   &   0.12\\
			&     &   10   &   89   &   16.11   &   1000   &   6.90   &  1.00  &   4   & 0.15 & ON & 0.00 & SK MB NB PE NL & 0.75 & 0.49 &   118.91   &   0.19\\
			2013     &   50   &   1   &   8   &   3.58   &   21   &   19.05   &  1.00  &   0   & 0.26 & AB & 0.00 & all players & 0.50 & 0.00 &   0.17   &   0.03\\
			&     &   2   &   23   &   4.74   &   705   &   0.85   &  1.00  &   2   & 0.30 & ON & 0.00 & AB SK MB NS NB PE NL & 0.74 & 0.54 &   12.50   &   0.02\\
			&     &   3   &   24   &   4.69   &   892   &   0.34   &  1.00  &   0   & 0.13 & BCYT  & 0.00 & SK MB ON QC NS NB PE NL & 0.22 & 0.19 &   13.63   &   0.01\\
			&     &   4   &   29   &   6.93   &   1000   &   3.10   &  0.96  &   4   & 0.17 & ON & 0.00 & AB SK MB ON QC NB PE NL & 0.60 & 0.24 &   25.63   &   0.02\\
			&     &   5   &   23   &   7.47   &   1000   &   0.60   &  0.99  &   11   & 0.18 & BCYT  & 0.00 & AB SK MB NS NB PE NL & 0.67 & 0.36 &   18.27   &   0.02\\
			&     &   6   &   14   &   5.74   &   662   &   0.15   &  1.00  &   0   & 0.23 & ON & 0.00 & SK MB NS NB PE NL & 0.77 & 0.77 &   5.19   &   0.03\\
			&     &   7   &   42   &   10.25   &   1000   &   0.10   &  0.98  &   5   & 0.15 & ON & 0.00 & MB NS NB PE NL & 0.54 & 0.54 &   49.68   &   0.03\\
			&     &   8   &   19   &   4.44   &   437   &   0.46   &  0.82  &   31   & 0.15 & ON & 0.00 & BCYT  AB SK QC NS NB PE NL & 0.41 & 0.21 &   5.79   &   0.02\\
			&     &   9   &   1   &   1.11   &   1   &   100.00   &  1.00  &   0   & 0.57 & ON & 0.00 & AB SK MB QC NS NB PE NL & 1.00 & 1.00 &   0.01   &   0.04\\
			&     &   10   &   51   &   10.25   &   1000   &   0.00   &  $<$ 0.94  &     &     &     &      &     &       &       &   47.02   &   0.04\\
		\end{tabularx}
		\caption{Continuation of Table~\ref{Table:Provinces_game_weigh_1_10}. Results for the  {\proKEGwei} with the Canadian KPD instances.}
		\label{Table:Provinces_game_weight_1_10_V50}
	\end{table}
	
\end{landscape}


\begin{landscape}
	\setlength{\tabcolsep}{6pt}
	\begin{table}[t]
		\hspace*{-7cm}
		\tiny
		\centering
		\begin{tabularx}{\textwidth}{rrrr|rrrrrrrrrrrrr}
			year & $\vert V \vert$ & ins & $\vert E \vert$ &
			$W(M)_{max}$ & $K$   & \% W-NE & $\epsilon$ & iter &$ IMP_{max}$  &  $Players_{max} $ &   $ IMP_{min}$ &  $Players_{min} $&
			{\IndA}$_{high}$ & {\IndA}$_{low}$ &
			$time_{gen}$ & $time_{NE}$   \\
			\cline{1-17}
			2009     &   30   &   11   &   31   &   6.73   &   1000   &   0.50   &  0.88  &   82   & 0.11 & ON & 0.00 & AB & 0.28 & 0.11 &   27.13   &   0.02\\
			&     &   12   &   65   &   13.92   &   1000   &   0.00   &  $<$ 0.90  &     &     &     &      &     &       &       &   90.81   &   0.03\\
			&     &   13   &   21   &   6.93   &   1000   &   0.10   &  0.97  &   22   & 0.09 & BCYT  & 0.00 & AB & 0.17 & 0.17 &   14.13   &   0.01\\
			&     &   14   &   31   &   7.96   &   1000   &   0.90   &  0.94  &   40   & 0.22 & AB & 0.02 & ON & 0.44 & 0.29 &   31.14   &   0.02\\
			&     &   15   &   28   &   6.10   &   1000   &   0.00   &  $<$ 0.50  &     &     &     &      &     &       &       &   21.05   &   0.02\\
			&     &   16   &   38   &   8.11   &   1000   &   1.20   &  0.96  &   79   & 0.16 & BCYT  & 0.07 & ON & 0.50 & 0.29 &   45.68   &   0.03\\
			&     &   17   &   13   &   2.84   &   53   &   3.77   &  0.96  &   3   & 0.00 & AB ON & -0.00 & BCYT  & 0.00 & 0.00 &   0.72   &   0.02\\
			&     &   18   &   24   &   8.71   &   1000   &   0.10   &  0.84  &   175   & 0.14 & AB & 0.07 & BCYT  & 0.26 & 0.26 &   23.96   &   0.01\\
			&     &   19   &   13   &   4.27   &   85   &   1.18   &  0.85  &   15   & 0.23 & ON & 0.00 & AB & 0.33 & 0.33 &   0.94   &   0.02\\
			&     &   20   &   30   &   6.08   &   1000   &   0.20   &  1.00  &   1   & 0.15 & ON & -0.00 & BCYT  & 0.37 & 0.22 &   33.60   &   0.02\\
			2013    &   30   &   11   &   6   &   3.66   &   15   &   6.67   &  1.00  &   0   & 0.32 & AB & 0.09 & BCYT  & 0.65 & 0.65 &   0.10   &   0.01\\
			&     &   12   &   2   &   1.56   &   2   &   50.00   &  1.00  &   0   & 0.00 & BCYT  AB ON & 0.00 & BCYT  AB ON & 0.00 & 0.00 &   0.01   &   0.03\\
			&     &   13   &   10   &   5.13   &   69   &   8.70   &  1.00  &   1   & 0.31 & AB & 0.00 & AB ON & 0.40 & 0.00 &   0.46   &   0.02\\
			&     &   14   &   4   &   2.63   &   6   &   100.00   &  1.00  &   0   & 0.56 & ON & 0.00 & BCYT  AB & 1.00 & 0.36 &   0.03   &   0.03\\
			&     &   15   &   12   &   3.59   &   59   &   22.03   &  1.00  &   0   & 0.42 & ON & 0.00 & BCYT  AB & 0.90 & 0.22 &   0.65   &   0.03\\
			&     &   16   &   0   &   0.00   &   1   &   100.00   &  1.00  &   0   & 0.00 & all players & 0.00 & all players & 0.00 & 0.00 &   0.00   &   0.03\\
			&     &   17   &   13   &   5.48   &   144   &   0.69   &  1.00  &   0   & 0.14 & ON & 0.00 & AB & 0.24 & 0.24 &   0.89   &   0.01\\
			&     &   18   &   8   &   3.55   &   41   &   2.44   &  0.76  &   3   & 0.00 & BCYT  & 0.00 & AB ON & 0.00 & 0.00 &   0.25   &   0.01\\
			&     &   19   &   10   &   3.54   &   53   &   11.32   &  0.90  &   7   & 0.36 & BCYT  & 0.00 & AB & 0.60 & 0.25 &   0.47   &   0.03\\
			&     &   20   &   1   &   0.97   &   1   &   100.00   &  1.00  &   0   & 0.72 & BCYT  & 0.00 & ON & 1.00 & 1.00 &   0.01   &   0.03\\
			2009     &   40   &   11   &   24   &   9.61   &   1000   &   0.50   &  1.00  &   0   & 0.28 & BCYT  & 0.10 & AB & 0.59 & 0.47 &   21.21   &   0.02\\
			&     &   12   &   13   &   7.15   &   439   &   0.23   &  0.96  &   4   & 0.17 & ON & 0.05 & BCYT  & 0.34 & 0.34 &   3.37   &   0.01\\
			&     &   13   &   59   &   11.70   &   1000   &   0.20   &  0.95  &   660   & 0.11 & BCYT  & 0.07 & AB & 0.28 & 0.28 &   77.10   &   0.02\\
			&     &   14   &   39   &   8.58   &   1000   &   0.10   &  0.98  &   6   & 0.21 & BCYT  & 0.00 & AB & 0.31 & 0.31 &   46.15   &   0.02\\
			&     &   15   &   26   &   4.91   &   835   &   0.12   &  0.86  &   57   & 0.13 & BCYT  & 0.00 & AB & 0.22 & 0.22 &   14.28   &   0.01\\
			&     &   16   &   75   &   11.01   &   1000   &   0.10   &  0.93  &   764   & 0.07 & ON & 0.04 & AB & 0.23 & 0.23 &   104.58   &   0.04\\
			&     &   17   &   33   &   8.91   &   1000   &   0.20   &  0.92  &   68   & 0.14 & ON & 0.03 & AB & 0.29 & 0.21 &   41.93   &   0.03\\
			&     &   18   &   86   &   15.29   &   1000   &   1.00   &  0.95  &   890   & 0.17 & AB & 0.00 & ON & 0.50 & 0.42 &   107.80   &   0.08\\
			&     &   19   &   52   &   11.49   &   1000   &   0.60   &  0.98  &   5   & 0.26 & AB & 0.00 & BCYT  & 0.62 & 0.50 &   61.80   &   0.02\\
			&     &   20   &   92   &   13.29   &   1000   &   0.20   &  1.00  &   60   & 0.11 & BCYT  & 0.07 & ON & 0.39 & 0.39 &   148.90   &   0.07\\
			2013    &   40   &   11   &   9   &   4.83   &   71   &   14.08   &  1.00  &   0   & 0.43 & BCYT  & 0.00 & AB & 0.80 & 0.40 &   0.53   &   0.03\\
			&     &   12   &   9   &   4.32   &   63   &   3.17   &  0.97  &   2   & 0.43 & BCYT  & 0.14 & AB & 0.80 & 0.80 &   0.43   &   0.01\\
			&     &   13   &   15   &   3.72   &   128   &   0.78   &  1.00  &   2   & 0.00 & BCYT  AB ON & 0.00 & BCYT  AB ON & 0.00 & 0.00 &   1.52   &   0.01\\
			&     &   14   &   22   &   6.02   &   625   &   0.32   &  0.95  &   2   & 0.24 & BCYT  & 0.00 & AB & 0.36 & 0.36 &   9.86   &   0.01\\
			&     &   15   &   9   &   4.00   &   31   &   100.00   &  1.00  &   0   & 0.56 & ON & 0.00 & BCYT  AB ON & 1.00 & 0.21 &   0.20   &   0.04\\
			&     &   16   &   32   &   5.36   &   1000   &   0.10   &  0.88  &   228   & 0.00 & ON & 0.00 & BCYT  AB & 0.00 & 0.00 &   25.69   &   0.01\\
			&     &   17   &   6   &   2.76   &   15   &   20.00   &  0.86  &   3   & 0.23 & BCYT  & 0.00 & AB & 0.36 & 0.26 &   0.10   &   0.03\\
			&     &   18   &   46   &   10.22   &   1000   &   0.20   &  0.90  &   109   & 0.25 & BCYT  & 0.09 & AB & 0.56 & 0.56 &   66.48   &   0.03\\
			&     &   19   &   5   &   2.33   &   11   &   9.09   &  0.95  &   1   & 0.00 & BCYT  AB ON & 0.00 & BCYT  AB ON & 0.00 & 0.00 &   0.07   &   0.03\\
			&     &   20   &   13   &   6.37   &   335   &   1.79   &  0.99  &   1   & 0.24 & AB & 0.00 & BCYT  & 0.44 & 0.26 &   2.54   &   0.03\\
		\end{tabularx}
		\caption{Results for the {\proKEGwei} with the Canadian KPD instances with  the 3 largest players: Ontario, British Columbia \& Yukon and Alberta.}
		\label{Table:Provinces_game_weight_3}
	\end{table}
	
\end{landscape}


\begin{landscape}
	\setlength{\tabcolsep}{6pt}
	\begin{table}[t]
		\hspace*{-7cm}
		\tiny
		\centering
		\begin{tabularx}{\textwidth}{rrrr|rrrrrrrrrrrrr}
			year &$\vert V \vert$ & ins & $\vert E \vert$ &
			$W(M)_{max}$ & $K$   & \% W-NE & $\epsilon$ & iter &$ IMP_{max}$  &  $Players_{max} $ &   $ IMP_{min}$ &  $Players_{min} $&
			{\IndA}$_{high}$ & {\IndA}$_{low}$ &
			$time_{gen}$ & $time_{NE}$   \\
			\cline{1-17}
			2009    &   50   &   11   &   47   &   10.19   &   1000   &   0.30   &  0.99  &   20   & 0.20 & BCYT  & 0.08 & AB & 0.69 & 0.68 &   50.84   &   0.04\\
			&      &   12   &   168   &   19.28   &   1000   &   0.00   &  $<$ 1.00  &     &     &     &      &     &       &       &   319.22   &   0.13\\
			&     &   13   &   103   &   14.21   &   1000   &   0.00   &  $<$ 0.95  &     &     &     &      &     &       &       &   139.76   &   0.05\\
			&     &   14   &   88   &   16.28   &   1000   &   0.00   &  $<$ 0.98  &     &     &     &      &     &       &       &   166.19   &   0.05\\
			&     &   15   &   48   &   9.78   &   1000   &   0.10   &  0.95  &   111   & 0.25 & BCYT  & 0.04 & AB & 0.50 & 0.50 &   74.86   &   0.02\\
			&     &   16   &   52   &   13.31   &   1000   &   0.20   &  1.00  &   14   & 0.16 & BCYT  & 0.11 & ON & 0.44 & 0.44 &   60.27   &   0.03\\
			&     &   17   &   56   &   14.17   &   1000   &   0.00   &  $<$ 0.93  &     &     &     &      &     &       &       &   87.98   &   0.04\\
			&     &   18   &   76   &   14.31   &   1000   &   0.00   &  $<$ 0.97  &     &     &     &      &     &       &       &   102.77   &   0.02\\
			&     &   19   &   77   &   12.29   &   1000   &   0.00   &  $<$ 0.97  &     &     &     &      &     &       &       &   109.10   &   0.04\\
			&     &   20   &   63   &   12.06   &   1000   &   0.00   &  $<$ 0.92  &     &     &     &      &     &       &       &   79.00   &   0.03\\
			2013    &   50   &   11   &   22   &   6.87   &   1000   &   0.10   &  0.99  &   2   & 0.00 & BCYT  AB ON & 0.00 & BCYT  AB ON & 0.00 & 0.00 &   15.93   &   0.02\\
			&      &   12   &   22   &   7.16   &   1000   &   0.90   &  0.98  &   3   & 0.37 & ON & 0.07 & BCYT  & 0.82 & 0.62 &   12.26   &   0.02\\
			&     &   13   &   23   &   7.57   &   1000   &   0.20   &  0.94  &   14   & 0.26 & ON & 0.06 & BCYT  & 0.42 & 0.32 &   12.84   &   0.01\\
			&     &   14   &   30   &   7.71   &   1000   &   0.30   &  0.79  &   158   & 0.10 & AB & 0.00 & BCYT  & 0.14 & 0.10 &   24.21   &   0.02\\
			&     &   15   &   20   &   7.30   &   1000   &   0.20   &  0.99  &   1   & 0.17 & ON & 0.03 & BCYT  & 0.32 & 0.26 &   13.79   &   0.02\\
			&     &   16   &   12   &   5.18   &   81   &   1.23   &  0.77  &   13   & 0.00 & ON & 0.00 & BCYT  AB & 0.00 & 0.00 &   0.73   &   0.02\\
			&     &   17   &   52   &   6.77   &   1000   &   0.10   &  0.86  &   582   & 0.11 & ON & 0.00 & AB & 0.14 & 0.14 &   68.30   &   0.03\\
			&     &   18   &   14   &   4.59   &   127   &   0.79   &  0.98  &   2   & 0.27 & BCYT  & 0.00 & AB & 0.45 & 0.45 &   1.44   &   0.01\\
			&     &   19   &   25   &   7.07   &   1000   &   0.10   &  0.79  &   301   & 0.07 & ON & 0.00 & AB & 0.08 & 0.08 &   21.72   &   0.02\\
			&     &   20   &   10   &   3.48   &   34   &   5.88   &  1.00  &   0   & 0.05 & ON & 0.00 & BCYT  AB & 0.21 & 0.19 &   0.37   &   0.02\\
		\end{tabularx}
		\caption{Continuation of Table~\ref{Table:Provinces_game_weight_3}. Results for the {\proKEGwei} with the Canadian KPD instances with  the 3 largest players: Ontario, British Columbia \& Yukon and Alberta.}
		\label{Table:Provinces_game_weight_4}
	\end{table}
	
\end{landscape}

\begin{landscape}
	\setlength{\tabcolsep}{6pt}
	\begin{table}[t]
		\hspace*{-7cm}
		\tiny
		\centering
		\begin{tabularx}{\textwidth}{rrrr|rrrrrrrrrrrrr}
			year & $\vert V \vert$ & ins & $\vert E \vert$ &
			$W(M)_{max}$ & $K$   & \% W-NE & $\epsilon$ & iter & $ IMP_{max}$  &  $Players_{max} $ &   $ IMP_{min}$ &  $Players_{min} $&
			{\IndA}$_{high}$ & {\IndA}$_{low}$ &
			$time_{gen}$ & $time_{NE}$   \\
			\cline{1-17}
			2009    &   30   &   21   &   32   &   9.62   &   1000   &   0.10   &  0.81  &   539   & 0.13 & QC & 0.00 & PE & 0.37 & 0.37 &   27.19   &   0.02\\
			&     &   22   &   62   &   9.01   &   1000   &   1.90   &  1.00  &   4   & 0.08 & ON & 0.00 & QC PE & 0.62 & 0.18 &   89.24   &   0.08\\
			&     &   23   &   30   &   9.41   &   1000   &   5.70   &  0.92  &   18   & 0.26 & QC & 0.00 & PE & 0.60 & 0.44 &   36.98   &   0.04\\
			&     &   24   &   40   &   7.76   &   1000   &   0.20   &  0.88  &   249   & 0.08 & QC & 0.00 & ON QC PE & 0.18 & 0.13 &   46.54   &   0.03\\
			&     &   25   &   40   &   8.75   &   1000   &   0.70   &  0.97  &   39   & 0.16 & QC & 0.00 & PE & 0.57 & 0.47 &   55.56   &   0.02\\
			&     &   26   &   20   &   6.01   &   567   &   0.88   &  0.90  &   7   & 0.12 & QC & 0.00 & QC PE & 0.60 & 0.12 &   6.89   &   0.03\\
			&     &   27   &   31   &   6.66   &   1000   &   0.50   &  0.97  &   17   & 0.06 & ON & 0.00 & ON QC PE & 0.45 & 0.29 &   35.84   &   0.03\\
			&     &   28   &   29   &   9.40   &   1000   &   0.40   &  0.91  &   7   & 0.08 & ON & 0.00 & QC PE & 0.55 & 0.38 &   32.54   &   0.03\\
			&     &   29   &   51   &   11.44   &   1000   &   0.00   &  $<$ 0.91  &     &     &     &      &     &       &       &   79.86   &   0.03\\
			&     &   30   &   16   &   4.03   &   157   &   24.20   &  0.96  &   6   & 0.32 & ON & 0.00 & ON QC PE & 0.69 & 0.00 &   1.69   &   0.03\\
			2013   &   30   &   21   &   1   &   1.07   &   1   &   100.00   &  1.00  &   0   & 0.49 & ON & 0.00 & QC PE & 1.00 & 1.00 &   0.01   &   0.03\\
			&     &   22   &   8   &   4.88   &   65   &   100.00   &  1.00  &   0   & 0.34 & ON & 0.00 & ON QC PE & 1.00 & 0.15 &   0.31   &   0.04\\
			&     &   23   &   9   &   3.47   &   29   &   37.93   &  1.00  &   0   & 0.18 & QC & 0.00 & ON QC PE & 0.71 & 0.00 &   0.24   &   0.03\\
			&     &   24   &   19   &   6.28   &   318   &   6.92   &  1.00  &   0   & 0.20 & ON & 0.00 & QC PE & 0.85 & 0.56 &   2.81   &   0.02\\
			&     &   25   &   25   &   5.51   &   1000   &   0.20   &  1.00  &   0   & 0.21 & QC & 0.00 & PE & 0.77 & 0.56 &   21.07   &   0.03\\
			&     &   26   &   3   &   1.08   &   3   &   100.00   &  1.00  &   0   & 0.51 & ON & 0.00 & ON QC PE & 1.00 & 0.85 &   0.02   &   0.03\\
			&     &   27   &   8   &   2.55   &   29   &   10.34   &  1.00  &   0   & 0.26 & ON & 0.00 & QC PE & 0.53 & 0.24 &   0.17   &   0.01\\
			&     &   28   &   4   &   1.48   &   7   &   28.57   &  0.93  &   2   & 0.49 & ON & 0.00 & ON QC PE & 0.46 & 0.00 &   0.03   &   0.02\\
			&     &   29   &   1   &   0.85   &   1   &   100.00   &  1.00  &   0   & 0.00 & ON PE & -0.00 & QC & 0.00 & 0.00 &   0.01   &   0.03\\
			&     &   30   &   3   &   2.10   &   5   &   60.00   &  1.00  &   0   & 0.16 & ON & 0.00 & ON QC PE & 0.37 & 0.00 &   0.02   &   0.03\\
			2009    &   40   &   21   &   32   &   7.75   &   1000   &   0.10   &  0.96  &   10   & 0.06 & ON & 0.00 & QC PE & 0.13 & 0.13 &   42.68   &   0.02\\
			&     &   22   &   79   &   12.78   &   1000   &   0.60   &  0.99  &   153   & 0.08 & ON & 0.00 & PE & 0.66 & 0.66 &   125.39   &   0.07\\
			&     &   23   &   8   &   3.09   &   23   &   4.35   &  0.98  &   1   & 0.00 & ON QC PE & 0.00 & ON QC PE & 0.00 & 0.00 &   0.21   &   0.03\\
			&     &   24   &   40   &   11.20   &   1000   &   4.60   &  0.96  &   180   & 0.20 & QC & 0.03 & PE & 0.72 & 0.56 &   43.87   &   0.05\\
			&     &   25   &   67   &   11.75   &   1000   &   0.00   &  $<$ 0.94  &     &     &     &      &     &       &       &   78.70   &   0.03\\
			&     &   26   &   33   &   7.92   &   1000   &   0.80   &  0.92  &   172   & 0.12 & ON & 0.00 & QC PE & 0.47 & 0.42 &   30.10   &   0.03\\
			&     &   27   &   55   &   10.86   &   1000   &   16.00   &  1.00  &   0   & 0.29 & ON & 0.00 & PE & 0.89 & 0.71 &   76.12   &   0.07\\
			&     &   28   &   48   &   10.76   &   1000   &   4.30   &  1.00  &   3   & 0.19 & ON & 0.00 & PE & 0.70 & 0.47 &   63.68   &   0.04\\
			&     &   29   &   61   &   12.20   &   1000   &   1.60   &  1.00  &   1   & 0.17 & QC & 0.00 & PE & 0.90 & 0.67 &   74.77   &   0.03\\
			&     &   30   &   37   &   9.75   &   1000   &   18.70   &  1.00  &   1   & 0.16 & QC & 0.00 & PE & 0.85 & 0.67 &   50.19   &   0.03\\
			2013    &   40   &   21   &   27   &   7.00   &   1000   &   100.00   &  1.00  &   0   & 0.46 & ON & 0.00 & QC PE & 1.00 & 0.70 &   20.68   &   0.06\\
			&     &   22   &   18   &   6.31   &   572   &   5.94   &  1.00  &   0   & 0.35 & ON & 0.00 & QC PE & 0.67 & 0.40 &   6.06   &   0.03\\
			&     &   23   &   15   &   6.35   &   239   &   0.84   &  1.00  &   0   & 0.24 & ON & 0.00 & QC PE & 0.81 & 0.61 &   3.19   &   0.03\\
			&     &   24   &   7   &   2.10   &   12   &   50.00   &  1.00  &   0   & 0.24 & PE & 0.00 & ON QC PE & 0.69 & 0.00 &   0.09   &   0.03\\
			&     &   25   &   16   &   5.47   &   141   &   2.84   &  1.00  &   0   & 0.18 & ON & 0.00 & ON QC PE & 1.00 & 0.74 &   1.69   &   0.01\\
			&     &   26   &   6   &   3.66   &   20   &   45.00   &  1.00  &   0   & 0.33 & ON & 0.00 & QC PE & 1.00 & 0.40 &   0.12   &   0.02\\
			&     &   27   &   50   &   7.59   &   1000   &   0.00   &  $<$ 0.90  &     &     &     &      &     &       &       &   53.50   &   0.03\\
			&     &   28   &   34   &   8.80   &   1000   &   0.40   &  0.94  &   26   & 0.12 & ON & 0.00 & QC PE & 0.67 & 0.35 &   31.37   &   0.02\\
			&     &   29   &   14   &   3.82   &   80   &   2.50   &  0.89  &   14   & 0.27 & ON & 0.00 & QC PE & 0.36 & 0.29 &   0.77   &   0.03\\
			&     &   30   &   7   &   4.20   &   37   &   100.00   &  1.00  &   0   & 0.34 & ON & 0.00 & ON QC PE & 1.00 & 0.19 &   0.20   &   0.04\\
		\end{tabularx}
		\caption{Results for the {\proKEGwei} with the Canadian KPD instances with  the largest player, a medium player and the smallest player: Ontario, Quebec and Prince Edward Island.}
		\label{Table:Provinces_game_weight_5}
	\end{table}
	
\end{landscape}


\begin{landscape}
	\setlength{\tabcolsep}{6pt}
	\begin{table}[t]
		\hspace*{-7cm}
		\tiny
		\centering
		\begin{tabularx}{\textwidth}{rrrr|rrrrrrrrrrrrr}
			year & $\vert V \vert$ & ins & $\vert E \vert$ &
			$W(M)_{max}$ & $K$   & \% W-NE & $\epsilon$ & iter & $ IMP_{max}$  &  $Players_{max} $ &   $ IMP_{min}$ &  $Players_{min} $&
			{\IndA}$_{high}$ & {\IndA}$_{low}$ &
			$time_{gen}$ & $time_{NE}$   \\
			\cline{1-17}
			2009    &   50   &   21   &   101   &   14.74   &   1000   &   0.00   &  $<$ 0.98  &     &     &     &      &     &       &       &   180.86   &   0.29\\
			&     &   22   &   85   &   12.96   &   1000   &   0.00   &  $<$ 0.97  &     &     &     &      &     &       &       &   92.70   &   0.06\\
			&     &   23   &   82   &   12.16   &   1000   &   0.50   &  0.98  &   223   & 0.14 & ON & 0.00 & PE & 0.56 & 0.54 &   107.51   &   0.23\\
			&     &   24   &   100   &   15.37   &   1000   &   0.10   &  0.99  &   174   & 0.08 & QC & 0.00 & PE & 0.58 & 0.58 &   228.05   &   0.14\\
			&     &   25   &   55   &   12.71   &   1000   &   1.70   &  0.97  &   178   & 0.09 & ON & 0.00 & PE & 0.66 & 0.47 &   76.85   &   0.08\\
			&     &   26   &   29   &   9.32   &   1000   &   3.90   &  0.94  &   72   & 0.26 & ON & 0.00 & QC PE & 0.56 & 0.26 &   20.00   &   0.02\\
			&     &   27   &   94   &   11.08   &   1000   &   1.40   &  0.99  &   397   & 0.13 & ON & 0.00 & PE & 0.51 & 0.32 &   154.15   &   0.15\\
			&     &   28   &   47   &   9.97   &   1000   &   10.70   &  1.00  &   0   & 0.22 & ON & 0.00 & QC PE & 0.65 & 0.41 &   39.25   &   0.08\\
			&     &   29   &   53   &   10.22   &   1000   &   2.90   &  1.00  &   31   & 0.20 & ON & 0.00 & QC PE & 0.75 & 0.65 &   64.09   &   0.10\\
			&     &   30   &   73   &   12.08   &   1000   &   0.20   &  0.98  &   119   & 0.06 & ON & 0.00 & PE & 0.47 & 0.46 &   104.49   &   0.05\\
			2013    &   50   &   21   &   9   &   3.94   &   53   &   9.43   &  1.00  &   0   & 0.11 & ON & 0.00 & ON QC PE & 0.48 & 0.00 &   0.37   &   0.03\\
			&     &   22   &   15   &   5.60   &   336   &   0.60   &  1.00  &   0   & 0.07 & ON & 0.00 & QC PE & 0.12 & 0.12 &   2.32   &   0.01\\
			&     &   23   &   23   &   6.85   &   1000   &   0.30   &  1.00  &   0   & 0.20 & ON & 0.00 & PE & 0.48 & 0.36 &   10.90   &   0.02\\
			&     &   24   &   41   &   12.88   &   1000   &   0.20   &  0.92  &   43   & 0.12 & ON & 0.00 & PE & 0.55 & 0.47 &   53.98   &   0.03\\
			&     &   25   &   14   &   5.49   &   327   &   7.95   &  1.00  &   0   & 0.20 & ON & 0.00 & ON PE & 0.55 & 0.16 &   2.78   &   0.02\\
			&     &   26   &   37   &   7.80   &   1000   &   0.10   &  0.96  &   16   & 0.11 & ON & 0.00 & QC PE & 0.31 & 0.31 &   33.71   &   0.02\\
			&     &   27   &   25   &   6.08   &   1000   &   0.20   &  1.00  &   0   & 0.29 & ON & 0.00 & PE & 0.78 & 0.76 &   20.08   &   0.02\\
			&     &   28   &   38   &   8.03   &   1000   &   14.10   &  1.00  &   0   & 0.24 & QC & 0.00 & QC PE & 0.91 & 0.58 &   37.52   &   0.06\\
			&     &   29   &   21   &   6.69   &   780   &   0.38   &  0.86  &   43   & 0.32 & ON & 0.00 & QC PE & 0.69 & 0.48 &   8.84   &   0.02\\
			&     &   30   &   35   &   6.98   &   1000   &   3.90   &  0.98  &   9   & 0.19 & QC & -0.00 & ON & 0.57 & 0.35 &   37.05   &   0.03\\
		\end{tabularx}
		\caption{Continuation of Table~\ref{Table:Provinces_game_weight_5}. Results for the {\proKEGwei} with the Canadian KPD instances with  the largest player, a medium player and the smallest player: Ontario, Quebec and Prince Edward Island.}
		\label{Table:Provinces_game_weight_6}
	\end{table}
	
\end{landscape}

\section{Conclusions and open questions}\label{sec:conclusions}
In this work, we reviewed the literature on cross-border Kidney Exchange Programs, and motivated the reason to concentrate on a decentralized non-cooperative game theoretical model, where players control internal exchanges. Known results for the 2-player case were generalized to an arbitrary finite number of players. In this context, for {\proKEGcard} (1) it was discussed the mechanism that incentives players to fully disclose their compatibility graphs and (2) it was established the existence of social welfare equilibria that can be computed in polynomial time, stressing the feasibility of such solution being computable by the players in practice.  Furthermore, it was discussed how the game outcome would change upon the availability of information on transplants quality. Namely, we proved that the verification of Nash equilibrium for the {\proKEGwei} becomes an NP-complete problem, which is an indication that such solution might not be practical. Nevertheless, our computational results have shown that, in practice, this is a tractable problem (recall Section~\ref{sec:computational}, in particular, columns  $time_{NE}$ on the computational tables). Finally, on our computational results, we tested the concept of Nash equilibria on instances closely inspired by the topology of  the Canadian Kidney Pair Donation Program. These results demonstrated that for the  {\proKEGcard}, random selection of a social optimal solution, \ie, a randomly-selected maximum-cardinality matching, might not result in an equilibrium, revealing the need to use this concept as the only mean to avoid unstable outcomes. Interestingly, the computed social welfare equilibria  provide a good balance in terms of advantage to the participating players. In what concerns {\proKEGwei}, computational results show that {\SWE} might fail to exist and even (pure) equilibria might not exist, at least among the $K$ best solutions computed in practice. Nevertheless, a majority of instances had equilibria, which did not seem to be far from social optimal solutions.

The analysis of this game theoretical model opens interesting research directions. The {\proKEGcard} enjoys excellent properties from the players and social welfare point of view and thus, it seems a promising model to extend the study, namely, to exchanges of size larger than 2 and to include chains. However, once exchanges of size 3 are allowed, {\SWE} might not exist (see the example in Figure~\ref{fig:Example2}) and if there are 3 players, the game might not have a pure equilibrium (see Figure~\ref{fig:N3L3}). A possible way to overcome these issues can be to restrict the {\IndA} to exchanges of size 2 or consider that the game is played over time, so that cooperation might be achieved because players can punish any disobedient on future plays. Indeed, similarly to many KEP platforms, in practice, our game would be played over multiple periods. In this context, \citet{Dickerson2015} showed good social outcomes for a centralized mechanism over an infinite time horizon, where players have a cardinality utility. Hence, given that our static version with $L=2$, {\proKEGcard}, already overcome impossibility results for (static) centralized mechanisms, we believe that accounting for a dynamic setting adds an extra degree of flexibility for obtaining equilibrium guarantees when $L>2$.

Our computational results on {\proKEGcard}  show that KEPs can easily have multiple optimal solutions and thus, for sake of fairness and proper evaluation of KEPs outcomes, it is crucial to develop tools that enable the generation of random optimal solutions\footnote{In the KEP literature, interest on this direction of random selection of maximum cardinality matchings is present in the context of egalitarian mechanisms (which do not focus on uniform randomization), {\eg}, see~\cite{Roth_Sonme_Unver_2005_b,Jian2014}.}. While we proposed a dynamic programming approach, more scalable methods are necessary.

This work also studied for the first time {\proKEGwei}. Under this setting, although examples can be found where equilibria do not exist, the computational results  reveal that, in practice, the social optimum might not be too sacrificed in order to get an equilibrium. Ideally, one would prefer to have theoretical guarantees. Some criticism might be pointed out to this game model: players might manipulate the information provided in order  to associate high weights to exchanges that would benefit their patients and/or players might use different technologies to evaluate pairs health. In fact, achieving medical consensus on the attribution of weights to exchanges and / or establishing a common international legislation may not be possible. Plus, the consideration of weights is by itself a limit on the players and the {\IndA} optimization criteria modeling. As a matter of fact, there are KEPs, such as the Dutch program, using hierarchical optimization instead of weighted matchings. As generally the first optimization criterion is the maximization on the number of patients, a potential way to overcome the aforementioned limitations of {\proKEGwei} is through the consideration that players use hierarchical optimization, optimizing first the objectives of {\proKEGcard}, and, second, among all solutions that lead to the same number of their patients matched, they maximize matchings' weights. 
An interesting alternative could be to allow players to provide a list of preferences to their potential transplants, similarly to what has been considered in the game theory literature when patients are the players and they have preferences based on transplant quality~\cite{NICOLO2012,NICOLO2017}. In this way, two games  would be played: countries play  {\proKEGcard} (game 1)  while guaranteeing that their pairs are under a \textit{stable} matching (game 2), \ie, no two pairs would be better off (higher in their preference list) by leaving their current match and establishing a matching among them. The idea is that players seek an equilibrium to {\proKEGcard} that is also a stable matching. 

Another interesting research venue with potential to overcome the no pure Nash equilibrium guarantees and players coordination towards a social optimum when $L>2$ and exchange weights are considered is the study of other solution concepts such as correlated equilibria. Correlated equilibrium is a  more broader concept than Nash equilibrium and, computationally, it is generally faster to determine. In fact, the latter would benefit from the fact that computing a player best response is efficient in practice  allowing the use of our Algorithm~\ref{Alg:Best_weighted_Optimistic_response}\footnote{Algorithm~\ref{Alg:Best_weighted_Optimistic_response} can be adapted in a straightforward way to account for $L>2$.} as the separation oracle in the framework of~\cite{JIANG2015347} to compute correlated equilibria.

In conclusion, there are many important points both on modeling cross boarder Kidney Exchange Programs and on their analysis that need to be addressed and, in this paper, we provided a consistent base to support further research. The ultimate goal is to design an interaction among the countries that implements the social choice function targeting social optimality (\eg, see  \cite{PALFREY20022271} for details in implementation theory). 

\begin{figure}[!h]
	\centering
	\begin{tikzpicture}[-,>=stealth',shorten >=0.3pt,auto,node distance=2cm,
		thick,countryA node/.style={circle,fill=black!20,draw},countryB node/.style={diamond,draw}, scale=.75, transform shape]]
		\tikzstyle{matched} = [draw,line width=3pt,-]
		
		\node[countryA node] (5) {5};
		\node[countryA node] (4) [right of=5] {4};
		\node[countryA node] (1) [above of=4] {1};
		\node[countryB node] (2) [right of=1] {2};
		\node[countryB node] (3) [below of=2] {3};
		\node[countryB node] (7) [below of=4] {7};
		\node[countryB node] (6) [left of=7] {6};
		
		\path[->,every node/.style={font=\sffamily\small}]
		( 1) edge[blue] node {} (2)
		( 2) edge[blue] node {} (3)
		( 3) edge[blue] node {} (1)
		(1) edge node {} (4)
		(4) edge node {} (5)
		(5) edge node {} (1)
		(4) edge[blue] node {} (7)
		(7) edge[blue] node {} (6)
		(6) edge[blue] node {} (4);
	\end{tikzpicture}
	\caption{A game instance with $l=3$. The social optimum (in blue) is 6 but the only equilibrium is for the gray player to select her internal cycle of length 3. }
	\label{fig:Example2}
\end{figure}
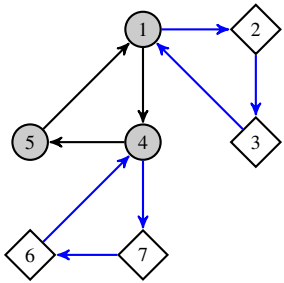

\begin{figure}[!h]
	\centering
	\begin{tikzpicture}[-,>=stealth',shorten >=0.3pt,auto,node distance=2cm,
		thick,countryA node/.style={circle,fill=black!20,draw},countryB node/.style={diamond,draw}, countryC node/.style={rectangle,fill=magenta!60,draw},scale=.75, transform shape]]
		\tikzstyle{matched} = [draw,line width=3pt,-]
		
		\node[countryB node] (1) {1};
		\node[countryB node] (2) [right of=1] {2};
		\node[countryB node] (3) [right of=2] {3};
		\node[countryA node] (4) [right of=3] {4};
		\node[countryC node] (5) [right of=4] {5};
		\node[countryA node] (6) [above of=2] {6};
		\node[countryA node] (7) [above of=4] {7};
		\node[countryC node] (8) [above of=6] {8};
		\node[countryA node] (9) [right of=8] {9};
		\node[countryB node] (10) [below of=3] {10};

		\path[->,every node/.style={font=\sffamily\small}]
		( 1) edge node {} (6)
		( 2) edge node {} (1)
		edge node {} (3)
		( 3) edge node {} (10)
		(4) edge node {} (3)
		edge node {} (7)
		(5) edge node {} (4)
		(6) edge node {} (2)
		edge node {} (4)
		edge node {} (8)
		(7) edge node {} (6)
		edge node {} (5)
		(8) edge node {} (9)
		(9) edge node {} (6)
		(10) edge node {} (2)
		edge node {} (4);
	\end{tikzpicture}
	\caption{ Player {\color{magenta}square} has only the strategy $\emptyset$; Player {\color{gray}ball} has two possible strategies $\emptyset$ and $(6,4,7)$; Player diamond has two possible strategies $\emptyset$ and $(3,10,2)$. The {\IndA} selects cycle (1,6,2) and (4,3,10) when no internal exchange is selected.}
	\label{fig:N3L3}
\end{figure}
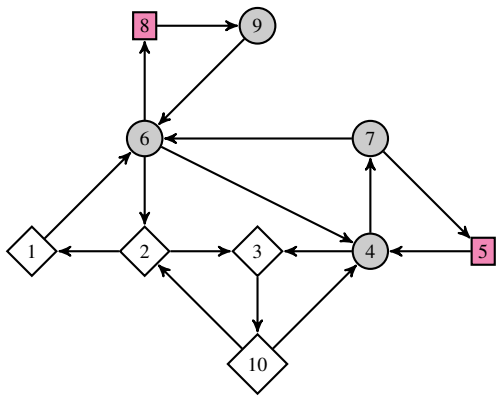

\section*{Acknowledgements}

The authors thank the Canadian Blood Services for providing the data used in this work.

The first author  wishes to thank the support of the Institut de valorisation des donn\'ees and Fonds de Recherche du Qu\'ebec through the FRQ–IVADO Research Chair in Data Science for Combinatorial Game Theory, and the Natural Sciences and Engineering Research Council of Canada through the discovery grant 2019-04557.


\newpage

\begin{appendix}
	\section{Proof of Theorem \ref{THM:SWE_EXISTENCE}}
	\label{app:SWE}
	
	For sake of clarity and in order to make the paper self contained, we recall some known results on matching theory and kidney exchange games.
	
	Algorithm~\ref{Alg:maximum matching} recalls the high level computation of a maximum matching: starting in a matching $M$, iteratively, the matching is augmented through the symmetric difference with a $M$-augmenting path, increasing the set of matched vertices. A valid starting matching can always be $M=\emptyset$. \citet{Edmonds1965a} proved that step~\ref{step:augmentation_path} can run in polynomial time which implies that the full algorithm also runs in polynomial time.
	
	\begin{algorithm}[!h]
		\caption{Computation of a maximum matching for graph $G$ given a matching $M$ (possibly, $M=\emptyset$).}
		\begin{algorithmic}[1]
			\FOR{$v \in V$, $M$-unmatched}
			\STATE Find a $M$-augmenting path $\mathfrak{p}$ starting in $v$ \label{step:augmentation_path}
			\IF{$\mathfrak{p}$ exists}
			\STATE $ M:= M \oplus \mathfrak{p}$ \hspace*{0.7cm}		{\color{gray}\# $\oplus$ represents the symmetric difference of two sets}
			\ENDIF
			\ENDFOR
			\RETURN $M$
		\end{algorithmic}
		\label{Alg:maximum matching}
	\end{algorithm}

	Next, we recover from \citet{Carvalho2017} the theorem providing the necessary and sufficient condition under which a player has incentive to deviate from a maximum matching. For an illustration of their result see Figures~\ref{fig:incentiveDeviateA} and \ref{fig:incentiveDeviateB}.
	
	\begin{theorem}
		Let $M$ be a maximum matching and $G^p(M)$ the subgraph of $G$ restricted to the edges in $E^p$ and the edges of $E^I$ without $M$-matched internally vertices. A player $p \in  N$ has incentive to deviate from $M \cap E^p$ if and only if
		\begin{enumerate}
			\item there is a $M^p \cup \left(M\cap E^I \right)$-alternating path $\mathfrak{p}$ in  $G^p(M)$ whose origin is a node in $V^p$, unmatched in this path, and the destination is a $M\cap E^I_p$-matched node in $V^{-p}$ and
			\item $\mathcal{A}\left(\mathbf{M\oplus \mathfrak{p}}\right)$ is equal to $ \left( M \oplus \mathfrak{p} \right) \cap E^I$.\footnote{This condition is not on the original paper by~\citet{Carvalho2017} which is restricted to the 2-player case, where players are indifferent about $\mathcal{A}$.}
		\end{enumerate}

			\begin{proof}
				In Theorem 6 of \citet{Carvalho2017}, if $M$  is a maximum matching of $G$, then only condition (ii) is necessary and sufficient, as the others correspond to augmenting paths and maximum matchings cannot have augmenting paths. 
				
				Next, note that Theorem 6 of \citet{Carvalho2017} refers exclusively to a player $p$ incentive to unilaterally deviate and, in particular, it just considers the graph $G^p(M)$, \ie, the part of the graph that player $p$'s internal matching can change when the opponents play according to $M$. Moreover, their proof does not consider any assumption on the {\IndA} selection when there are multiple optimal external exchanges. Hence, it does not matter  the fact that \citet{Carvalho2017} focus in 2-player games, neither the deterministic selection of the {\IndA} for their result to extend to any number of opponents for player $p$ and deterministic algorithm $\mathcal{A}$.
			\end{proof}

		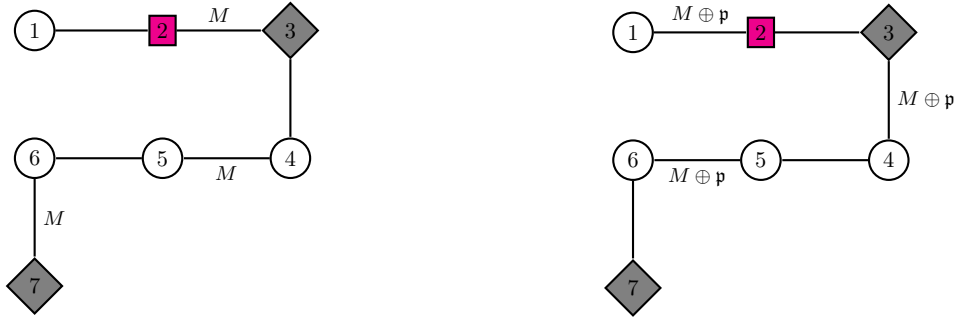
\begin{figure}
			\centering
			\begin{minipage}{.5\textwidth}
				\centering
				\begin{tikzpicture}[-,>=stealth',shorten >=0.3pt,auto,node distance=2cm,
					thick,countryA node/.style={circle,draw},countryB node/.style={diamond,draw,fill=gray},countryC node/.style={rectangle,draw,fill=magenta} , scale=.85, transform shape]
					\tikzstyle{matched} = [draw,line width=3pt,-]
					
					\node[countryA node] (1) {$1$};
					\node[countryC node] (2) [right of=1] {$2$};
					\node[countryB node] (3) [right of=2] {$3$};
					\node[countryA node] (4) [below of=3] {$4$};
					\node[countryA node] (5) [left of=4] {$5$};
					\node[countryA node] (6) [left of=5] {$6$};
					\node[countryB node] (7) [below of=6] {$7$};
					
					\path[-,every node/.style={font=\sffamily\small}]
					(1) edge node  {} (2)
					(2) edge node  {$M$} (3)
					(3) edge node  {} (4)
					(4) edge node  {$M$} (5)
					(5) edge node  {} (6)
					(6) edge node  {$M$} (7);
				\end{tikzpicture}
			\end{minipage}%
			\begin{minipage}{0.5\textwidth}
				\centering
				\begin{tikzpicture}[-,>=stealth',shorten >=0.3pt,auto,node distance=2cm,
					thick,countryA node/.style={circle,draw},countryB node/.style={diamond,draw,fill=gray},countryC node/.style={rectangle,draw,fill=magenta} , scale=.85, transform shape]
					\tikzstyle{matched} = [draw,line width=3pt,-]
					
					\node[countryA node] (1) {$1$};
					\node[countryC node] (2) [right of=1] {$2$};
					\node[countryB node] (3) [right of=2] {$3$};
					\node[countryA node] (4) [below of=3] {$4$};
					\node[countryA node] (5) [left of=4] {$5$};
					\node[countryA node] (6) [left of=5] {$6$};
					\node[countryB node] (7) [below of=6] {$7$};
					
					\path[-,every node/.style={font=\sffamily\small}]
					(1) edge node  {$M \oplus \mathfrak{p}$} (2)
					(2) edge node  {} (3)
					(3) edge node  {$M \oplus \mathfrak{p}$} (4)
					(4) edge node  {} (5)
					(5) edge node  {$M \oplus \mathfrak{p}$} (6)
					(6) edge node  {} (7);
				\end{tikzpicture}
			\end{minipage}
			\caption{Player 1 owns the white circles, player 2 owns the gray diamonds and player 3 owns the magenta squares. If $ \mathfrak{p} = (1,2,3,4,5,6,7)$, player 1 has incentive to deviate from $M \cap E^1=\lbrace (5,4) \rbrace$ to $(M \oplus \mathfrak{p}) \cap E^1 = \lbrace (6,5) \rbrace$, since $\mathcal{A}(M \oplus \mathfrak{p} )= \lbrace (4,3),(2,1) \rbrace$, resulting in more vertices of player 1 matched.}
			\label{fig:incentiveDeviateA}
		\end{figure}
		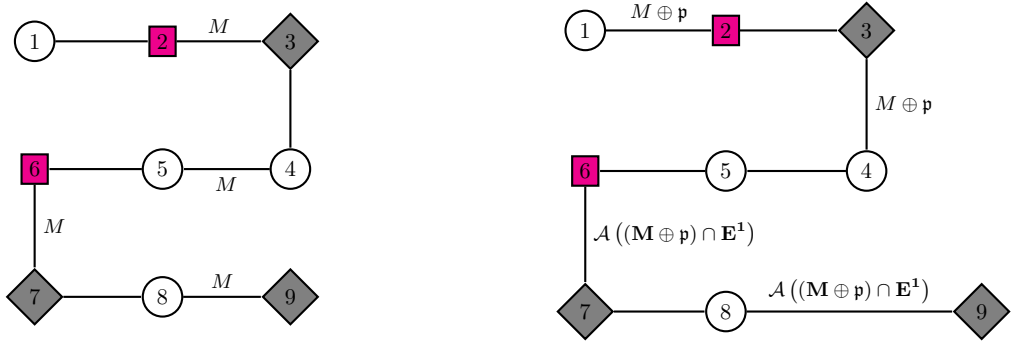
\begin{figure}
			\centering
			\begin{minipage}{.5\textwidth}
				\centering
				\begin{tikzpicture}[-,>=stealth',shorten >=0.3pt,auto,node distance=2cm,
					thick,countryA node/.style={circle,draw},countryB node/.style={diamond,draw,fill=gray},countryC node/.style={rectangle,draw,fill=magenta} , scale=.85, transform shape]
					\tikzstyle{matched} = [draw,line width=3pt,-]
					
					\node[countryA node] (1) {$1$};
					\node[countryC node] (2) [right of=1] {$2$};
					\node[countryB node] (3) [right of=2] {$3$};
					\node[countryA node] (4) [below of=3] {$4$};
					\node[countryA node] (5) [left of=4] {$5$};
					\node[countryC node] (6) [left of=5] {$6$};
					\node[countryB node] (7) [below of=6] {$7$};
					\node[countryA node] (8) [right of=7] {$8$};
					\node[countryB node] (9) [right of=8] {$9$};
					
					\path[-,every node/.style={font=\sffamily\small}]
					(1) edge node  {} (2)
					(2) edge node  {$M$} (3)
					(3) edge node  {} (4)
					(4) edge node  {$M$} (5)
					(5) edge node  {} (6)
					(6) edge node  {$M$} (7)
					(7) edge node  {} (8)
					(8) edge node  {$M$} (9);
				\end{tikzpicture}
			\end{minipage}%
			\begin{minipage}{0.5\textwidth}
				\centering
				\begin{tikzpicture}[-,>=stealth',shorten >=0.3pt,auto,node distance=2.2cm,
					thick,countryA node/.style={circle,draw},countryB node/.style={diamond,draw,fill=gray},countryC node/.style={rectangle,draw,fill=magenta} , scale=.85, transform shape]
					\tikzstyle{matched} = [draw,line width=3pt,-]
					
					\node[countryA node] (1) {$1$};
					\node[countryC node] (2) [right of=1] {$2$};
					\node[countryB node] (3) [right of=2] {$3$};
					\node[countryA node] (4) [below of=3] {$4$};
					\node[countryA node] (5) [left of=4] {$5$};
					\node[countryC node] (6) [left of=5] {$6$};
					\node[countryB node] (7) [below of=6] {$7$};
					\node[countryA node] (8) [right of=7] {$8$};
					\node[countryB node] (9) [right of=8,node distance=4cm] {$9$};
					
					\path[-,every node/.style={font=\sffamily\small}]
					(1) edge node  {$M \oplus \mathfrak{p}$} (2)
					(2) edge node  {} (3)
					(3) edge node  {$M \oplus \mathfrak{p}$} (4)
					(4) edge node  {} (5)
					(5) edge node  {} (6)
					(6) edge node  {$\mathcal{A}\left(\mathbf{(M \oplus \mathfrak{p})\cap E^1} \right)$} (7)
					(7) edge node  {} (8)
					(8) edge node  {$\mathcal{A}\left(\mathbf{(M \oplus \mathfrak{p})\cap E^1 }\right)$} (9);
				\end{tikzpicture}
			\end{minipage}
			\caption{Player 1 owns the white circles, player 2 owns the gray diamonds and player 3 owns the magenta squares. Suppose algorithm $\mathcal{A}$ takes decisions accordingly with the figures above. If $ \mathfrak{p}= (1,2,3,4,5,6,7,8,9)$, player 1 does not have incentive to deviate from $M \cap E^1= \lbrace(5,4) \rbrace$   to $(M \oplus \mathfrak{p}) \cap E^1 = \emptyset$ because algorithm $\mathcal{A}$ does not output matchings $(5,6)$ and $(7,8)$, but $(6,7)$ and $(8,9)$ which is also maximum (external) matching.}
			\label{fig:incentiveDeviateB}
		\end{figure}
		\label{thm:Carvalho2KEG}
	\end{theorem}

	The results above will enable us to prove Theorem~\ref{THM:SWE_EXISTENCE}.  To that end we  introduce the definition of maximal $M$-augmenting path in $G^{\mathcal{D}}=(V \cup \mathcal{D}, E \cup E^{\mathcal{D}})$ which is a $M$-alternating path starting in a $M$-unmatched vertex of $V$, ending in a edge of $E^{\mathcal{D}}$ which cannot be removed and replaced by other edges without the path stopping to be $M$-augmenting. Now, we have all elements to build the overall algorithm that we denote as Algorithm~\ref{Alg:SWE_matching}.
	
	\begin{algorithm}[!ht]
		\caption{Computation of a {\SWE} given a maximum matching $M$.}
		\begin{algorithmic}[1]
			\STATE $\mathcal{D} := \lbrace d_{(u,v)}, d_{(v,u)}: (u,v) \in E^I \cap M, u \in V^i, v \in V^j, i\neq j \rbrace$ \label{step:artificialGraph}
			\STATE $E^\mathcal{D} := \lbrace (d_{u,v},u): d_{(u,v)} \in \mathcal{D} \rbrace$
			\STATE $G^\mathcal{D}:=(V \cup \mathcal{D}, E \cup E^\mathcal{D})$
			\STATE $ \bar{M}:=M$ // maximum matching of $G$ obtained from Algorithm~\ref{Alg:maximum matching} \label{step:maxmat}
			\WHILE{there is a maximal $\bar{M}$-augmenting path $\mathfrak{p}$ in  $G^\mathcal{D}$}  \label{step:disjoint_paths}
			\STATE $ \bar{M}:= \bar{M} \oplus \mathfrak{p}$
			\ENDWHILE
			\RETURN  $ \bar{M}$
		\end{algorithmic}
		\label{Alg:SWE_matching}
	\end{algorithm}
	
	The correctness of  Algorithm~\ref{Alg:SWE_matching} is proven by the following theorem.
	
	\begin{theorem}
		If for an output of Algorithm~\ref{Alg:SWE_matching}, $\mathcal{A}\left(\mathbf{\bar{M}} \right) = \bar{M} \cap E^I$, then $\bar{M}$ is a {\SWE} and it can be computed in polynomial time.
	\end{theorem}
	\begin{proof}
		By construction, $\bar{M}$ is a maximum matching of $G$: $\bar{M}$ starts by being a maximum matching of $G$, step~\ref{step:maxmat}, and afterwards,   maximal $\bar{M}$-alternating paths in $G^{\mathcal{D}}$ are applied to $\bar{M}$ which does not change the cardinality of $\bar{M}$ in $G$ (but augments the number of vertices matched in  $G^{\mathcal{D}}$). Therefore, the output $\bar{M}$ is a social optimum. It remains to show that $\bar{M}$ is a Nash equilibrium. To that end, one must prove that there is no $\bar{M}$-alternating path as in Theorem~\ref{thm:Carvalho2KEG}.
		
		Let us define $\mathcal{P}$ as the set of disjoint paths obtained through   $\bar{M}  \oplus M $; $\bar{M}$ and $M$ are maximum matchings of $G$, thus the operation  $\oplus$ leads to cycles of even size and paths. 
		Note that $\bar{M}  =  M \displaystyle \oplus_{\mathfrak{p} \in \mathcal{P}} \mathfrak{p}$. By construction,  any path in $\mathcal{P}$ is a maximal $M$-augmenting path in $G^{\mathcal{D}.}$
		
		By contradiction, suppose that  $\mathfrak{p} = (v_1, \ldots, v_k)$ is a $\bar{M}$-alternating path in $G$ as in Theorem~\ref{thm:Carvalho2KEG}.
		If $\mathfrak{p}$ is a maximal $M$-augmenting path in $G^{\mathcal{D}}$, then it should have been found in step~\ref{step:disjoint_paths} (or was contained in a maximal $M$-augmenting path in $G^{\mathcal{D}}$). Thus, $\mathfrak{p}$ is not maximal $M$-augmenting path, \ie, it is not in  $\mathcal{P}$. This implies that $v_1$ is $\bar{M}$-matched with a $d_{u,v_1} \in \mathcal{D}$. Let $\hat{\mathfrak{p}} \in \mathcal{P}$ be the path containing the edge $(d_{v_1,u},v_1)$, then $\hat{\mathfrak{p}}$ is not a maximal $M$-augmenting path in $G^{\mathcal{D}}$, since it could have been extended through  $\mathfrak{p}$ (ending in $(d_{v_k,r},r)$).
		
		Algorithm~\ref{Alg:SWE_matching} can run in polynomial time since its most  consuming time operation, step~\ref{step:disjoint_paths}, can be done in polynomial time; recall Algorithm~\ref{Alg:maximum matching}.
		
	\end{proof}
	Finally, we relax the previous theorem dependence in the output of $\mathcal{A}$. First, make $M$ equal to $\left( \tilde{M} \setminus E^I \right) \cup \mathcal{A}\left(\mathbf{\tilde{M}}\right)$, where $\tilde{M}$ is a maximum matching, and give it as input for Algorithm~\ref{Alg:SWE_matching}. Now, proceed as in the previous proof to compute $\mathcal{P}$, but without edges in $E^{\mathcal{D}}$. Each $\mathfrak{p} \in \mathcal{P}$ can be partitioned in $\mathfrak{p}_1, \mathfrak{p}_2, \ldots, \mathfrak{p}_k$ such that $\mathfrak{p} = \mathfrak{p}_1 \mathfrak{p}_2 \ldots \mathfrak{p}_k$ and each $\mathfrak{p}_i \in E^p \cup E^I$. For each  $\mathfrak{p}$, while $\mathcal{A}$ agrees with the application of $ \mathfrak{p}_i$ to $M$ through operation $\oplus$ ( \ie,  $\mathcal{A} \left(\mathbf{M \oplus \mathfrak{p}_i }\right)= (M \oplus \mathfrak{p}_i) \cap E^I$), do $M = M \oplus \mathfrak{p}_i$. In the end, by the same arguments as the previous proof, the obtained $M$ is an equilibrium.

	\section{Proof of Theorem~\ref{THM:FULLINFORMATION}}
	\label{app:FullInformation}
	
	Without loss of generality, assume that the players' set of revealed vertices satisfy $|V^1| \geq |V^2| \geq \ldots \geq |V^n|$.
	
	The optimal solution of the following maximum-weighted matching is a maximum matching and prioritizes the players with more vertices revealed:
	\begin{subequations}
		\begin{alignat}{4}
			\max_y & \sum_{(u,v) \in E^I} y_{u,v} w_{u,v} \\
			\hbox{s.t. } & \sum_{(u,v) \in E^I} y_{u,v} \leq 1 & \forall u \in V \\
			& y_{u,v} \in \lbrace 0,1 \rbrace & \forall (u,v) \in E^I,
		\end{alignat}
		\label{Optimization:FullInfo}
	\end{subequations}
	where $w_{u,v}$   is equal to $(|V^1| + |V^2|) |V| + |V^i|+|V^j|$ for $u \in V^i$ and $v \in V^j$. Note that  Problem~\eqref{Optimization:FullInfo} is a maximum-weighted matching and thus, it can be solved in polynomial time.
	
	Next,  we split the proof in two parts by proving the following claims.
	
	\begin{claim}
		The optimal solution of Problem~\eqref{Optimization:FullInfo} is a maximum cardinality matching
	\end{claim}
	A maximum cartinality matching $M$ has weight at least $|M|( (|V^1|+|V^2|) |V| +|V^{n-1}| + |V^n| )$. Suppose that a maximum-weighted matching $\bar{M}$ is not of maximum cardinality. Then, it has weight at most
	
	\begin{subequations}
		\begin{alignat*}{4}
			|\bar{M}|\left( (|V^1|+|V^2|) |V| +|V^1| + |V^2| \right) =&|\bar{M}| \left( |V^1|+|V^2| \right) \left( |V| +1 \right) \\
			\leq & \left(|M|-1 \right) \left( |V^1|+|V^2| \right) \left( |V| +1 \right) \\
			= & |M| \left( |V^1|+|V^2| \right) \left( |V| +1 \right)\\
			& - \left( |V^1|+|V^2| \right) \left( |V| +1 \right) \\
			= & |M| \left( |V^1|+|V^2| \right) |V| + \\
			& |M| \left( |V^1|+|V^2| \right)  - \\
			& \left( |V^1|+|V^2| \right) \left( |V| +1 \right)\\
			= & |M| \left( |V^1|+|V^2| \right) |V| + \\
			& \left( |V^1|+|V^2| \right)\left( |M| -|V|-1 \right)\\
			\leq & |M| \left(|V^1|+|V^2| \right) |V|
		\end{alignat*}
	\end{subequations}
	which contradicts the fact that $\bar{M}$ is a maximum-weighted matching.
	
	\begin{claim}
		The optimal solution of Problem~\eqref{Optimization:FullInfo} prioritizes the players with more revealed vertices.
	\end{claim}
	Suppose that there is a $M$  solution for Problem~\eqref{Optimization:FullInfo} and a $\bar{M}$ with $|M| = |\bar{M}|$, such that $|M_1| = |\bar{M}_1|$, $|M_2| = |\bar{M}_2|$, $\ldots$, $|M_{p-1}| = |\bar{M}_{p-1}|$ and $|M_p| < |\bar{M}_p|$. Then, there is a path $\mathfrak{p}$ in $M \oplus \bar{M}$ such that it starts in an edge $(u_0,v_0) \in \bar{M}$ with $u_0 \in V^p$ and ends in an edge $(u_f,v_f) \in M$  with $v_f \notin V^p$. Therefore, it holds $\sum_{e \in M \cap \mathfrak{p}} w_e < \sum_{e \in \bar{M} \cap \mathfrak{p}} w_e $ which implies that $M$ is not the optimum of Problem~\eqref{Optimization:FullInfo}.

	\section{Proof of Theorem \ref{THM:COMPLEXITY}}\label{app:NP_complete}
	
	Next, we show that {\BestReponsepes} is NP-complete even when player $p$'s internal graph is bipartite. To that end, we reduce in polynomial time the following NP-complete problem to {\BestReponsepes}:
	
	\begin{equation}
		\begin{aligned}
			& \textrm{Problem: {\FVS} }  \\[1.0ex]
			& \textrm{Instance: A diagraph $G=(V,A)$ and a positive integer $l$. } \\[1.0ex]
			& \textrm{Question: Does there exist a set $V'\subseteq V$ such that} \\
			& \hspace{1.7cm}\textrm{$\vert V'\vert \leq l$ and $G'=(V \setminus V',A)$  has no directed cycle?}
		\end{aligned}
		\tag{FVS}\label{quote_feedbackvertex_problem}
	\end{equation}

	\begin{lemma}
		{\BestReponsepes} is NP-complete
	\end{lemma}
	\begin{proof} We split the proof in two parts by proving the following claims.

		\begin{claim}
			{\BestReponsepes} is in NP.
		\end{claim}
		Given $M^p$,  the pessimistic {\IndA} problem can be solved in polynomial time. Compute (in polynomial time) a maximum cardinality matching for {\IndA} given $M^p$. Let $k$ be the size of such matching. Then, the {\IndA} problem reduces to find a minimum weighted matching for player $p$ of cardinality $k$ which is equivalent to the well-known $k$-cardinality assigment problem. A $k$-cardinality assigment problem can be solved in polynomial time~\cite{Dellamico1997, Lodi2001}.

		\begin{claim}
			{\BestReponsepes} is NP-hard.
		\end{claim}
		First, we reduce in polynomial time, an instance of {\FVS} to one of {\BestReponsepes}:
		\begin{itemize}
			\item Let $N=\lbrace A,B \rbrace$ with $p=A$.
			\item Let $V=\lbrace v_1,v_2, \ldots,v_n\rbrace$. Set $V^A=V \cup \mathcal{L}$ with $\mathcal{L}= \lbrace 1,2,\ldots,l \rbrace$ and $V^B=\lbrace v'_1,v'_2, \ldots,v'_n\rbrace$.
			\item Set $E^A=\lbrace (i,j): i \in \mathcal{L}, j \in V\rbrace$ and $w_e^A=2$ for $e \in E^A$. Note that $G^A$ is bipartite. See Figure~\ref{fig:gadget_i}.
			\item $E^I=Id \cup \mathcal{A}$ with $Id=\lbrace (v_i,v'_j): i=j\rbrace$ and $\mathcal{A}= \lbrace (v_i,v'_j): (v_i,v_j) \in A\rbrace$, and $w^A_e= n$ for $e \in Id$ and $w_e^A= 1$ for $e \in A$. See Figure~\ref{fig:external_gadget}.
			\item Set $P=2l+n(n-l)$.
		\end{itemize}
		We claim that in the constructed instance of {\BestReponsepes}, player A can have a profit of $2l+n(n-l)$ if and only if the {\FVS} instance has answer YES.
		
		\paragraph{Proof of if} Assume the {\FVS} instance has  answer YES. Then, there is a $V' \subseteq V$ with $\vert V' \vert \leq  l$ such that $G'=(V\setminus V',A)$ does not contain any directed cycle. Set $M^A= \lbrace (i,v_i): v_i \in V' \rbrace$, leading to an internal profit of  $2 \vert V' \vert$. Now, the vertices $v_i \in V'$ are $M^A$-matched and thus are not available for the {\IndA}. In this way, the maximum cardinality matching for the {\IndA} is $n-\vert V' \vert$ (number of player A's vertices available for the {\IndA}), which can be attained by matching $(v_i,v'_i)$ for $v_i \in V \setminus V'$.  However, the {\IndA} has incentive to avoid to select edges in $Id$ since they have associated a weight of $n$ for player A. If there is an  {\IndA} matching of maximum cardinality  $n-\vert V' \vert$ that contains an edge $(v_i,v'_j) \in A$, then there is a cycle in the {\FVS} instance: $v_j$ must be matched with $v'_{j_1}$, ${j_1} \neq j$, $v_{j_1}$ must be matched with $v'_{j_2}$, $j_1 \neq j_2$, $\dots$, $v_{j_k}$ is matched with $v'_i$ (note that there is a finite number of vertices); then, $(v_i,v_{j_1}, v_{j_2}, \ldots, v_{j_k}, v_i)$ is a cycle in $G'$. Therefore, player A profit is $2 \vert V' \vert +n(n- \vert V' \vert ) \geq 2l+n(n-l)$ since $\vert V' \vert \leq l$.
		
		\paragraph{Proof of only if} Now assume that the {\BestReponsepes} is a YES instance. Then, there is $M^A \subseteq E^A$ such that player $A$'s profit is equal or greater to $2l+n(n-1)$. Let $V'$ be equal to the set of $M^A$-matched vertices in $V$. Note that $\vert V' \vert \leq l$, since $G^A$ is bipartite with the vertices partition $\mathcal{L}$ with size $l$. If the {\IndA} matching $M^I$ for $M^A$ is not contained in $Id$, as show in the proof of if, there is a cycle in $G'$ of size at least 2. Recall that the {\IndA} behaves pessimistically for player A and thus, she has incentive to select matchings associated with cycles in $G'$. Therefore, player A has incentive to match internally $l$ vertices which means $\vert V' \vert = l$. In this way, player A's profit is equal or less than $2l + 2 + n(n-2-l )$ which contradicts the fact that player A's profit was equal or greater than $2l+n(n-l)$.

		\begin{figure}[!h]
			\centering
			\begin{tikzpicture}[-,>=stealth',shorten >=0.3pt,auto,node distance=2cm,
				thick,countryA node/.style={circle,fill=black!20,draw},countryB node/.style={diamond,draw}, countryNone node/.style={diamond,fill=white},  scale=.75, transform shape]]
				\tikzstyle{matched} = [draw,line width=3pt,-]
				
				\node[countryA node] (1) {1};
				\node[countryA node] (2) [below of=1] {2};
				\node[countryNone node] (3) [below of=2] {$\vdots$};
				\node[countryA node] (4) [below of=3] {$l$};
				\node[countryA node] (5) [right of=2] {$v_i$};
				
				\path[-,every node/.style={font=\sffamily\small}]
				( 1) edge node {2} (5)
				( 2) edge node {2} (5)
				( 3) edge node {2} (5)
				(4) edge node {2} (5);
			\end{tikzpicture}
			\caption{Gadget $i=1,2, \ldots, n$.}
			\label{fig:gadget_i}
		\end{figure}
		
		\begin{figure}[!h]
			\centering
			\begin{tikzpicture}[-,>=stealth',shorten >=0.3pt,auto,node distance=2cm,
				thick,countryA node/.style={circle,fill=black!20,draw},countryB node/.style={diamond,draw}, countryNone node/.style={diamond,fill=white},  scale=.75, transform shape]]
				\tikzstyle{matched} = [draw,line width=3pt,-]
				
				\node[countryA node] (1) {$v_i$};
				\node[countryB node] (2) [right of=1] {$v'_i$};
				\node[countryB node] (3) [below of=2] {};
				\node[countryNone node] (4) [below of=3] {$\vdots$};
				\node[countryB node] (5) [below of=4] {};
				
				\path[-,every node/.style={font=\sffamily\small}]
				( 1) edge node {$n$} (2)
				( 1) edge node {1} (3)
				( 1) edge node {1} (4)
				( 1) edge node {1} (5);
			\end{tikzpicture}
			\caption{Gadget $i=1,2, \ldots, n$.}
			\label{fig:external_gadget}
		\end{figure}
	\end{proof}
	
	Next, we show that {\BestReponseopt} is NP-complete. To that end, we reduce in polynomial time the following NP-complete problem to {\BestReponseopt}:
	
	\begin{equation}
		\begin{aligned}
			& \textrm{Problem: {\threeDIM} }  \\[1.0ex]
			& \textrm{Instance: Given sets $X$, $Y$, $Z$, each of size $n$, and a set $T \subset X \times Y \times Z$. } \\[1.0ex]
			& \textrm{Question: Is there a set of $n$ triples in $T$ such that} \\
			& \hspace{1.7cm}\textrm{each element is contained in exactly one triplet?}
		\end{aligned}
		\tag{3-DIM}\label{quote_3dimentional_matching}
	\end{equation}
	
	\begin{lemma}
		{\BestReponseopt} is NP-complete.
	\end{lemma}
	\begin{proof} We split the proof in two parts by proving the following claims.
		
		\begin{claim}
			{\BestReponseopt} is in NP.
		\end{claim}
		Given $M^p$, we can solve the optimistic {\IndA} problem in polynomial time by reducing it to a maximum-weighted matching problem using a standard trick:  replace the weights $w_e^p$ for $e \in E^I_p$ by $w_e^p+\vert V^{-p}\vert w_{max}$  and the remaining  $e \in E^I\setminus E^I_{p}$ by $w_{min}-\epsilon+\vert V^{-p} \vert w_{max}$, where $w_{max}= \max_{e \in E^I_p} w_e^p$ $w_{min}= \min_{e \in E^I_p} w_e^p$ and $w_{min}-\epsilon>0$; with these new weights,  the corresponding  maximum-weighted matching $\bar{M}$  is always a matching of maximum cardinality plus it is the one that benefits the most player $p$.
		
		\begin{description}
			\item[-$\bar{M}$ is a maximum matching in the international graph.  ] A maximum matching $M$ has weight at least $\vert M \vert (w_{min}-\epsilon+\vert V^{-p} \vert  w_{max})$ . Suppose that $\bar{M}$ is not a maximum matching. The maximum-weighted matching  $\bar{M}$ has weight at most
			\begin{align*}
				\vert \bar{M} \vert (w_{max} + \vert V^{-p} \vert w_{max}) &= \vert \bar{M} \vert w_{max} (1+\vert V^{-p} \vert) \\
				& \leq (\vert \bar{M} \vert-1) w_{max} (1+\vert V^{-p} \vert)\\
				& = \vert \bar{M} \vert w_{max} (1+\vert V^{-p} \vert)- w_{max} (1+\vert V^{-p} \vert)\\
				& =  \vert \bar{M} \vert w_{max} \vert V^{-p} \vert + \vert \bar{M} \vert w_{max}  -w_{max} (1+\vert V^{-p} \vert)\\
				&= \vert \bar{M} \vert w_{max} \vert V^{-p} \vert +w_{max} (\vert \bar{M} \vert-1-\vert V^{-p} \vert)\\
				& <  \vert \bar{M} \vert w_{max} \vert V^{-p} \vert
			\end{align*}
			but $M$ has a higher weight, which contradicts the fact that $\bar{M}$ is a maximum-weighted matching for the international graph.
			\item[-$\bar{M}$ is the optimistic matching for player $p$.  ]  Otherwise, suppose that there is a maximum cardinality matching $M$ of the international graph that leads to a higher weighted matching for player $p$. Since both $M$ and $\bar{M} $ have maximum cardinality, then $\vert M \cap E^I_{-p} \vert <\vert \bar{M} \cap E^I_{-p} \vert$ and  $\vert M \cap E^I_{p} \vert >\vert \bar{M} \cap E^I_{p} \vert$. However, any edge in $E^I_{-p}$ has weight $w_{min} - \epsilon +\vert V^{-p} \vert$ which is strictly smaller  than the weight of any edge in $E^I_{p}$ (which is $w_e+\vert V^{-p} \vert$), implying that $\bar{M}$ would not have been optimal.
		\end{description}
		
		\begin{claim}
			{\BestReponseopt} is NP-hard.
		\end{claim}
		First, we reduce in polynomial time, an instance of {\threeDIM} to one of {\BestReponseopt}:
		\begin{itemize}
			\item Let $N = \lbrace A,B \rbrace$ with $p=A$.
			\item Let $T=\lbrace t_1,t_2, \ldots, t_{\vert T \vert} \rbrace$. Make $V^A=X \cup Y \cup Z\cup T \cup T^x \cup T^y \cup T^z$, where $T^{(\cdot)} = \lbrace t_1^{(\cdot)}, t_2^{(\cdot)}, \ldots, t_{\vert T \vert}^{(\cdot)} \rbrace$ (that is, we make 3 copies of each $t \in T$).  Make $V^B= T \cup T'$, where $T' =\lbrace t'_1,t'_2, \ldots, t'_{\vert T \vert} \rbrace$ (that is, a copy of $T$).
			\item Set $E^A = \lbrace (x_i,t_l^x), (y_j,t_l^y), (z_k,t_l^z):  t_l=(x_i,y_j,z_k) \in T\rbrace$ and
			
			$E^I= \lbrace (t^x,t), (t^y,t), (t^z,t), (t,t), (t,t'): t' \in T' \rbrace$. All edges have weight 1, except nodes $(t,t) \in E^I$ that have weight $2\vert T \vert$. See Figure \ref{fig:gadget_t}.
			\item Set $P=3n+2n \vert T \vert + 2(\vert T\vert -n)$.
		\end{itemize}
		
		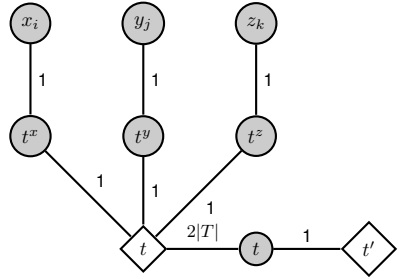
\begin{figure}[!h]
			\centering
			\begin{tikzpicture}[-,>=stealth',shorten >=0.3pt,auto,node distance=2cm,
				thick,countryA node/.style={circle,fill=black!20,draw},countryB node/.style={diamond,draw}, countryNone node/.style={diamond,fill=white},  scale=.75, transform shape]]
				\tikzstyle{matched} = [draw,line width=3pt,-]
				
				\node[countryA node] (1) {$x_i$};
				\node[countryA node] (2) [right of=1] {$y_j$};
				\node[countryA node] (3) [right of=2] {$z_k$};
				\node[countryA node] (4) [below of=1] {$t^x$};
				\node[countryA node] (5) [below of=2] {$t^y$};
				\node[countryA node] (6) [below of=3] {$t^z$};
				\node[countryB node] (7) [below of=5] {$t$};
				\node[countryA node] (8) [right of=7] {$t$};
				\node[countryB node] (9) [right of=8] {$t'$};
				
				\path[-,every node/.style={font=\sffamily\small}]
				( 1) edge node {1} (4)
				( 2) edge node {1} (5)
				( 3) edge node {1} (6)
				( 4) edge node {1} (7)
				( 5) edge node {1} (7)
				( 6) edge node {1} (7)
				( 7) edge node {$2\vert T \vert$} (8)
				( 8) edge node {1} (9);
			\end{tikzpicture}
			\caption{Gadget $t=(x_i,y_j,z_k) \in T$. Gray vertices belong to player A and white vertices belong to player B.}
			\label{fig:gadget_t}
		\end{figure}
		We claim that in the constructed instance of {\BestReponseopt}, player A can have a profit of at least $P=3n+2n \vert T \vert + 2(\vert T\vert -n)$ if and only if the {\threeDIM} instance has answer YES.
		
		\paragraph{Proof of if}   Assume the {\threeDIM} instance has  answer YES. Then, there is a 3-dimensional matching $M \subset T$ of size $n$. Select the internal matching $M^A= \lbrace (x_i,t_l^x), (y_j,t_l^y), (z_k,t_l^z): t_l \in M \rbrace$. Note that $M^A$ is a matching since w.l.o.g. if $(x_i,t_l^x), (x_i,t_r^x) \in M^A$, then $t_l$ and $t_r$ are in $M$, meaning that $M$ was not a valid 3-dimensional matching. This gives an internal profit for player A equal to $3n$ . For all gadgets $t$, the optimistic {\IndA} selects the edge $(t,t) \in E^I$ leading to a total weight of $2n \vert T \vert$. Finally, for all the gadgets $t \notin M$, the {\IndA} has to select a matching of cardinality 2 which can be $(t^x,t), (t,t')$, leading to a total weight of $2(\vert T\vert -n)$. In this way, player A's profit is equal to $3n + 2n\vert T \vert + 2(\vert T\vert -n) = P.$

		\paragraph{Proof of only if} Now, assume that the {\BestReponseopt} is a YES instance. Then, there is $M^A$ such that player A's profit is greater or equal to $P = 3n + 2n \vert T \vert+ 2(\vert T\vert -n)$. If this internal matching does not match $t^x, t^y$ and $t^z$ for any gadget $t$, the optimistic {\IndA} has always a matching of size two available in each gadget, leading to a profit of at most $4 \vert T \vert$ that  is less than $P$.   In this way, player A must match internally as many as possible of the vertices $t^x, t^y$ and $t^z$ for each gadget $t$ so that the {\IndA} does not have available the external matching of size 2 and it is forced to choose $(t,t)$ with weight $2\vert T \vert$. Whenever this is done, player A adds $3+2\vert T \vert$ to her profit. If only $k<n$ gadgets have all $t^x, t^y$ and $t^z$ matched internally, the utility is at most  $3k+2k\vert T \vert + 4(\vert T\vert -k )$. If $k=n-1$, we have utility equal to $3(n-1)+2(n-1)\vert T \vert + 4(\vert T\vert -n+1) = 3n+2n \vert T \vert +2 (\vert T \vert -n)+\overbrace{1-n}^{< 0},$ which is smaller than $P$.
		
	\end{proof}

	\section{Verifying Equilibria}\label{app:Best_responses}
	
	Algorithm~\ref{Alg:Best_weighted_Optimistic_response} enables to verify if there is an algorithm $\mathcal{A}$ such that a player $p$ has no incentive to deviate from a matching $M$: player $p$ utility under $M$ is a lower bound $LB$ for her best response to $M^{-p}$; player $p$'s opponents strategies are fixed to $M^{-p}$; player $p$ optimal solution when she controls the {\IndA} output and does not repeat previous evaluated solutions is determined leading to an upper bound $UB$; if the $UB$ is less or equal to $LB$ the player $p$ has no incentive to deviate and the algorithm returns $True$; otherwise, the {\IndA} decision for new player $p$ strategy is determined and it is evaluated if it improves player $p$ solutions (\ie, if it is greater than $LB$); if yes, the algorithm returns $False$, otherwise, it computes a new strategy for player $p$.
	
	Algorithm~\ref{Alg:Pessimistic_IA} computes the {\IndA} solution that minimizes the benefits for player $p$.
	
	\renewcommand{\baselinestretch}{1}
	\begin{algorithm}[hbtp]
		\caption{Return True if player $p$ has no incentive to deviate from given $(M^p,M^{-p})$, and False otherwise.}
		\begin{algorithmic}[1]
			\STATE $x^{p,0} \leftarrow$ binary vector of iteration 0 with one entry for each edge in $E^p$, where an entry is 1 if the associated edge is in $M^p$ and 0, otherwise.   {\color{gray}\# Representation of player $p$'s strategy } \\ $k \leftarrow 1$ \\ $aux = True$ {\color{gray}\#  represents the non-incentive to deviate} \\ $LB \leftarrow \sum_{e \in M^p} w_e^p + \sum_{e \in M \cap E^I_p} w_e^p$ {\color{gray}\# Player $p$ utility under $M$ }
			\FOR{$j \in N \setminus \lbrace p \rbrace$}
			\STATE $x^j \leftarrow$ binary  vector with one entry for each edge in $E^j$, where an entry is 1 if the associated edge is in $M^j$ and 0, otherwise.   {\color{gray}\# Representation of player $j$'s strategy }
			\ENDFOR
			\WHILE{aux}
			\STATE Solve
			\begin{subequations}
				\begin{alignat}{5}
					(x^{p,k},\bar{y}) = &\argmax_{x,y \in \{ 0,1\}}  &&   \sum_{e \in E^p} w_e x_e^p+ \sum_{e \in E^I_p} w_e^p y_e \\[0.4ex]
					&\mbox{s.t.~~}  &&\sum_{e \in E^p: i \in e} x^p_e+\sum_{e \in E^I_p: i \in e} y_e \leq 1 \quad  \forall i \in V^{p}     \\
					&&& \sum_{e \in E^I_p: i \in e} y_e \leq 1-\sum_{j=1}^N\sum_{e \in E^j: i\in e} x_e^j \quad  \forall i \in V^{-p}\\
					&&& \sum_{e \in E^p: x_e^{p,j}=0} x^p_e+\sum_{e \in E^p: x_e^{p,j}=1}1- x^p_e\geq 1 \quad j=0,\ldots,k-1,
				\end{alignat}
				\label{BestResponseNogood}
			\end{subequations}
			where the solutions, $x^{p,k}$  and $\bar{y}$ represent a matching in $E^p$ and  in $E^I$, respectively.
			Let $UB$ be the optimal objective value of Problem~\eqref{BestResponseNogood}.
			\IF{Optimization~\eqref{BestResponseNogood} infeasible or $UB\leq LB$}
			\RETURN aux {\color{gray}\# which is True, \ie, no incentive to deviate}
			\ENDIF
			\STATE $y^k \leftarrow $ Algorithm \ref{Alg:Pessimistic_IA} for $x^{p,k}$
			\IF{$\sum_{e \in E^p} w_e x_e^{p,k} + \sum_{e \in E^I_p} w_e^p y_e^k > LB$}
			\RETURN False	{\color{gray}\# \ie, incentive to deviate}
			\ENDIF
			\STATE $k \leftarrow k+1$
			\ENDWHILE
		\end{algorithmic}
		\label{Alg:Best_weighted_Optimistic_response}
	\end{algorithm}


	\begin{algorithm}[hbtp]
		\caption{Compute an {\IndA} decision that minimizes player $p$ utility for $(x^{-p},x^{p,k})$.}
		\begin{algorithmic}[1]
			\STATE Let $OPT^{k}$ be the optimal value to the {\IndA} given $(x^{-p},x^{p,k})$.
			\STATE Solve  {\color{gray}\# minimize player $p$ utility}
			\begin{subequations}
				\begin{alignat}{5}
					y^k = &\arg\min_{y \in \{ 0,1\}}  &&  \ \ \sum_{e \in E^I_p} w_e^p y_e \\[0.4ex]
					&\mbox{s.t.~~}  &&\sum_{e \in E^I_p: i \in e} y_e \leq 1-\sum_{e \in E^p: i\in e} x_e^{p,k} \quad  \forall i \in V^{p}     \\
					&&& \sum_{e \in E^I: i \in e} y_e \leq 1-\sum_{j=1}^N\sum_{e \in E^j: i\in e} x_e^j \quad  \forall i \in V^{-p}\\
					&&& \sum_{e \in E^I} w_e y_e \geq OPT^k \textrm{\color{gray} \# keep {\IndA} solution optimal}
				\end{alignat}
			\end{subequations}
			\RETURN $y^k$
		\end{algorithmic}
		\label{Alg:Pessimistic_IA}
	\end{algorithm}
	
	\renewcommand{\baselinestretch}{1.5}
	\section{Uniform generation of maximum matchings}\label{app:uniform_gen}
	
	In \cite{Dyer:2003} a dynamic programming method for counting the number of optimal solutions for a multiple-knapsack is presented; furthermore, such computation is then used to uniformly generate optimal solutions. We adapt this paper idea to  our matching polytope
	\begin{subequations}
		\begin{alignat}{4}
			&    &
			\sum_{e \in E: i \in e} x_e \leq 1 & \hspace*{2cm} \forall i \in V\\
			& &  \sum_{e \in E}  x_e = OPT \\
			& & x_e \in \lbrace 0,1 \rbrace & \hspace*{2cm} \forall e \in E, \label{eq:2}
		\end{alignat}
		\label{IntialModelPoly}
	\end{subequations}
	where $OPT$ is the size of a maximum cardinality matching.
	
	Our strategy for the uniform generation of maximum matchings will be based on counting them as we describe next.
	
	Let $F(K, r, s) $ for a graph $G=(V,E)$ be the number of matchings with cardinality $K$ when the subset of edges $ \lbrace e_1,e_2, \ldots, e_r \rbrace \subseteq E= \lbrace e_1,e_2, \ldots, e_m \rbrace$  and the subset of vertices $\lbrace v \in V: s_v = 1 \rbrace$ are available ($s$  is a binary vector of size $|V|$ that represents the vertices available). We aim to determine $F(OPT, m, s') $ when $s'=\left[1,1, \ldots, 1 \right]$, \ie,  the number of matchings with maximum cardinality. This can be done through the dynamic programming recursion
	{\footnotesize
		$$F(K,r,s)= \left\{
		\begin{array}{ll}
			0, &  K<0 \vee r <0 \vee \min(s)<0 \quad \tiny  \textit{(the entries of F must be non negative)}\\
			0, & K>r  \quad \tiny \textit{(there is no matching}  )\\
			1, & K=0 \wedge   r \geq 0\wedge \min(s) \geq 0  \quad \tiny \textit{(there is one matching:  $\emptyset$)}  \\
			0, & K>0 \wedge   r = 0\wedge \min(s) \geq 0  \quad \tiny \textit{(there is no matching}  )\\
			F(K,r-1,s)+ & \\
			F(K-1,r-1,\hat{s}) & K>0 \wedge   r > 0\wedge \min(s) \geq 0  \quad \tiny \textit{(number of matchings not using  $e_r$}\\
			& \tiny  \textit{ plus number of matchings using  $e_r$  )}
		\end{array}
		\right.  $$
	}
	where $\hat{s}_i=s_i$, except for $i \in e_r$ where $\hat{s}_i= s_i-1$, \ie, if edge $e_r$ is part of a matching then the incident vertices become matched and thus, unavailable for the subsequent recursions.
	The computational time of this approach is $O(|E|^2 2^{|V|})$.
	
	The $F(\cdot)$ table can be used to determine a uniform point in \eqref{IntialModelPoly}, by tracing back probabilistically from $F(OPT, m, s) $, as follows: with probability $\frac{F(OPT, m-1, s')}{F(OPT, m, s')} $ set $x_{e_m}=0$, else set $x_{e_m}=1$; if $x_{e_m}=0$, recursively determine $x_{e_{m-1}}$, $\ldots$, $x_{e_1}$ by tracing back from $F(OPT, m-1, s')$ and if $x_{e_m}=1$, trace back similarly from $F(OPT-1, m-1, \hat{s}')$. The resulting matching has probability $\frac{1}{F(OPT, m, s')}$ of being sampled, \ie, it is uniformly generated.
	
	We remark that the problem of randomly generating a matching of maximum cardinality is \#P-hard~\cite{VALIANT1979}. Indeed, this is a very difficult problem in practice, restricting the scalability of the approach described here.
	
\end{appendix}

\bibliographystyle{abbrv}
\bibliography{ref_revision}
\end{document}